\documentclass{ecai} 

\usepackage{latexsym}
\usepackage{amssymb}
\usepackage{amsmath}
\usepackage{amsthm}
\usepackage{amsfonts}
\usepackage{booktabs}
\usepackage{enumitem}
\usepackage{graphicx}
\usepackage{color}
\usepackage{tcolorbox}
\usepackage{nicefrac}
\usepackage{hyperref}
\hypersetup{
    colorlinks=false,
    pdfborder={0 0 0},
    pdfborderstyle={/S/U/W 0},
}
\usepackage[capitalise]{cleveref}
\usepackage{subfig}
\usepackage{floatrow}

\usepackage{siunitx}
\newtheorem{theorem}{Theorem}[section]
\newtheorem{lemma}[theorem]{Lemma}
\newtheorem{corollary}[theorem]{Corollary}
\newtheorem{proposition}[theorem]{Proposition}

\newtheorem{claim}[theorem]{Claim}

\newenvironment{claimproof}
{
    
    \proof
}
{
    \endproof
    
}

\newenvironment{proofsketch}{%
  \proof}{\endproof}

\newcommand{\BibTeX}{B\kern-.05em{\sc i\kern-.025em b}\kern-.08em\TeX}
\graphicspath{{./0a_Figures/}}

\newcommand{\longRCE}{\textsc{Resilient Committee Elections}\xspace}

\newcommand{\RCE}{\textsc{RCE}\xspace}

\newcommand{\CI}{\textsc{CI}\xspace}

\newcommand{\CCAV}{CCAV\xspace}

\newcommand{\GreedyCC}{Greedy-CC\xspace}
\newcommand{\GreedyPAV}{Greedy-PAV\xspace}

\newcommand{\Pclass}{\ensuremath{\mathsf{P}}\xspace}

\newcommand{\NP}{\ensuremath{\mathsf{NP}}\xspace}
\newcommand{\coNP}{\ensuremath{\mathsf{coNP}}\xspace}

\newcommand{\FPT}{\ensuremath{\mathsf{FPT}}\xspace}
\newcommand{\XP}{\ensuremath{\mathsf{XP}}\xspace}

\newcommand{\Wone}{\ensuremath{\mathsf{W[1]}}\xspace}
\newcommand{\coWone}{\ensuremath{\mathsf{coW[1]}}\xspace}

\definecolor{diy-gray}{gray}{0.4}
\definecolor{diy-brown}{HTML}{792500}
\definecolor{calpolypomonagreen}{rgb}{0.12, 0.3, 0.17}
\definecolor{dartmouthgreen}{rgb}{0.05, 0.5, 0.06}
\definecolor{forestgreen}{rgb}{0.0, 0.27, 0.13}
\definecolor{vviolet}{rgb}{0.6, 0.3, 0.75}

\newcommand{\ccomment}[1]{}

\newcommand{\Cabove}{\ensuremath{C_{\text{above}}}\xspace}
\newcommand{\Cequal}{\ensuremath{C_{\text{equal}}}\xspace}
\newcommand{\Cbelow}{\ensuremath{C_{\text{below}}}\xspace}

\newcommand{\Add}{\ensuremath{\textsc{Add}}\xspace}
\newcommand{\Remove}{\ensuremath{\textsc{Remove}}\xspace}

\usepackage{csquotes}

\usepackage{xspace} % spaces after macros

\usepackage{xargs} % gives access to command 'newcommandx' with predefined entries

\newcommand{\NN}{\mathbb{N}\xspace} % natural numbers
\usepackage{mathtools}

\newcommand{\RRR}{\ensuremath{\mathcal{R}}\xspace} % voting rule
\newcommand{\GRRR}{\ensuremath{\text{Greedy-}\mathcal{R}}\xspace} % greedy voting rule
\newcommand{\problemname}[1]{\textsc{#1}}

\newcommand{\bigbrace}[1]{\ensuremath{\left( #1 \right)}\xspace}

\newcommand{\app}{\ensuremath{\mathrm{app}}\xspace}
\newcommand{\degg}{\ensuremath{\mathrm{deg}}\xspace}
\newcommand{\dist}{\ensuremath{\mathrm{dist}}\xspace}
\newcommand{\Dist}{\ensuremath{\mathrm{Dist}}\xspace}

\newcommand{\lamscore}{\ensuremath{\lambda\text{-score}}\xspace}

\newcommand{\lscoreE}[2]{\ensuremath{\text{score}(#2, #1)}\xspace}

\newcommand{\ceil}[1]{\ensuremath{\left \lceil #1 \right \rceil}\xspace}
\newcommand{\floor}[1]{\ensuremath{\left \lfloor #1 \right \rfloor}\xspace}

\newcommand{\bbigO}[1]{\mathcal{O}\bigbrace{#1}\xspace} % Landau symbol

\newcommand{\IS}{\problemname{IS}\xspace}

\newcommand{\ISWFV}{\problemname{ISWFV}\xspace}

\begin{document}

%%%%%%%%%%%%%%%%%%%%%%%%%%%%%%%%%%%%%%%%%%%%%%%%%%%%%%%%%%%%%%%%%%%%%%%%

\begin{frontmatter}

%%% Use this command to specify your submission number.
%%% In doubleblind mode, it will be printed on the first page.

\paperid{305} 

%%% Use this command to specify the title of your paper.

\title{Multiwinner Temporal Voting with Aversion to Change}

\author[A]{\fnms{Valentin}~\snm{Zech}\thanks{Corresponding Author. Email: zech@vzech.de.}}
\author[B]{\fnms{Niclas}~\snm{Boehmer}}
\author[A,C]{\fnms{Edith}~\snm{Elkind}} 
\author[A]{\fnms{Nicholas}~\snm{Teh}} 

\address[A]{University of Oxford, UK}
\address[B]{Hasso Plattner Institute, University of Potsdam, Germany}
\address[C]{Alan Turing Institute, UK}

%%% Use this environment to include an abstract of your paper.

\begin{abstract}
We study two-stage committee elections where voters have dynamic preferences over candidates; at each stage, a committee is chosen under a given voting rule. We are interested in identifying a winning committee for the second stage that overlaps as much as possible with the first-stage committee. We show a full complexity dichotomy for the class of Thiele rules: this problem is tractable for Approval Voting (AV) and hard for all other Thiele rules (including, in particular, Proportional Approval Voting  and the Chamberlin--Courant rule). We extend this dichotomy to the greedy variants of Thiele rules. We also explore this problem from a parameterized complexity perspective for several natural parameters. We complement the theory with experimental analysis: e.g., we investigate the average number of changes in the committee as a function of changes in voters' preferences and the role of ties.
\end{abstract}

\end{frontmatter}

%%%%%%%%%%%%%%%%%%%%%%%%%%%%%%%%%%%%%%%%%%%%%%%%%%%%%%%%%%%%%%%%%%%%%%%%

\section{Introduction}
A local town council advisory committee is responsible for making decisions on various community issues such as education, infrastructure, and public services. Elections are held biennially to fill the positions on this committee.
Numerous residents, each with their own platform and priorities for the town's development, step forward as candidates for the election to this advisory committee. Voters from different neighborhoods and demographics then go to the polls to elect members of this committee, to make decisions on their behalf.

Between election cycles, due to varying campaign performances and evolving community concerns, some voters change their preferences over the candidates.
%the town experiences demographic shifts and evolving community concerns, resulting in voters changing their preferences over the candidates. %due to factors such as campaign performances, candidate debates, and emerging community concerns.
While the voting rule to be used to select the advisory committee is fixed by the bylaws, it often results in multiple tied committees, i.e., it does not fully determine
the election outcome. There are many ways to break these ties; in particular, one may want to maintain contiguity by prioritizing committees that have a substantial overlap with the previous committee, so as to build on the existing expertise and maintain stability, while remaining representative of the population's preferences.

This problem can be viewed through the lens of \emph{multiwinner temporal voting}, a natural temporal extension of the well-studied multiwinner voting model (see the taxonomy of \citet{boehmer_broadening_2021} and a survey by \citet{elkind2024temporal}).
In multiwinner temporal elections, a set of voters have dynamic preferences over a set of candidates, and we want to elect a committee at each timestep.
In this work, we seek to study, given a voting rule, how winning committees adapt to changes in the voters' preferences in situations where it is undesirable to replace many committee members at once. 
%This measure of how much the winning committee changes given augmented preferences in the election can then be denoted as the \emph{resilience} of a rule.
%we seek to study the \emph{resilience} of rules with desirable guarantees, i.e., given agents' evolving preferences, we want a rule with desirable qualities, and which minimizes the changes made to each consecutive winning committee.

\paragraph{Our Contributions}
We consider two-stage elections that use a fixed voting rule to select a winning committee at each stage; the voters have approval preferences that may evolve between the stages.
When the first-stage committee is elected, the second-stage preferences are not yet known. Hence, a winning committee is chosen arbitrarily from the committees tied for winning. In the second stage, the goal is to identify a committee that wins under the new preferences, yet has as much overlap as possible with the first-stage committee. 

We consider this problem for the well-known class of Thiele rules~\cite{lackner_multiwinner_2023} (this class includes, in particular, Approval Voting (AV)~\cite{brams_1978_approval}, Proportional Approval Voting (PAV) \cite{thiele_om_1895}, and the Chamberlin--Courant rule (CCAV)~\cite{chamberlin_representative_1983}) and their greedy variants.
In \cref{sec:thiele}, we present a full complexity classification of our decision problem for Thiele rules.
In particular, we obtain a dichotomy: 
the decision version of our problem is in \Pclass for Approval Voting (AV) (\cref{sec:committee_scoring_rules}), and \coNP-hard for all Thiele rules other than AV (\cref{sec:ICE_udT} and \cref{sec:beyond}). 
We then extend this dichotomy to all \emph{greedy} Thiele rules (\cref{sec:ICE_gudT}); surprisingly, our problem remains hard for all rules other than AV, even though under greedy Thiele rules, computing a single winning committee is computationally tractable.
To complement our hardness results, we provide parameterized complexity results for a selection of natural parameters and their combinations (\cref{sec:parameterized}).
Finally, we use experiments to obtain quantitative results (\cref{sec:experiments}): we measure the amount of change in the committees as a function of change in voters' preferences, and investigate the role of ties. In particular, our experiments show that simply breaking ties lexicographically is far from optimal with respect to the contiguity objective, thereby justifying our theoretical analysis.

All missing proofs and additional experimental results can be found in the appendix.

%%%%%%%%%%%%%%%%%%%%%%%%%%%%%%%%%%%%%%%%%%%%%%%%%
\paragraph{Related Work}
Our analysis extends the line of work initiated by \citet{bredereck_when_2022}, but our model and contributions differ 
from theirs in several key aspects: \citet{bredereck_when_2022}
study a model where (i) committees are not selected from the set of winning committees output by a voting rule, but need to satisfy a lower bound on their score (via a \emph{single non-transferable vote} committee scoring rule), (ii) the subsequent committees need not be of the same size, and (iii) voters preferences in all stages are known at the start.
In contrast, we consider a large class of popular voting rules and require each selected committee to win under a given rule for a fixed committee size; moreover, our model is more realistic in that it does not assume that the second-stage preferences are known initially.
%for our model, we select committees from the set of winning committees output by a voting rule (and we study this for a large class of rules), and we require committees to be of the same fixed size. 

Other papers in this vein include the work of
\citet{bredereck_electing_2020}, that considered maximizing the changes made to a committee (the \enquote{revolutionary} setting) to find diverse committees, and the paper by
\citet{deltl2023seqcommittee}, which looks into treating agents fairly.

Our work is also related to the study of \emph{robustness} in temporal voting \cite{bredereck_robustness_2021}.
\citet{bredereck_robustness_2021} conducted an axiomatic study of several voting rules, focusing on the worst-case change that may have to be made to a winning committee after a single change was made to some agent's preference. 
Conversely, for a given election and a voting rule, they ask for the minimum number of changes to the agents' preferences so that the election outcome changes (in any way) under the given rule. This framework was subsequently adapted and studied in approval-based elections \cite{gawron_robustness_2019}, for greedy approval-based rules \cite{faliszewski_robustness_2022}, and for nearly-structured preferences \cite{misra_robustness_2019}.

Our work is different in that we go beyond this worst-case measure and conduct a more fine-grained analysis, where we ask \emph{by how much} a winning committee needs to change in a given altered election and study the related computational problems. Moreover, \citet{bredereck_robustness_2021} and \citet{faliszewski_robustness_2022} assume a fixed tie-breaking order of candidates, whereas we study rules under parallel-universe tie-breaking.
%
%In particular, \citet{bredereck_robustness_2021} and \cite{gawron_robustness_2019} showed that robustness level for the rules they study is either $1$ (so that a single small change to the voters’preferences leads to replacing at most one candidate)or $k$ (so that a singlechange can lead to replacing the whole committee).\citet{misra_robustness_2019} showed that the problem remained hard for nearly-structured preferences.  and \citet{faliszewski_robustness_2022}In contrast, our work allows for a more fine-grained analysis
%This problem is different from ours in that they study how \emph{small a change} to agents' preferences is necessary so as to observe any change in the output of the voting rule. Our model, in contrast, allows for a more fine-grained analysis of robustness, where we ask \emph{by how much} does a winning committee need to change in an altered election.
%This robustness concept was subsequently adapted and studied in approval-based elections \cite{gawron_robustness_2019}, for greedy approval-based rules  \cite{faliszewski_robustness_2022}, and in restricted domains \cite{misra_robustness_2019}.\footnote{Moreover, \citet{bredereck_robustness_2021} and \citet{faliszewski_robustness_2022} ignored ties (i.e., assume a fixed tie-breaking order of candidates) and focused on \emph{resolute} voting rules, whereas we study rules under parallel-universe tie-breaking.}
%
Further, while most of the prior work focuses on a few 
easy-to-compute rules, 
in our work, we focus on the important class of Thiele rules
and their greedy variants.

There are also models of temporal voting where a single winner is chosen at each timestep.
These works generally consider temporal extensions of popular multiwinner voting rules and study various axiomatic properties \cite{lackner_perpetual_2020,lackner2023proportionalPV}, welfare measures \cite{elkind2024temporalelections,Neoh_Teh_2024}, or extensions of the \emph{justified representation} axioms and its variants \cite{bulteau_justified_2021,do22,chandak23proportional,elkind2024verifying}.
Other works look into the special case of fair scheduling, where preferences and the outcome are permutations of candidates \cite{elkind2022temporalslot}.

We note that similar problems have also been studied (albeit under different names) in the context of stable matching \cite{bredereck2020stableevolving}, coloring \cite{hartung2013coloring}, clustering \cite{charikar2004incremental,luo2021cluster}, and reoptimization \cite{bockenhauer2022reoptimization,schieber2018reoptimization}. 
%To the best of our knowledge, this problem is novel in the context of voting. 

%%%%%%%%%%%%%%%%%%%%%%%%%%%%%%%%%%%%%%%%%%%

\section{Preliminaries}
Given a logical expression $\varphi$, we use the Iverson bracket notation $[\varphi]$ to denote the evaluation of that expression: $[\varphi]=1$ if $\varphi$ is true and $[\varphi]=0$ otherwise.
We assume familiarity with the basics of classic complexity theory~\cite{papadimitriou_computational_2007} 
and parameterized complexity \cite{flum_parameterized_2006,niedermeier_invitation_2006}.

% Election
\paragraph{Elections} 
Let $C = \{c_1,\dots, c_m\}$ be a set of $m$ candidates, and let $N=\{1, \dots, n\}$
be a set of $n$ voters.
Each voter $i\in N$ has \emph{approval preferences} over candidates in $C$, captured by a {\em ballot}
$v_i \subseteq C$; we require that $v_i\neq\emptyset$. Let $V = (v_1,\dots,v_n)$. We refer to the pair $(C, V)$ as an \emph{election}.  
%We consider two kinds of preferences: \emph{ordinal}, where each voter $i \in N$ has a linear order $\succ_i$ over $C$, and \emph{approval}, where each voter $i \in N$ has an approval set $C_i \subseteq C$. For each $i\in N$, we write $v_i$ to denote the ballot of voter $i$ (which is a linear order in case of ordinal preferences and a subset of $C$ in case of approval preferences), and let $V = (v_1,\dots,v_n)$. We refer to the pair $(C, V)$ as an \emph{election}. 

% Voting Rule
\paragraph{Rules}
A \emph{multiwinner voting rule} $\RRR$ maps an election $E = (C,V)$ and an integer $k \in [|C|]$ to a non-empty family of sets $\RRR(E,k)$, where each set in $\RRR(E, k)$ is a size-$k$ subset of $C$. The sets in $\RRR(E, k)$ are called the {\em winning committees} for $E$ under $\RRR$.
%All members of $\RRR(E,k)$ are said to be \emph{tied} for winning in election $E$ under voting rule $\RRR$.

In this work, we focus on 
a well-studied class of voting rules known as Thiele rules \cite{kilgour_approval_2010,lackner_multiwinner_2023}, and their greedy variants.
We say that a function $\lambda : \NN^+ \rightarrow [0, 1]$ is an \emph{Ordered Weighted Averaging (OWA) function}
if $\lambda(1) = 1$ and $\lambda$ is non-increasing, i.e.\ for all $i, j \in \NN^+$ it holds that $i > j$ implies $\lambda(i) \leq \lambda(j)$ \cite{faliszewski_robustness_2022}. In what follows, we assume that all OWA functions $\lambda$ we consider take values in ${\mathbb Q}\cap [0, 1]$ and are polynomial-time computable.
Each OWA function $\lambda$
defines a \emph{Thiele rule} $\RRR_\lambda$ as follows.
%\footnote{An alternative definition of these rules use a set of score vectors  \cite{faliszewski_multiwinner_2017}.}
Given an election $E = (C, V)$ and a committee size $k \in [|C|]$, the {\em \lamscore} of a candidate set $S \subseteq C$ in $E$ is defined by: 
$
    \lamscore^E(S) = \sum_{i \in N} \bigbrace{\sum_{j = 1}^{|S \cap v_i|} \lambda(j)}.
$
The rule $\RRR_\lambda$ then outputs all size-$k$ subsets of $C$ with the maximum \lamscore: $\RRR_\lambda(E, k) = \mathrm{argmax}\{\lamscore(S)\mid S \in \binom{C}{k}\}$.
We omit $\lambda$ and/or $E$ whenever it is clear from the context.
This framework captures many well-known voting rules. Specifically, 
\begin{itemize}
    \item $\lambda(i) = 1$ for all $i \in \NN^+$ corresponds to the Approval Voting (AV) rule \cite{brams_1978_approval};
    \item $\lambda(i) = \nicefrac{1}{i}$ for all $i \in \NN^+$ corresponds to the Proportional Approval Voting (PAV) rule \cite{thiele_om_1895}; and
    \item $\lambda(i) = [i = 1]$ for all $i \in \NN^+$ corresponds to the Chamberlin--Courant Approval Voting (CCAV) rule \cite{chamberlin_representative_1983}.
\end{itemize}
%The simplest example of a Thiele rule is the AV voting rule, which is defined by the constant function $\lambda_{\mathrm{AV}}(i) = 1$ for all $i \in \NN^+$. The AV rule is most appropriate for situations in which , while voter representation is of subordinate priority. 
The AV rule is appropriate when the aim is to select a group of individually excellent candidates, whereas CCAV aims to represent as many voters as possible; the PAV rule provides strong proportional representation guarantees~\cite{lackner_multiwinner_2023}.
We say that an OWA function $\lambda$ (and the associated Thiele rule $\RRR_\lambda$) is {\em unit-decreasing} if 
$\lambda(1) = 1 > \lambda(2)$. AV is not unit-decreasing, whereas PAV and CCAV are. Intuitively, unit-decreasing rules capture that the voters' marginal utility from having an additional representative in the winning committee is lower than their utility from being represented~at~all.
%Prominent examples of unit-decreasing Thiele rules include Proportional Approval Voting (PAV) \cite{thiele_om_1895}, induced by the function $\lambda_{\mathrm{PAV}}(i) = \nicefrac{1}{i}$ and Chamberlin--Courant (CCAV) \cite{chamberlin_representative_1983}, induced by the function $\lambda_{\mathrm{CCAV}}(i) = [i = 1]$.

Since committees under both PAV and CCAV are \NP-hard to compute \cite{aziz_computational_2015,procaccia_complexity_2008,yang_parameterized_2018}, there exist greedy approximation variants of these rules, which are based on the notion of {\em marginal contributions}.
Given an election $E = (C, V)$, a subset $S \subseteq C$ of candidates  and an OWA function $\lambda$, we define the \emph{marginal contribution} (or \emph{points}) of a candidate $c \in C\setminus S$ with respect to $S$ as 
%$\lamscore(S) - \lamscore(S \setminus \{c\}$ if $c \in S$ and 
$\lamscore(S \cup \{c\}) - \lamscore(S)$. 
%otherwise. %In other words, the marginal contribution of a candidate describes the value they add to a subset of candidates they are a part of or the value they would add to a subset in case they were added to it. For readability, when the subset of candidates is clear from the context, we will sometimes refer to the marginal contribution of a candidate simply as the candidate's points.
Then, the greedy variant of a Thiele rule $\RRR_\lambda$, which we denote by $\GRRR_\lambda$, outputs all committees that can be obtained by the following iterative procedure: start with an empty committee $S$ and perform $k$ iterations; at each iteration, add a candidate in $C\setminus S$ with maximum marginal contribution to $S$ under $\lambda$ (at each iteration there may be multiple candidates with maximum marginal contribution to the current committee; a committee is in the output of the rule if it can be obtained by some way of breaking these ties at each iteration). Clearly, a committee in the output of $\GRRR_\lambda$ can be computed in polynomial time (recall that we assume that $\lambda$ itself is polynomial-time computable).
We note that Greedy-AV is equivalent to AV; however, other Thiele rules differ from their greedy variants.

\paragraph{Distances} 
Given two committees $S,S' \subseteq C$ of equal size $k$, we define the \emph{distance} between $S$ and $S'$ as $\dist(S,S') = k - |S \cap S'|$.

To this end, we define elementary \Add and \Remove operations on elections. An \Add operation adds a previously unapproved candidate to the ballot of a single voter. 
A \Remove operation removes a single candidate from a single voter's ballot.
Then, the distance between two elections $E$ and $E'$, denoted $\Dist(E,E')$, is defined as the length of the shortest sequence of \Add and \Remove operations that transforms $E$ into $E'$ (and $+\infty$ if $E$ cannot be transformed into $E'$ using these two operations).

\paragraph{Decision Problem}
We are now ready to present the family of decision problems we are interested in. 
%Elements of this family correspond to multiwinner voting rules.
\begin{tcolorbox}
\longRCE ($\RRR$-\RCE):\\
\textbf{Input}: Elections $E = (C, V)$ and $E' = (C, V')$ 
over the same set of candidates $C$, 
a committee size $k \in \NN$, 
a winning committee $S \in \RRR(E, k)$, 
and a distance bound $\ell \in \NN$.\\
\textbf{Question}: Does there exist a committee $S' \in \RRR(E', k)$ such that $\dist(S, S') \leq \ell$?
\end{tcolorbox}

%We investigate the problem both for ordinal-ballot and approval-based elections. However, we will restrict focus to scenarios where election $E'$ can be obtained from election $E$ by only performing \Swap operations in the case of ordinal-ballots, or \Add and \Remove operations in the case of approval-based elections. 

%We require that $\Dist_{\{\Swap\}}(E, E') \in \NN$ for ordinal-ballot elections and $\Dist_{\{\Add, \Remove\}}(E, E') \in \NN$ for approval-based elections.

\section{A Dichotomy for Thiele Rules}\label{sec:thiele}
We present a full complexity classification of \RCE for Thiele rules. \RCE 
is tractable for AV and hard for all other Thiele rules. To show this, we proceed in three steps: first (\Cref{sec:committee_scoring_rules}) we present a polynomial-time algorithm for AV, then (\Cref{sec:ICE_udT}) we give a hardness proof for all unit-decreasing Thiele rules, and finally (\Cref{sec:beyond}) we extend it to all Thiele rules other than AV. 
Our hardness result for unit-decreasing Thiele rules
also establishes that this problem is \coWone-hard with respect to the committee size $k$.

%%%%%%%%%%%%%%%%%%%%%%%%%%%%%%%%
\subsection{Tractability for Approval Voting} \label{sec:committee_scoring_rules}
We first observe that \RCE is easy for Approval Voting.
\begin{proposition}\label{prop:av-poly}
    \text{AV}-\RCE admits a polynomial-time algorithm.
\end{proposition}
\begin{proof}
    Consider two elections $E = (C, V)$ and $E' = (C, V')$ over a candidate set $C = \{c_1, \ldots, c_m\}$. 
    %where $V = (v_1, \ldots, v_n)$ and $V' = (v'_1, \ldots, v'_{n})$. 
    Let $k \in \NN$ be the committee size and let $S \in \RRR(E, k)$ be a winning committee for $E$. 
    For each $c\in C$, let $s(c)$ be the approval score of $c$ in $E'$. Without loss of generality, assume that $s(c_1) \geq \cdots \geq s(c_m)$. 
    Then, we partition $C$ into three disjoint sets: $\Cabove = \{c_j \mid s(c_j) > s(c_k)\}$, $\Cequal = \{c_j \mid s(c_j) = s(c_k)\}$, and $\Cbelow = \{c_j \mid s(c_j) < s(c_k)\}$.
    %\begin{align*}
    %    \Cabove & = \{c_j \mid s(c_j) > s(c_k)\}, \\ 
    %    \Cequal & = \{c_j \mid s(c_j) = s(c_k)\}, \\
    %    \Cbelow & = \{c_j \mid s(c_j) < s(c_k)\}.
    %\end{align*}
    A winning committee under AV in $E'$ must include all of the candidates in \Cabove and $k - |\Cabove|$ candidates from \Cequal; by construction, $0< k-|\Cabove|\le |\Cequal|$.
    Thus, we construct a committee $S^*$ by first including all candidates from \Cabove as well as $\min\{k - |\Cabove|, |\Cequal\cap S|\}$ candidates from $\Cequal \cap S$. We then fill the committee with arbitrary candidates from \Cequal. Obviously, the resulting committee $S^*$ wins in election $E'$ under AV. Furthermore, our approach ensures that $S^*$ contains as many candidates from $S$ as possible. Finally, we check if $\dist(S, S^*) \leq \ell$. As all steps are computable in polynomial time, this concludes the proof. 
\end{proof}
%We remark that Proposition~\ref{prop:av-poly} can be extended to ranked ballots (where each voter specifies a linear order over $C$) and separable committee scoring rules, i.e., rules that assign individual scores to the candidates and then choose the $k$ candidates with the highest score \cite{elkind2017propertiesmwv,faliszewski2017mwvtrend}. 
%%%%%%%%%%%%%%%%%%%%%%%%%%%%%%%%%%%%%%

\subsection{Hardness for Unit-Decreasing Thiele Rules}
\label{sec:ICE_udT}

%In our reduction from \IS to \RCE, we exploit the property of these rules that even a single change in voters' preferences can make it necessary to exchange the entire winning committee, which we have discussed in detail earlier.

We will now present our hardness result for unit-decreasing Thiele rules. Our proof also shows
parameterized hardness with respect to the committee size, and applies even to the case where $E$ and $E'$ differ only in a single approval.

\begin{theorem}\label{thm:ICE_ucT_coNP_ell=0_r=1foreverything}
    For every unit-decreasing Thiele rule $\RRR$ and every fixed value of $\ell \in [k-1]$ %and every distance $r \in \NN^+$ between elections $E$ and $E'$, 
    the problem $\RRR$-\RCE is \coNP-hard and \coWone-hard when parameterized by the committee size $k$, even if every voter approves at most two candidates, and even if $\Dist(E, E')=1$.
\end{theorem}

\begin{proof}
    We reduce from {\sc Independent Set (\IS)}. An instance of \IS is a pair 
    $(G, \kappa)$, where $G = ({\mathcal V}, {\mathcal E})$ is an arbitrary graph and 
    $\kappa \in \NN$ is a non-negative integer. It is a yes-instance if there is an independent set of size $\kappa$ in $G$, and a no-instance otherwise. IS is \NP-hard and \Wone-hard when parameterized by the solution size $\kappa$. Given an instance $(G, \kappa)$ of \IS where $G=({\mathcal V}, {\mathcal E})$ and 
    $|{\mathcal V}| = \nu$, we construct an instance of $\RRR$-\RCE as follows. 
    
    We set committee size $k$ to $\kappa$. 
    %and the allowed difference between committees $\ell$ to an arbitrary value in $[0, k-1]$. 
    In election $E$, the set of candidates $C$ is defined as $C_{\mathcal V}\cup D$, where 
    $C_{\mathcal V}=\{c_w\mid w\in {\mathcal V}\}$ is the set of {\em vertex candidates}, 
    and $D$ is a set of $k$ {\em dummy candidates}.
    Let $\alpha = \lambda(2)$, where $\lambda$ is the underlying OWA function of \RRR, and let $t = \ceil{\frac{2}{1-\alpha}}$. %; note that $t$ is an integer. 
    
    We introduce the following four voter groups in election $E$. 
    \begin{enumerate}
        \item For every edge $\{u, w\} \in {\mathcal E}$, there are $t$ \emph{edge voters} who approve the vertex candidates $c_u$ and $c_w$.
        \item For every vertex $w \in {\mathcal V}$, there are $(\nu - \degg(w)) \cdot t$ voters who approve the vertex candidate $c_w$.
        \item For every pair of candidates $d \in D$ and $c_w \in C_{\mathcal V}$, there are $t$ voters who approve both $d$ and $c_w$.
        \item For every dummy candidate $d \in D$, there are $k \cdot t$ voters who approve $d$.
    \end{enumerate}
    Note that every candidate is approved by exactly $(\nu + k) \cdot t$ voters.

    Since no two dummy candidates are approved by the same voter, the size-$k$ committee $D$ has the maximum possible score of $k \cdot (\nu+k) \cdot t$ in election $E$. Therefore,  $D \in \RRR(E, k)$. 

    Now, consider an election $E' = (C, V')$, where $V'$ is obtained by picking an arbitrary candidate $d^* \in D$ and an arbitrary voter $i$ who approves $d^*$ and some vertex candidate, and removing $d^*$ from $i$'s ballot.
    %\footnote{Note that we can increase the distance between elections $E$ and $E'$ arbitrarily by adding additional dummy candidates and ensuring that they never have sufficiently many approvals to be part of any winning committee.} 
    Then the score of the committee $D$ in $E'$ is $k \cdot (\nu+k) \cdot t - 1$. We will show that $D$ wins in election $E'$ if and only if there is no size-$\kappa$ independent set in graph $G$ and that otherwise, every winning committee $S \in \RRR(E', k)$ is entirely disjoint from $D$, i.e.,~$D \cap S = \emptyset$. Note that this generalizes the statement to all values of $\ell \in [k-1]$.

    $(\Rightarrow)$ Assume that there is an independent set $I$ of size $\kappa$ in graph $G$, and let $S_I = \{c_w \mid w \in I\}$ be the size-$k$ committee that corresponds to the vertices in $I$. Since $I$ is an independent set, no two candidates in $S_I$ are approved by the same voter, so the score of $S_I$ in $E'$ is $k \cdot (\nu + k) \cdot t$. Moreover, $S_I \cap D = \emptyset$. Assume for contradiction that there is a committee $S \in \RRR(E', k)$ such that $S \cap D \neq \emptyset$, i.e., there is a candidate $d \in S \cap D$. Since $\lscoreE{E'}{D} = k \cdot (\nu+k) \cdot t - 1$, it must hold that $S \neq D$, i.e., there is a candidate $c \in S \cap C_{\mathcal V}$. However, then there are $t$ voters who approve both $d$ and $c$. Therefore, the score of $S$ is at most $k \cdot (\nu+k) \cdot t - (1-\alpha) \cdot t$. Now, since $\RRR$ is a unit-decreasing Thiele rule, we have $\alpha < 1$ and, hence,
    $\lscoreE{E'}{S_I} > \lscoreE{E'}{S}$. Thus, $S \not\in \RRR(E', k)$, a contradiction.

    $(\Leftarrow)$ Assume that there is no independent set $I$ of size $\kappa$ in $G$. Assume for contradiction that there is a size-$k$ committee $S \subseteq C$ such that $\lscoreE{E'}{S} > \lscoreE{E'}{D}$, i.e.,\ $D \notin \RRR(E', k)$. Then $S\cap C_{\mathcal V}\neq\emptyset$; let $c$ be some candidate in $S\cap C_{\mathcal V}$. Now, if 
    $S \cap D \neq \emptyset$, there exists a candidate 
    $d \in S$, and $t$ voters who approve both $c$ and $d$.  Similarly, if $S \cap D = \emptyset$, we have $S \subseteq C_{\mathcal V}$. Since $S$ does not correspond to an independent set, there are at least $t$ edge voters who approve two candidates in $S$. Therefore, in either case, the score of $S$ is at most $k \cdot (\nu+k) \cdot t - (1-\alpha) \cdot t$. But then $t = \ceil{\frac{2}{1-\alpha}}$ implies that the quantity $\lscoreE{E'}{D} - \lscoreE{E'}{S} = (k \cdot (\nu+k) \cdot t - 1) - (k \cdot (\nu+k) \cdot t - (1-\alpha) \cdot t) = (1-\alpha) \cdot t - 1  = (1-\alpha) \cdot \ceil{\frac{2}{1-\alpha}} - 1 \geq 1$,
    %\begin{align*}
    %     \quad & (k \cdot (\nu+k) \cdot t - 1) - (k \cdot (\nu+k) \cdot t - (1-\alpha) \cdot t) \\
    %    = \quad & (1-\alpha) \cdot t - 1  = (1-\alpha) \cdot \ceil{\frac{2}{1-\alpha}} - 1 \geq 1,
    %\end{align*}
    i.e., $D$ has a strictly higher score than $S$ in  $E'$, a contradiction.% with the choice of $S$.
    
    Note that the committee size $k$ in our constructed election is equal to the solution size $\kappa$ of the given \IS instance. Therefore, our reduction is parameter-preserving.
\end{proof}
Note that we do not claim that $\RRR$-\RCE is coNP-complete; 
%this is because for Thiele rules it may be difficult to verify that a given committee is in the output of the rule.
because in the naive guess-and-check approach, one would first guess a committee $S'$ with a low enough distance from the original committee $S$, which indicates to the class \NP. Then to verify that the chosen committee $S'$ wins in the altered election, one would guess a second committee $S''$ and check if $S''$ has a higher score than $S'$ in the altered election, which indicates to the class \coNP.

\subsection{Beyond Unit-Decreasing Thiele Rules}\label{sec:beyond}
Consider a Thiele rule $\RRR_\lambda$ that is not unit-decreasing.
This means that $\lambda(1)=\lambda(2)=1$. Then either 
$\lambda(i)=\lambda(j)$ for all $i, j\in{\mathbb N}$ (i.e., $\RRR_\lambda$ is the Approval Voting rule)
or there exists an $s>1$ such that $\lambda(j)=1$
for all $j\le s$ and $\lambda(s+1)<1$. 

%In the former case, $\RRR_\lambda$ is equivalent to AV: indeed, we have $\lamscore^E(S) = \sum_{i \in N} \sum_{j = 1, \dots, |S \cap v_i|} \lambda(j) = \sum_{i\in N}|S\cap v_i|=\sum_{c\in S}\sum_{i\in N}[c\in v_i]$.

In the latter case, we can modify the proof
of Theorem~\ref{thm:ICE_ucT_coNP_ell=0_r=1foreverything}
by (1) adding a set of $s-1$ candidates $F$ that are approved by all voters,
and (2) increasing the committee size by $s-1$. Then  in both $E$ and $E'$, every 
committee with the maximum score would contain $F$. Moreover, $F\cup D$
is optimal for $\RRR_\lambda$ in $E$, and it remains optimal in $E'$ 
if and only if the underlying graph does not admit an independent set of size $\kappa$.
This establishes that $\RRR_\lambda$-\RCE is \coNP-hard. 

We are now ready to state our dichotomy result.

\begin{theorem}\label{thm:dich}
    Consider a Thiele rule $\RRR$ associated with the OWA $\lambda$.
    If $\lambda(i)=1$ for all $i \in{\mathbb N}^+$, 
    then the problem $\RRR$-\RCE is polynomial-time solvable.
    Otherwise, it is \coNP-hard and \coWone-hard when parameterized by $k$. These hardness results hold for all fixed $\ell \in [k-s]$ where $s \in \NN^+$ is the smallest number such that $\lambda(s+1) < 1$, and even if $\Dist(E, E')=1$.
\end{theorem}

%\begin{theorem}\label{thm:dich}
%    Consider a Thiele rule $\RRR_\lambda$ associated with the OWA $\lambda$.
%    If $\lambda(i)=\lambda(j)$ for all $i, j\in{\mathbb N}$, 
%    then the problem $\RRR_\lambda$-\RCE is polynomial-time solvable.
%    Otherwise, it is \coNP-hard. 
%    The hardness result holds even if $\Dist(E, E')=1$.
%\end{theorem}
%
%Note that we can no longer claim that our reduction is parameter-preserving with respect to $k$, 
%or that the hardness result holds for every value of $\ell\in [k-1]$, or under the assumption that each voter approves a constant number 
%of candidates. In particular, the parameterized complexity of $\RRR$-\RCE
%with respect to $k$ for general Thiele rules remains an intriguing open question.

%%%%%%%%%%%%%%%%%%%%%%%%%%%%%%%%%%%%%%%%%%%%%%%

\section{A Dichotomy for Greedy Thiele Rules} \label{sec:ICE_gudT}
In this section, we focus on greedy Thiele rules, 
and establish a dichotomy result that is similar to \cref{thm:dich}: if $\GRRR_\lambda$ is a greedy Thiele rule, 
$\GRRR_\lambda$-\RCE is NP-hard unless $\lambda(i)=1$ for all $i\in\mathbb N^+$ (i.e., unless $\GRRR_\lambda$ is AV). This is despite greedy Thiele rules having better computational properties than Thiele rules: e.g., it is easy to find a winning committee under a greedy Thiele rule.

%We provide a reduction from the NP-hard \textsc{Restricted Exact Cover by Three Sets (RX3C)} problem \cite{gonzalez_clustering_1985}, a variation of the well-known \textsc{Exact Cover by Three Sets} problem \cite{garey_computers_1979}.
%An instance of \textsc{RX3C} comprises of a finite set of elements $\mathcal{U} = \{u_1, \ldots, u_{3h}\}$ and a family $\mathcal{M} = \{M_1, \ldots, M_{3h}\}$ of size-$3$ subsets of $\mathcal{U}$ such that every element of $\mathcal{U}$ belongs to exactly three sets in $\mathcal{M}$; it is a yes-instance if there is a selection of exactly $h$ sets from $\mathcal{M}$ whose union is $\mathcal{U}$ (i.e., if there is an exact cover of $\mathcal{U}$ with sets from $\mathcal{M}$), and a no-instance otherwise.
%This problem is also known to be \NP-hard \cite{gonzalez_clustering_1985}.

However, our argument becomes much more involved.
Again, we start by establishing a hardness result 
for unit-decreasing rules. Our hardness reduction for this class of rules proceeds in two steps. We first define a new problem, which we call
\textsc{Candidate Inclusion (\CI)} and show it to be \NP-hard for greedy unit-decreasing Thiele rules. We then give a reduction from \CI to \RCE.

For a fixed voting rule $\RRR$, an instance of $\RRR$-\CI comprises of an election $E = (C, V)$, a committee size $k \in \NN$ and a set of candidates $P \subseteq C$; it is a yes-instance if there exists a winning committee $S \in \RRR(E, k)$ such that $P \subseteq S$, and a no-instance otherwise. This problem can be seen as a generalization of the \textsc{Winner Checking (WC)} problem studied by~\citet{aziz_computational_2015}, i.e., the task of checking whether a given committee is among the winners in a given election.

We provide a high-level idea of our hardness proof of \CI, and defer the full proof (which is quite technical) to the appendix.

\begin{proposition}\label{prop:cons_unit_Thiele_CI_NP}
    For every greedy unit-decreasing Thiele rule $\GRRR$ and
    size of $|P| \in [1, k]$, $\GRRR$-\problemname{CI} is \NP-hard.
\end{proposition}
\begin{proofsketch}
     Fix a greedy unit-decreasing Thiele rule $\GRRR$. 
     We reduce from \textsc{Restricted Exact Cover by Three Sets (RX3C)} \cite{gonzalez_clustering_1985}, which is a variant of  \textsc{Exact Cover by Three Sets}  \cite{garey_computers_1979}.
     An instance of \textsc{RX3C} comprises of a finite set of elements $\mathcal{U} = \{u_1, \ldots, u_{3h}\}$ and a family $\mathcal{M} = \{M_1, \ldots, M_{3h}\}$ of size-$3$ subsets of $\mathcal{U}$ such that every element of $\mathcal{U}$ belongs to exactly three sets in $\mathcal{M}$; it is a yes-instance if there is a selection of exactly $h$ sets from $\mathcal{M}$ whose union is $\mathcal{U}$, and a no-instance otherwise.

     Given an instance $(\mathcal{U}, \mathcal{M})$ of \problemname{RX3C}, %with $\mathcal{U} = \{u_1, \ldots, u_{3h})$ and $\mathcal{M} = \{M_1, \ldots, M_{3h}\}$, 
     we construct an instance of $\GRRR$-\CI with an election $E$, a subset $P$ of candidates, and a committee size $k$.
     Our set of candidates contains a \emph{set candidate} for every set in 
     $\mathcal{M}$, as well as three candidates $p$, $d$, and $x$. We set the committee size $k$ to $3 h + 2$. We construct voters so that $x$ is chosen in the first iteration, followed by a selection of $h$ set candidates. Then, in the $(h+1)$-th iteration, we reach the critical point where candidate $p$ \textit{can} be selected if and only if the previously selected set candidates correspond to an exact cover of $\mathcal{U}$, and otherwise, candidate $d$ is selected as a default. In the final $2 h$ iterations, all remaining set candidates are selected. The set $P$ consists of candidate $p$ and an arbitrary number of set candidates. Then, there is a winning committee $S \in \GRRR(E, k)$ with $P \subseteq S$ if and only if $\mathcal{U}$ can be covered with $h$ sets from $\mathcal{M}$. 
\end{proofsketch}

By reducing \CI to \RCE, we establish the following: 
\begin{theorem} \label{thm:ICE_gudT_all_ell}
    For every greedy unit-decreasing Thiele rule $\GRRR_\lambda$ and
    for every distance between committees $\ell \in [k-1]$, $\GRRR$-\RCE is \NP-hard, even if $\Dist(E, E')=1$.
\end{theorem}
\begin{proofsketch}
    We give a reduction from $\GRRR$-\CI to $\GRRR$-\RCE. 
    Fix a greedy unit-decreasing Thiele rule $\GRRR$. Given an instance $(\Tilde{E} = (\Tilde{C}, \Tilde{V}), \Tilde{P}, \Tilde{k})$ of $\GRRR$-\CI, we construct an instance of $\GRRR$-\RCE as follows. 
    We create two \emph{wrapper elections} $E$ and $E'$ at
    distance of $1$ from each other. 
    In $E$ and $E'$, we include all candidates in $\Tilde{C}$ and all voters in $\Tilde{V}$, an additional set $B$ of $\Tilde{k} - |\Tilde{P}|$ candidates, as well as two \emph{control candidates} $x$ and $y$. 
    We set the committee size $k$ to $\Tilde{k}+2$. 
    We construct voters in election $E$ so that candidate $x$ can be chosen in the first iteration. 
    Once $x$ is selected, all candidates in $\Tilde{C} \setminus \Tilde{P}$ lose sufficiently many points that they will not be selected in any of the subsequent iterations. 
    Thus, all candidates in $B \cup \Tilde{P}$ need to be chosen, before $y$ is chosen in the final iteration. 
    Therefore, the size-$k$ committee $S = B \cup \Tilde{P} \cup \{x, y\}$ wins in election $E$, i.e., $S \in \GRRR(E, k)$. 
    
    In contrast, we construct the voters in $E'$ so that $y$ must be selected in the first iteration. 
    Once $y$ is selected, all candidates in $B$ lose sufficiently many points that they will not be selected in any of the subsequent iterations. 
    Then, in the following $\Tilde{k}$ iterations, candidates from $\Tilde{C}$ must be chosen, before $x$ is chosen in the final iteration. 
    Thus, the intersection between $S$ and a winning committee $S'$ in $E'$ can only contain $x$, $y$, and candidates in $\Tilde{P}$. 
    
    After $y$ has been chosen and before $x$ is chosen, $\GRRR$ operates on $E'$ in the same way as it would on $\Tilde{E}$. 
    This implies that $\GRRR(E', k) = \{  \Tilde{S} \cup \{x, y\} \mid \Tilde{S} \in \GRRR(\Tilde{E}, \Tilde{k}) \}$, i.e., a committee $S'$ is winning in $E'$ if and only if it consists of candidates $x$ and $y$ and all candidates from a winning committee in $\Tilde{E}$. 
    We set the allowed difference $\ell$ between committees in our \RCE instance to $k-|\Tilde{P}| - 2$, i.e.,\ at least $|\Tilde{P}|+2$ candidates need to appear in both winning committees. Since $x$ and $y$ will always be chosen, this implies that a selection of at least $|\Tilde{P}|$ candidates from $\Tilde{S}$ need to be present in both $S$ and $S'$, which are exactly the candidates from $\Tilde{P}$. This ensures that the given instance of $\GRRR$-\CI  is a \emph{yes}-instance if and only if our constructed $\GRRR$-\RCE instance is~a~\emph{yes}-instance.
    %Then, by setting the allowed difference $\ell$ between committees in our \RCE instance to $\Tilde{k}-|\Tilde{P}|$, i.e.\ at least $|\Tilde{P}|+2$ candidates need to appear in both winning committees (note that $x$ and $y$ will always be chosen), we ensure that the given instance of $\GRRR$-\CI  is a \emph{yes}-instance if and only if our constructed $\GRRR$-\RCE instance is a \emph{yes}-instance. 
    
    The above establishes hardness for all values of $\ell \in [k-3]$. Recall that $\GRRR$-\CI is \NP-hard for every size of $|P| \geq 1$, and we remark that the distance between committees in $E$ and $E'$ can be increased by $2$ by setting $k$ to $\Tilde{k} + 1$. Then, under \GRRR, $y$ will not be chosen in the last iteration on $E$, and $x$ will not be chosen in the last iteration on $E'$. Thus, one can verify that we obtain hardness for the complete range of $\ell \in [k-1]$.
\end{proofsketch}

Recall that Greedy-AV is equivalent to AV and hence
Greedy-AV-\RCE is polynomial-time solvable.
On the other hand if $\lambda(s)=1$, $\lambda(s+1)<1$, we can use the same construction as in the proof of Theorem~\ref{thm:dich}, i.e., 
modify the proof of Theorem~\ref{thm:ICE_gudT_all_ell} by increasing the committee size by $s$ and adding $s$ candidates approved by all voters. We obtain the following corollary.

\begin{corollary} \label{cor:gr-dich}
    Consider a greedy Thiele rule $\GRRR$ associated with the OWA $\lambda$.
    If $\lambda(i) = 1$ for all $i \in \NN^+$, 
    then $\GRRR$-\RCE is polynomial-time solvable.
    Otherwise, it is \NP-hard. 
    The result holds for all $\ell \in [k-s]$, where $s \in \NN^+$ is the smallest number such that $\lambda(s+1) < 1$, and even if $\Dist(E, E') = 1$.
\end{corollary}

Previously, we have motivated the \RCE problem with settings in which it is costly to replace any member of an already implemented winning committee. 
In these settings, we are mostly interested in the computational complexity for small values of the parameter $\ell$, i.e., we allow for only very few candidates to be replaced. 
However, in other scenarios, one might mostly be concerned that there is at least \emph{some} intersection between subsequent winning committees. For instance, imagine a scenario where the board of directors of an organization is elected periodically. 
When a new board is elected without a candidate who was also part of the previous board, there might not be an adequate handover. As this would likely lead to a great reduction in productivity since the new board members will have to be acquainted with their roles completely independently, one might suggest that at least a few board members should be part of two subsequent boards of directors. We can think of this scenario in the light of a \emph{transition of power}. 
With regard to the \RCE problem, these types of scenarios motivate the study of the computational complexity for small values of $k - \ell > 0$, i.e., situations in which we want at least $k-\ell$ candidates to stay in the committee.
However, we see that such a problem is still computationally hard, with the following result.

\begin{theorem}\label{thm:ICE_cons_unit_Thiele_W[1]_k-ell+r}
For every greedy unit-decreasing Thiele rule $\GRRR_\lambda$ and
    for every fixed value of $k-\ell \in \NN^+$, parameterized by the solution size $k$, $\GRRR$-\RCE is \Wone-hard, even if every voter approves at most two candidates, and even if $\Dist(E, E')=1$.
\end{theorem}

%%%%%%%%%%%%%%%%%%%%%%%%%%%%%%%%%%%%%%%%%%%%%%%%%%

\section{Parameterized Complexity Results} \label{sec:parameterized}
 %To complement our hardness results, 
 Next, we consider \RCE for Thiele rules and their greedy variants from a parameterized complexity perspective. We present tractability results (\FPT and \XP) for some parameters, which are not already ruled out by the results in the previous section.
 
 %While positive results for some parameters (such as the committee size $k$) are already ruled out by the results in the previous section, we present tractability results (\FPT and \XP) for a selection of natural parameters and their combinations.

\subsection{Thiele Rules}\label{sec:param-thiele}
Fix a Thiele rule $\RRR$.
If $\RRR$ is unit-decreasing, then, according to \cref{thm:ICE_ucT_coNP_ell=0_r=1foreverything}, unless \FPT $=$ \coWone,
no algorithm can solve an instance $\mathcal{I}$
of $\RRR$-\RCE in $\bbigO{f(k) \cdot |I|^{\bbigO{1}}}$ time, where $k$ is the committee size and $f$ is some computable function. Thus, we cannot hope for an 
FPT algorithm with respect to $k$ that works for all Thiele rules.
However, $\RRR$-\RCE admits a simple algorithm that is \XP with respect to $k$ and \FPT with respect to $|C|=m$. 

\begin{proposition}\label{prop:ICE_udT_FPT_m_XP_k} 
For every Thiele rule $\RRR$
the problem $\RRR$-\RCE is \FPT in $m$ and \XP in $k$.
\end{proposition}
\begin{proof}
    Given an $\RRR$-\RCE instance $(E, E', S, k, \ell)$, we can go over all size-$k$ subsets of $C$, evaluate their scores in $E'$, and check if one of the committees with the maximum score is at distance at most $\ell$ from $S$. There are 
    $\binom{m}{k} \leq m^k\le m^m$ committees to consider; for each, its \RRR-score in $E'$ can be computed in polynomial time. 
    %Therefore, the overall running-time of the procedure is $\bbigO{m^k \cdot |\mathcal{I}|^{\bbigO{1}}}$ and thus \RCE is \XP in $k$.
    %Furthermore, since $k \leq m$, we have that $\bbigO{m^k \cdot |\mathcal{I}|^{\bbigO{1}}} \subseteq \bbigO{m^m \cdot |\mathcal{I}|^{\bbigO{1}}}$ and thus, \RCE is \FPT in $m$.
\end{proof}

Further, our problem is fixed-parameter tractable 
with respect to the combined parameter $n+k$, where $n$
is the number of voters. This proof, as well as some of the subsequent proofs,
is based on the idea that we can partition the candidates in $E'$
into at most $2^n$ non-empty {\em candidate classes}, so that all candidates in each class 
are approved by the same voters. 

\begin{proposition}\label{pro:ICE_udT_FPT_n+k}
For every Thiele rule $\RRR$
the problem $\RRR$-\RCE is \FPT in $n+k$.
\end{proposition}
\begin{proof}
    Given a Thiele rule \RRR and an \RCE instance $(E, E', S, k, \ell)$, we construct an election $\Tilde{E}=(\Tilde{C}, \Tilde{V})$ with $|\Tilde{C}|\le k \cdot 2^n$ such that there is a committee $\Tilde{S} \in \RRR(\Tilde{E}, k)$ with $\dist(S, \Tilde{S}) \leq \ell$ if and only if there is an $S' \in \RRR(E', k)$ with $\dist(S, S') \leq \ell$. %To this end, we partition the candidates in $E'$ into $2^n$ \emph{candidate classes}, where two candidates belong to the same class if they are approved by exactly the same voters. 
    The candidate set $\Tilde C$ contains all $k$ candidates in $S$, and,
    %(note that there are $k$ such candidates), and, 
    for every candidate class $K$, an arbitrary selection of at most $k - |K \cap S|$ of candidates from $K\setminus S$. The profile $\Tilde V$ is then obtained by restricting $V$ to $\Tilde C$.
    Since $S\subseteq \Tilde{C}$, and all candidates within each class are interchangeable, we can assume without loss of generality that a committee in $\RRR(E', k)$ that minimizes the distance to $S$ is a subset of $\Tilde C$.
    We can therefore use the same approach as in the proof of \cref{prop:ICE_udT_FPT_m_XP_k}, i.e., 
    go through all size-$k$ subsets of $\Tilde C$, 
    compute their score in $\Tilde{E}$ and distance to $S$. The bound on the running time follows from the analysis in \cref{prop:ICE_udT_FPT_m_XP_k} and the fact that $|\Tilde C|\le 2^n\cdot k$.
%    to election $\Tilde{E}$. Since all candidates that belong to the same candidate class will always have the same marginal contribution, every committee that wins in $\Tilde{E}$ also wins in $E$. Furthermore, since election $\Tilde{E}$ contains all candidates in $S$, and if we prioritize candidates in $S$ over candidates not in $S$ from the same class, it is easy to see that there is a winning committee in $\Tilde{E}$ with a maximum number of candidates from $S$. Finally, since $|S| = k$, election $\Tilde{E}$ contains at most $k$ candidates from every class and therefore, there are at most $k \cdot 2^n$ candidates in $\Tilde{E}$. From this and due to \cref{prop:ICE_udT_FPT_m_XP_k}, our result follows.
\end{proof}

Whether this result can be strengthened to a fixed-parameter tractable algorithm with respect to $n$ alone remains an open question. However, we can place our problem in \FPT with respect to $n$ for a specific well-studied rule, namely, \CCAV.
%Due to space constraints, we defer its proof to our full version \cite{us}.

\begin{proposition}\label{pro:ICE_CCAV_FPT_n}
    \CCAV-\RCE is \FPT in $n$.
\end{proposition}
%%%%%%%%%%%%%%%%%%%%%%%%%%%%%%%%%%%%%%%%%%%%%%%%%%%%%%%%%%%%

\subsection{Greedy Thiele Rules}
%Due to the simplicity of our algorithms, we omit formal proofs and only give intuitions for our approaches.
Greedy Thiele rules are less computationally demanding than Thiele rules.
As such, all easiness results from Section~\ref{sec:param-thiele}
extend to greedy Thiele rules; moreover, some variants of our problem that are hard for Thiele rules (under a suitable complexity assumption) admit \FPT algorithms for greedy Thiele rules. For instance, by \cref{thm:ICE_cons_unit_Thiele_W[1]_k-ell+r}, if $\RRR$ is a Thiele rule, 
$\RRR$-\RCE is \coWone-hard with respect to $k$ even for fixed $\ell$.
In contrast, our next proof shows that $\GRRR$-\RCE is \FPT in $k$
for any fixed value of~$\ell$.

\begin{proposition}\label{pro:ICE-greedy-k-m}
For every greedy Thiele rule $\GRRR$ the problem
    $\GRRR$-\RCE is \FPT in $k$ for every fixed value of $\ell \in \NN$, as well as \FPT in $m$ and \XP in $k$.
\end{proposition}
\begin{proof}
    Given an \RCE instance $(E, E', S, k, \ell)$, we guess a subset of candidates $S^-\subseteq S$, $|S^-|\le \ell$ to be replaced, and a subset of candidates $S^+\subseteq C\setminus S$, $|S^+|=|S^-|$, to replace them. Note that for $S'=(S\setminus S^-)\cup S^+$ we have $\dist(S, S')\le \ell$, and there 
    are at most ${k\choose \ell}\cdot{{m-k}\choose \ell}$ pairs $(S^-, S^+)$ to consider (polynomially many for constant $\ell$, and at most $2^k\cdot m^k$, as $\ell\le k$). We then guess a permutation $\pi$ of $S'$ and check if the rule $\GRRR$ can select the candidates in $S'$, in the order specified by $\pi$.
    There are $k!$ permutations to consider, so the bounds on the running time follow. 
\end{proof}

%\begin{proposition}\label{pro:ICE_gudT_FPT_k_fixed_ell}
%    \RCE for \emph{Greedy-$\mathcal{R}$} is \FPT in $k$ for every fixed value of $\ell \in \NN$.
%\end{proposition}

%\begin{proposition}\label{pro:ICE_gudT_FPT_m_XP_k}
%    \RCE for \emph{Greedy-$\mathcal{R}$} is \FPT in $m$ and \XP in $k$.
%\end{proposition}

Just as for Thiele rules, we combine the 
approach of \cref{pro:ICE-greedy-k-m} with the idea of partitioning candidates into classes to design an
algorithm that is \FPT in $n+k$.

%we can also expand this result into an \FPT algorithm parameterized by $n$ and $k$. The proof is the same as in \cref{pro:ICE_udT_FPT_n+k}, except that we combine it with \cref{pro:ICE_gudT_FPT_m_XP_k} instead of \cref{prop:ICE_udT_FPT_m_XP_k} here.
\begin{proposition}\label{pro:ICE_GcudT_FPT_n+k}
For every greedy Thiele rule $\GRRR$ the problem
    \GRRR-\RCE is \FPT in $n+k$.
\end{proposition}

For \GreedyCC, we can strengthen this result from \FPT in $n+k$ to \FPT in $n$; it remains an open problem if a similar tractability result holds for other greedy Thiele rules.
%As for the global variants, the existence of an \FPT algorithm with respect to $n$ is still an open problem, but we are able to obtain one for \GreedyCC:
%Due to space constraints, we defer its proof to our full version \cite{us}.

\begin{proposition}\label{pro:ICE_G-CCAV_FPT_n}
    \GreedyCC-\RCE is \FPT in $n$.
\end{proposition}

%%%%%%%%%%%%%%%%%%%%%%%%%%%%%%%%%%%%%%%%%%%

\section{Experiments} \label{sec:experiments}
To complement our theoretical analysis, we conduct experiments to gain insights into the practical facets of our problem. 
%We focus on approval-based elections, as most (if not all) of the current works on temporal elections only consider approval preferences.
Thereby, we also contribute to the small, growing body of experimental work on approval-based elections \cite{boehmer2024guide}.
Given that computing a winning committee for most Thiele rules is already computationally intractable, we focus on two popular greedy Thiele rules: \GreedyCC~and~\GreedyPAV.

Our experiments focus on the following three questions. \textbf{Q1:} How resilient are winning committees under  \GreedyCC and \GreedyPAV, i.e., how much do they need to change when votes do? \textbf{Q2:} How good are solutions of \RCE obtained by employing lexicographical tie-breaking when computing the original and updated winning committee? \textbf{Q3:} Is there a correlation between (i) the round in which the greedy rule included a candidate in the committee and (ii) how often the candidate gets replaced after changes in the votes occur?

In Experiments 1 and 3, we assume \emph{lexicographical tie-breaking} in the computation of \GreedyCC and \GreedyPAV, i.e., we break ties based on some fixed order of candidates whenever multiple candidates have the same marginal contribution. This implies that both rules become resolute, i.e., they return a unique winning committee. In particular, this allows us to solve \RCE in polynomial time. The code for our experiments is available online \cite{ExperimentCode}.

\paragraph{Experimental Design}
We consider two different models for generating approval-based preferences, both of which are well-studied in the literature \cite{boehmer2024guide,DBLP:conf/ijcai/DelemazureLLS22,DBLP:conf/aaai/GodziszewskiB0F21,szufa2022sampleapproval}.
%The two models are described below.
%EE for succintness, we can say that in both models voters and candidates are assigned to locations sampled uniformly at random from a metric space, and then say what that metric space is
\begin{itemize}
    \item In the \emph{1D-Euclidean Model (1D)}, each voter $v$ (resp.\ candidate $c$) is assigned (uniformly at random) a point $p_v$ (resp.\ $p_c$) in the interval $[0,1]$.
    The model is parameterized by a  \emph{radius} $\tau \in [0,1]$.
    A voter $v$ approves of a candidate $c$ if and only if $|p_v - p_c| \leq r$.
    \item In the \emph{2D-Euclidean Model (2D)}, each voter and candidate is assigned (uniformly at random) a point in the unit square $[0,1] \times [0,1]$.
    A voter $v$ approves of a candidate $c$ if and only if the Euclidean distance between their points is at most the radius $\tau$.
    %\item In the \emph{Resampling Model} (introduced by \citet{szufa2022sampleapproval}), we have two parameters, $p \in [0,1]$ and $\phi \in [0,1]$. %To generate an election with candidate set $C = \{c_1, \dots, c_m\}$ and with $n$ voters, 
    %We first choose uniformly at random a central vote $u$ approving exactly $\lfloor pm \rfloor$ candidates. Then, we generate the votes, considering the candidates one by one independently for each vote. For a vote $v$ and candidate $c$, with probability $1-\phi$ we copy $c$’s approval status from $u$ to $v$ (i.e., if $u$ approves $c$, then so does $v$; if $u$ does not approve $c$ then neither does $v$), and with probability $\phi$ we ``resample'' the approval status of $c$, i.e., we let $v$ approve $c$ with probability $p$ (and disapprove it with probability $1-p$). On average, each voter approves about $pm$ candidates.
\end{itemize}

We focus on elections with  $n = \num{1000}$, $m = 100$, and committee size $k=10$, which is standard in the literature \cite{boehmer2024guide}.
We set the radius for the 1D (resp.\ 2D) model to  $0.051$ (resp.\ $0.195$), so that, on average, every voter approves around $10$ candidates. In all three experiments, we sample $100$ elections for each of the two models.

We also conduct experiments for a greater range of radii, a sampling method known as \emph{Resampling}, as well as for  1D and 2D models in combination with Resampling. We defer the model definitions and detailed trends to the appendix, and only briefly highlight the differences to the above sampling models.

To capture change in the votes, we consider three different operations: \emph{ADD}, \emph{REMOVE}, and \emph{MIX}. 
Given a number $r$ of changes to be performed, for \emph{ADD}, we uniformly at random add $r$ new approvals to the elections (i.e., we sample an $r$-subset of all voter-candidate pairs where the voter does not approve the candidate). For \emph{REMOVE}, we uniformly at random delete $r$ existing approvals from the election, whereas for \emph{MIX}, we add and remove $\floor{\nicefrac{r}{2}}$ approvals each. 

We consider different levels of change as determined by a 
change percentage $p \in [0\%, 10\%]$. 
For an election $E$, a change percentage of 
$p$ corresponds to making $r = \floor{\mathrm{app}(E) \cdot p}$ changes, where $\mathrm{app}(E)$ is the total number of approvals in $E$. 
We consider $15$ change percentages, quadratically scaled, to ensure that smaller amounts of changes are captured in greater detail.  

%For each considered level of change $p\in [0\%, 10\%]$ and election $E$, we perform $100$ iterations, where for each iteration we independently apply $r = \floor{\mathrm{app}(E) \cdot p}$ changes. 

%Finally, for each model, we sample $100$ elections, and perform $100$ iterations (in other words, every data point in our results is informed by $10,000$ samples). 
%For every iteration, based on the original election $E$, we sample an election $E'$ such that $\Dist(E, E') = r$, where $r$ depends on the initial number of approvals: Let $r = \mathrm{app}(E) \cdot p$, where $\mathrm{app}(E)$ is the number of approvals in $E$, and $p \in [0\%, 10\%]$ with $15$ quadratically scaled steps, to ensure that smaller amounts of changes are captured in greater detail. 

%We proceed to discuss each of the  experiments conducted, together with the results obtained.

\paragraph{Experiment 1: Resilience of Greedy Thiele Rules}%\label{subsec:experiments_resilience}

\begin{figure}[t]
\vspace{-.5cm}
\centering
\subfloat[\centering \GreedyCC, 1D, $\tau= 0.051$]{{\includegraphics[height=3.3cm]{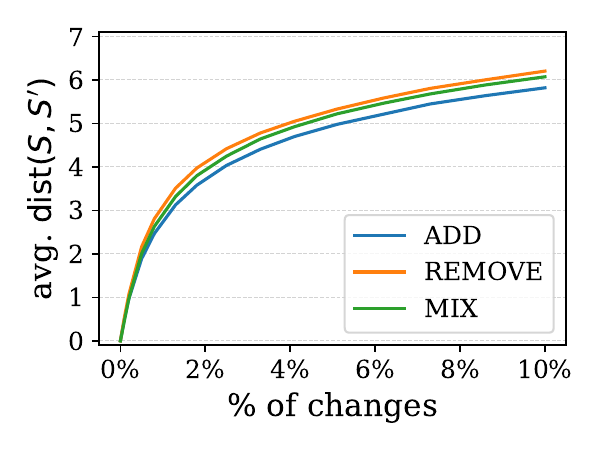}}}
\hspace{.1cm}
\subfloat[\centering \GreedyPAV, 1D, $\tau= 0.051$]{{\includegraphics[height=3.3cm]{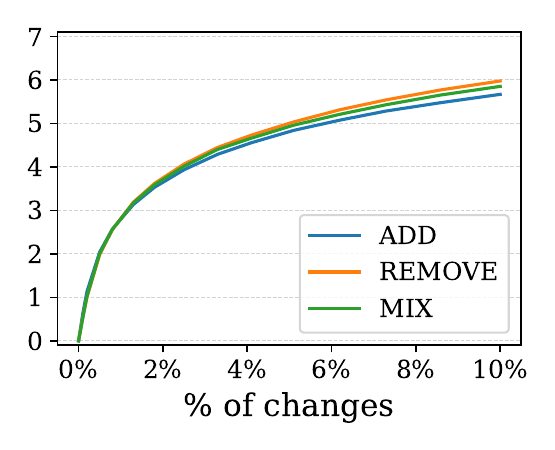}}}

\subfloat[\centering \GreedyCC, 2D, $\tau= 0.195$]{{\includegraphics[height=3.3cm]{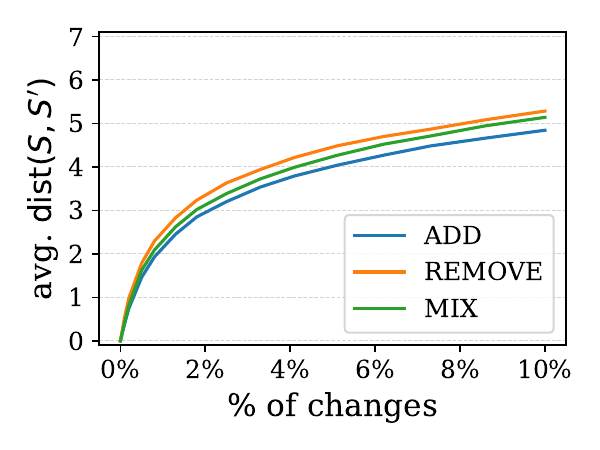}}}
\hspace{.1cm}
\subfloat[\centering \GreedyPAV, 2D, $\tau= 0.195$]{{\includegraphics[height=3.3cm]{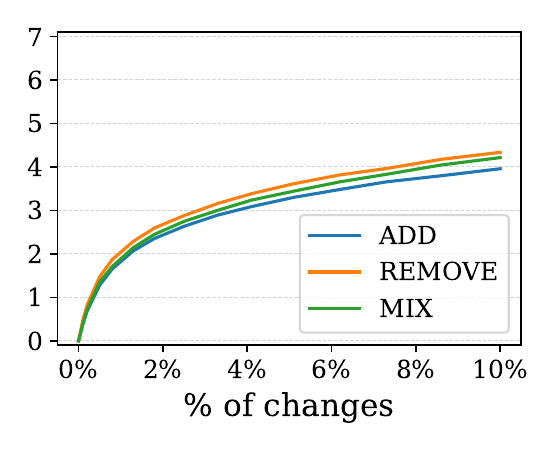}}}
\caption{Results of Experiment 1. $x$-axis is the percentage change between the original election $E$ and the adapted election $E'$; $y$-axis is the average distance between the two winning committees.}
\label{fig:Exp1}
\end{figure}

We analyze the resilience of winning committees, i.e., how much they change when voters' preferences change, and how this depends on the voting rule and operation type. For this, for each considered level of change $p$ and election $E$, we perform $100$ iterations. For each iteration, we sample an election $E'$ by applying $r = \floor{\mathrm{app}(E) \cdot p}$ changes and compute the distance between the winning committee in $E$ and $E'$.
The average distance between winning committees in the original and the modified election can be found in \Cref{fig:Exp1}.

Examining \Cref{fig:Exp1}, we find that in all considered settings, winning committees are highly non-resilient.
In particular, changing only $1\%$ of the approvals \emph{at random} leads to (on average) the replacement of two of the ten committee members.  
If we increase the change percentage further to $10\%$, around half of the committee gets replaced.
These observations hint at a general non-robustness of \GreedyCC and \GreedyPAV, and a high fragility of produced outcomes. 

%Analyzing the differences between the setups in more detail, it turns out that 
In fact, winning committees under \GreedyCC and \GreedyPAV tend to produce---on average---committees of similar resilience when elections are sampled with the 1D model. 
However, while elections sampled with the 2D model are generally more resilient for both voting rules, \GreedyPAV has a slight edge over \GreedyCC.

Turning to the different operation types, one might intuitively expect that removing approvals leads to greater changes in the winning committee, as randomly removed approvals, generally speaking, hurt winning candidates with a higher probability. However, while this is indeed the case, the observed difference is not very prominent. 

For 1D and 2D with Resampling, the trends are very similar, but both rules produce slightly more resilient committees (on average, around $0.5$ to $1$ fewer candidates need to be replaced given a change rate of $10\%$). For Resampling, the produced committees are a lot more resilient (the highest measured average of the number of candidates that need to be replaced was just over $2$ for \GreedyCC), and the outcome is highly dependent on the choice of sampling parameters (see the appendix for a detailed discussion).

\paragraph{Experiment 2: The Role of Ties}%\label{subsec:experiments_ties}

\begin{figure}[t]
%\vspace{-.5cm}
\centering
\subfloat[\centering \GreedyCC, 1D, $\tau= 0.051$]{{\includegraphics[height=3.3cm]{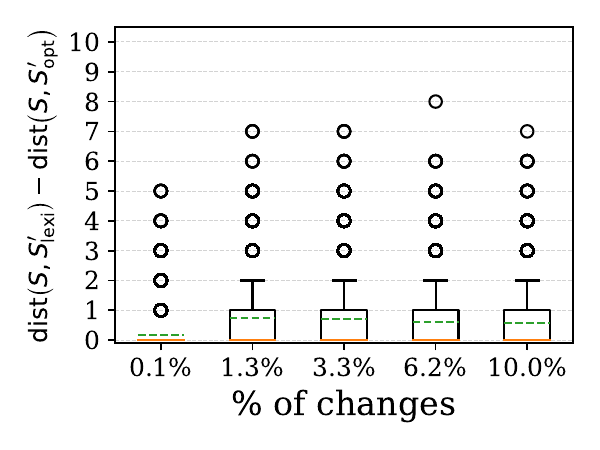}}}
\hspace{.1cm}
\subfloat[\centering \GreedyPAV, 1D, $\tau= 0.051$]{{\includegraphics[height=3.3cm]{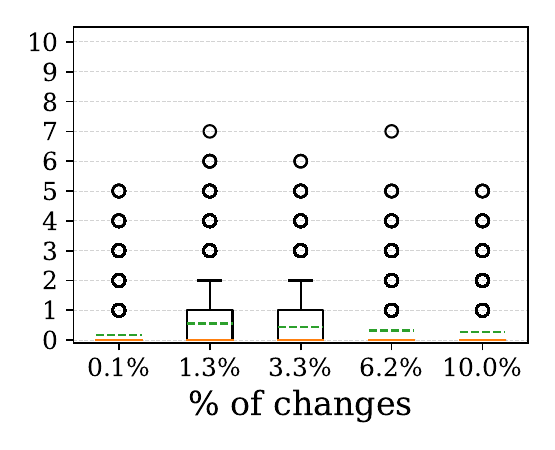}}}

\subfloat[\centering \GreedyCC, 2D, $\tau= 0.195$]{{\includegraphics[height=3.3cm]{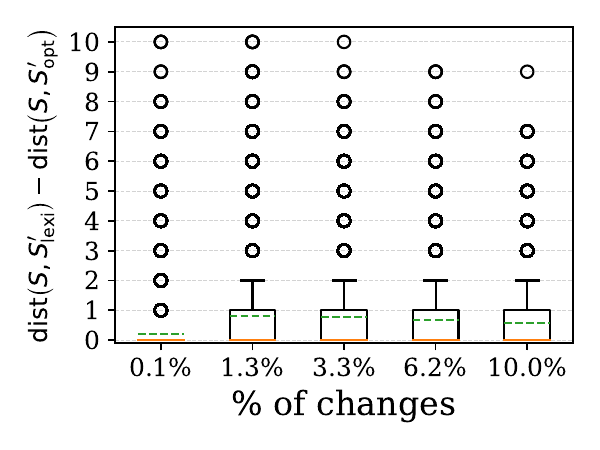}}}
\hspace{.1cm}
\subfloat[\centering \GreedyPAV, 2D, $\tau= 0.195$]{{\includegraphics[height=3.3cm]{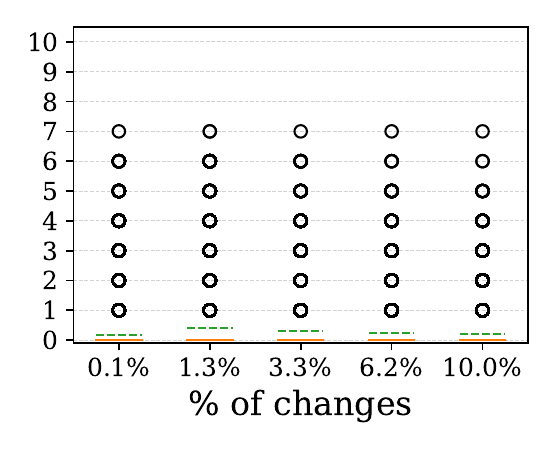}}}
\caption{Results of Experiment 2. Focus is only on MIX operation. Orange lines represent the median and dashed green lines the mean. $x$-axis is the percentage change between the original election $E$ and the adapted election $E'$; $y$-axis is the $\dist(S, S_{\mathrm{lexi}}') - \dist(S, S_{\mathrm{opt}}')$, where $S$, and $S_{\mathrm{lexi}}'$ are the respective winning committees in $E$ and $E'$ under lexicographic tie-breaking, and $S_{\mathrm{opt}}'$ is chosen out of $100$ tied winning committees in $E'$ to be closest to $S$.}
\label{fig:Exp2}
\end{figure}

In Experiment 1, we circumvented the intractability of \RCE for \GreedyCC and \GreedyPAV by applying lexicographic tie-breaking. 
In fact, this can be seen as a natural heuristic to solve \RCE (compute $S'$ using the same tie-breaking rule that was used to pick $S$). 
Ties play a surprisingly important role in winner determination of greedy Thiele rules \cite{DBLP:conf/ijcai/JaneczkoF23}.
%The results of \citet{DBLP:conf/ijcai/JaneczkoF23} indicate that ties play a much more important role in winner determination of greedy Thiele rules as one might suspect. 
%Hence, our second experiment (\Cref{fig:Exp2}) aims to investigate the effectiveness of this strategy and more generally, the importance of tie-breaking in the context of the \RCE problem. 
Hence, our second experiment (\Cref{fig:Exp2})  investigates the effectiveness of this strategy and the importance of tie-breaking for the \RCE problem. 

%For our second experiment, we investigate the role of ties in the adapted election $E'$.

%Arguably, the theoretical results obtained are highly dependent on the existence of ties: When there is a tie-breaking procedure in place, \RCE is equivalent to the the problem of finding the winning committee, and checking that its distance to the original winning committee does not exceed the given bound. 

%In many real-world elections, ties are often not seen as a considerable issue.
%Hence, the theoretical hardness results we obtained for greedy Thiele rules might not carry much weight, as it is possible that, in most cases, an optimal solution to \RCE can be found by applying alphabetic tie-breaking over considering all tied committees. 
%However, they are a frequent occurrence for many multiwinner voting rules \cite{DBLP:conf/ijcai/JaneczkoF23}. Therefore, we investigate the price that one needs to pay with lexicographical tie-breaking. 

For each election $E$, we make $S$ the winning committee in $E$ under lexicographical tie-breaking. Subsequently, for all considered change percentages, we sample $100$ elections $E'$. 
For each of these elections $E'$, we compute up to $100$ committees winning in the election and pick $S_{\mathrm{opt}}'$ to be the one closest to $S$. Further, let $S_{\mathrm{lexi}}'$ be the committee winning under lexicographic tie-breaking in $E'$. 
Using this, we compute the quantity $\dist(S, S_{\mathrm{lexi}}') - \dist(S, S_{\mathrm{opt}}')$, i.e., the number of additional candidates that need to be replaced in $S_{\mathrm{lexi}}'$ compared to $S_{\mathrm{opt}}'$.
This can be interpreted as 
a lower bound on the ``price'' paid for using our heuristic instead of solving~\RCE~optimally. %We present our results in \Cref{fig:Exp2}.

Across all considered sampling methods and voting rules, we find instances that show drastic differences in the contiguity offered by $S_{\mathrm{lexi}}'$ and $S_{\mathrm{opt}}'$. 
Specifically, for a percentage change of $1.3\%$, the difference was non-zero in $\sim \nicefrac{1}{3}$ of cases, and at least $3$ in $\sim 7.8\%$ of cases.
%Specifically, out of the $\num{40000}$ sampled committee pairs, there were the TODO instances where we encountered the worst possible distance of $10$, and in $TODO\%$ of instances, the difference was above $5$. 
For both voting rules, \GreedyCC and \GreedyPAV, the outliers tend to be more extreme on the 2D, compared to the 1D model. 

The surprising trend for the difference to decrease again for higher changes to the underlying elections is due to these elections being much \enquote{noisier} than pure 1D and 2D elections. Thus, they tend to have far fewer ties. However, while this behaviour in terms of ties is not present in the 1D and 2D with Resampling models, they produce an almost identical picture, with similarly drastic outliers. Despite the generally higher resilience in the Resampling model, here we also witness outliers up to the value of $6$.
%However, we investigated the behaviour on 1D and 2D elections that have been exposed to \enquote{noise} (Resampling), before any changes have been performed, and witnessed similarly drastic values of $\dist(S, S_{\mathrm{lexi}}') - \dist(S, S_{\mathrm{opt}}')$. The exact trends are deferred to the appendix.
The results show that, for greedy Thiele rules, lexicographic tie-breaking does not constitute a reliable approximation for finding a committee that is close to the original one, highlighting the prominence of ties in these rules and motivating the search for optimal solutions via the \RCE problem.

\paragraph{Experiment 3: Who Gets Replaced?}

\begin{figure}[t]
%\vspace{-.5cm}
\centering
\subfloat[\centering \GreedyCC, 1D, $\tau = 0.051$]{{\includegraphics[height=3.3cm]{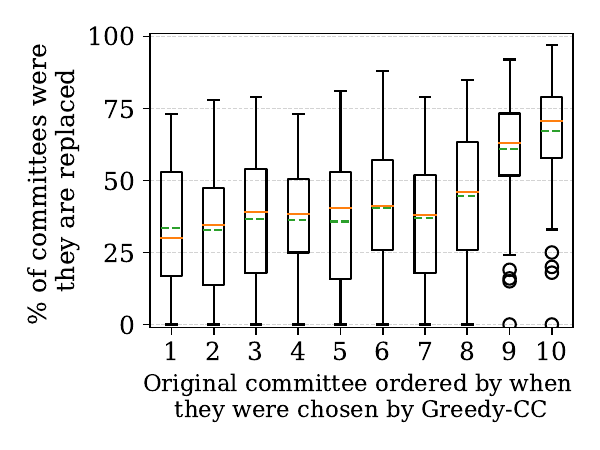}}}
\hspace{.1cm}
\subfloat[\centering \GreedyPAV, 1D, $\tau = 0.051$]{{\includegraphics[height=3.3cm]{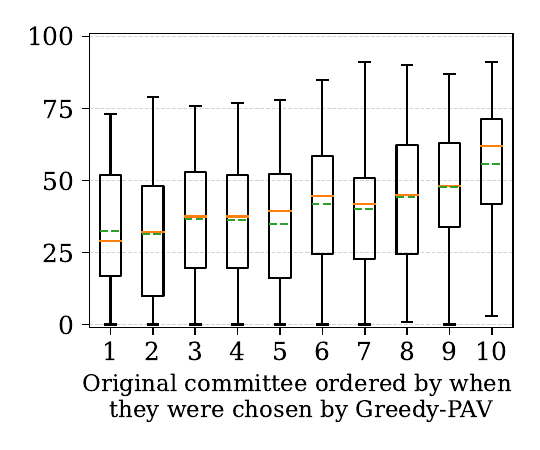}}}

\subfloat[\centering \GreedyCC, 2D, $\tau = 0.195$]{{\includegraphics[height=3.3cm]{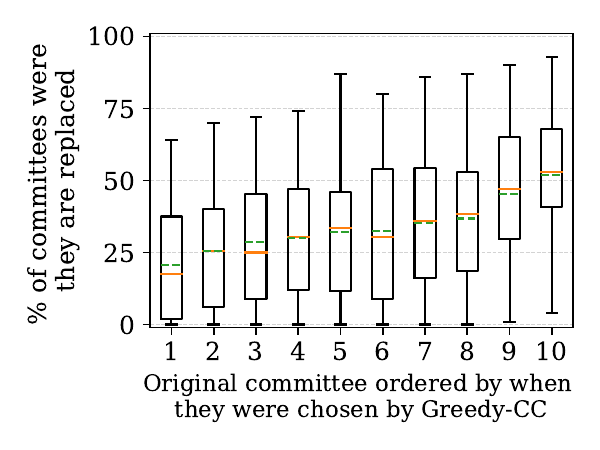}}}
\hspace{.1cm}
\subfloat[\centering \GreedyPAV, 2D, $\tau = 0.195$]{{\includegraphics[height=3.3cm]{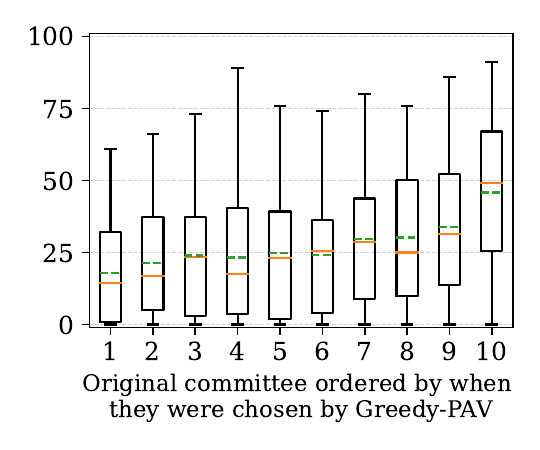}}}
\caption{Results of Experiment $3$. Focus is only on MIX operation and changes of $2.5\%$ in relation to the original number of approvals were applied. Orange lines represent the median and dashed green lines the mean. $x$-axis corresponds to candidates from the original winning committee, ordered by when they were chosen in the given greedy Thiele rule; $y$-axis is the percentage of winning committees in the adapted elections where each candidate is replaced.}
\label{fig:Exp3}
\end{figure}

We try to determine if some members of the initial winning committee get 
replaced more often when changes occur. 
In light of the round-based nature of greedy Thiele rules, one might expect the candidates chosen in later rounds to be generally weaker and hence more likely to be replaced.
To address this hypothesis, we fix a change percentage of $2.5\%$ for the \emph{MIX} operation, and for each election $E$, sample $100$ elections by applying $r = \floor{\mathrm{app}(E) \cdot 2.5\%}$ changes to $E$. 
For each member $c$ of the winning committee for $E$, we determine how many winning committees of the $100$ sampled elections contain $c$. \Cref{fig:Exp3} shows the results of this experiment as boxplots, grouped by the round in which candidates were added to the initially winning~committee.

While the suspected correlation is present, the dependence is surprisingly weak.
The most notable trend is that the last selected candidate is consistently the weakest, especially so for \GreedyCC under 1D elections.
While, compared to the 1D model, the correlation is slightly more prominent in the 1D with Resampling model, Resampling has seemingly no effect on the correlation in the 2D model.

\section{Conclusion}
We presented a complexity dichotomy (along with parameterized complexity results) for \RCE under Thiele rules and their greedy variants.
We also conducted three experiments that unveiled interesting practical insights for our problem; in particular, they show that ties are highly prevalent and therefore the decision on how to break them to achieve a particular goal (as captured by \RCE) is an~important~one.

Natural directions for future work include considering weighted
%--rather than equal--
cost in the replacement of candidates, or 
investigating \RCE with different classes of rules, as well as removing or adding candidates or votes. %Further, it would be interesting to investigate the resilience of voting rules that are specifically crafted for a temporal setting perform. 
It would also be interesting to explore the resilience of voting rules designed for temporal settings. 
A canonical candidate is, e.g., the voting mechanism for the AAAI Executive Council, where one-third of the positions are up for election each time.

%%%%%%%%%%%%%%%%%%%%%%%%%%%%%%%%%%%%%%%%%%%%%%%%%%%%%%%%%%%%%%%%%%%%%%%%

%%% Use this command to include your bibliography file.

\bibliography{abb,mybibfile}

%\iffalse
\newpage
\appendix
\begin{center}
\Large
\textbf{Appendix}
\end{center}

\vspace{2mm}

\section{Proofs Omitted from \cref{sec:ICE_gudT}}
\subsection{Proof of \cref{prop:cons_unit_Thiele_CI_NP}}
We reduce from the \NP-hard \problemname{RX3C} problem. Let $(\mathcal{U}, \mathcal{M})$ be an instance of \problemname{RX3C}, where $\mathcal{U} = \{u_1, \ldots, u_{3h}\}$ is the ground set,  set of elements and $\mathcal{M} = \{M_1, \ldots, M_{3h}\}$ is a family of size-$3$ subsets of $\mathcal{U}$ such that each element of $\mathcal U$ belongs to exactly three sets in $\mathcal{M}$. Assume that $h \geq 2$. Let \GRRR be a greedy unit-decreasing Thiele rule. 

    We construct an instance of $\GRRR$-\CI with an election $E = (C, V)$, a subset of candidates $P \subseteq C$ and committee size $k$. The set of candidates $C$ contains
    a candidate $m_i$ for every size-$3$ set $M_i \in \mathcal{M}$ (we refer to these candidates as {\em set candidates}), as well as three additional candidates $p, d$ and $x$. We set the committee size $k$ to $3h + 2$.  
    
    Let 
    $$
    \alpha = \lambda(2),\quad 
    \beta = \lambda(3), \quad 
    \gamma = \lambda(4),
    $$
    where $\lambda$ is the underlying OWA function of \GRRR. We focus on the case where $1 - \alpha > \alpha - \beta$ and $1 - \alpha > \beta - \gamma$;later we will show how to alter our construction if this assumption is not satisfied. L
    
    et $t = \ceil{\frac{3}{1-\alpha}}$ and let $T = h \cdot t^3$ be two integers. Since \GRRR is unit-decreasing, we have $0 \leq \alpha < 1$ and therefore $0 < 1 - \alpha \leq 1$ hold and thus $t \geq 3$. We introduce the following seven voter groups.
    \begin{enumerate}
        \item For every element $u_j \in \mathcal{U}$, there are $t$ \textit{element voters} who approve those set candidates $m_i$ that correspond to a set $M_i \in \mathcal{M}$ such that $u_j \in M_i$, as well as candidate $d$.
        \item For every set $M_i \in \mathcal{M}$, there are $T$ voters who approve the set candidate $m_i$.
        \item For every two sets $M_i, M_j \in \mathcal{M}$ with $M_i \neq M_j$, there are $T$ voters who approve the set candidates $m_i$ and $m_j$.
        \item There are $2 \cdot h \cdot T + 4 \cdot t$ voters who approve candidates $p$ and $d$. 
        \item There are $h \cdot T$ voters who approve candidates $p, d$ and $x$.
        \item There are $3 \cdot h \cdot t$ voters who approve candidates $p$ and $x$.
        \item There are $4 \cdot h \cdot T + 2 \cdot h \cdot t + 5 \cdot t$ voters who approve candidate $x$.
    \end{enumerate}

    In the table below, %\cref{tab:greedy_unit_Thiele_CI_NP}
     we give an overview of all voter approvals.
    \begin{table}
        \label{tab:greedy_unit_Thiele_CI_NP}
        \begin{tabular}{ c || p{1.4cm} | p{1.4cm} | p{0.95cm} | p{0.95cm} | p{0.95cm} }
          & $m_i$ & $m_j$ & $d$ & $p$ & $x$ \\ 
         \hline\hline
         $m_i$ & $3 h T + 3 t$ & $|M_i \cap M_j| \cdot t + T$ & $3 t$ & $0$ & $0$ \\
    \hline
    $m_j$ & $|M_i \cap M_j| \cdot t + T$ & $3 h T + 3 t$ & $3 t$ & $0$ & $0$ \\
    \hline
    $d$ & $3 t$ & $3 t$ & $3 h T + 3 h t + 4 t$ & $3 h T + 4 t$ & $h T$ \\
    \hline
    $p$ & $0$ & $0$ & $3 h T + 4 t$ & $3 h T + 3 h t + 4 t$ & $h T + 3 h t$ \\
    \hline
    $x$ & $0$ & $0$ & $h T$ & $h T + 3 h t$ & $5 h T + 5 h t + 5 t$ \\
        \end{tabular}
        \caption{\normalfont
        Overview of voter approvals in election $E$ in the proof of \cref{prop:cons_unit_Thiele_CI_NP} where $m_i$ and $m_j$ are two set candidates that correspond to two sets $M_i, M_j \in \mathcal{M}$. For a pair of candidates, the corresponding table entry states the number of voters that approve both candidates. Naturally, the table is symmetric.
    }
    \end{table}
    Let $S = \{m_i \mid M_i \in \mathcal{M}\} \cup \{x, p\}$ and let $S' = \{m_i \mid M_i \in \mathcal{M}\} \cup \{x, d\}$ be two size-$(3 \cdot h + 2)$ committees. We will show that $S$ and $S'$ are the only possible winning committees in election $E$ under \GRRR and, furthermore, that $S$ is one of the winning committees if and only if there is an exact cover of $\mathcal{U}$ with $h$ sets from $\mathcal{M}$. In all other cases, $S'$ is the unique winning committee. In favour of intelligibility, we will sometimes refer to the set candidates and element voters as if they were the actual sets and elements from the given \problemname{RX3C} instance and vice versa. 

    \begin{claim}\label{claim:greedy_unit_Thiele_CI_NP__x_first}
        In election $E$, in the first iteration of \GRRR, candidate $x$ is chosen.
    \end{claim}
    \begin{claimproof}
        One can verify that, before any candidate is selected, candidate $x$ has a score of exactly $5 h T + 5 h t + 5 t$, which is much higher than the score of any other candidate. \GRRR will therefore select candidate $x$ in the first iteration. 
    \end{claimproof}

    \begin{claim}\label{claim:greedy_unit_Thiele_CI_NP__2_to_h+1}
        In election $E$, in iterations $2$ to $h+1$ of \GRRR, a set of $h$ set candidates is chosen.
    \end{claim}
    \begin{claimproof}
        Because of \cref{claim:greedy_unit_Thiele_CI_NP__x_first}, candidate $x$ is chosen in the first iteration. 
        
        Then, candidates $p$ and $d$ have $2 h T + 4 t$ points from the fourth group of voters and $\alpha h T$ points from the fifth group of voters. Candidate $p$ has an additional $3\alpha h t$ points from the sixth group of voters. Candidate $d$ has at most $3 h t$ additional points from the first group of voters. Therefore, both $p$ and $d$ have at most $2 T + 4 t + \alpha h T + 3 h t$ points each. 
        
        On the other hand, each set candidate has $T$ points from the second group of voters and, since at most $h-1$ set candidates can have been selected, at least $2 h T + \alpha T \cdot (h-1)$ points from the third group of voters. Thus, every set candidate has at least $T + 2 h T + \alpha T \cdot (h-1)$ points.
    
        Therefore, with $T = h t^3$, $t = \ceil{\frac{3}{1-\alpha}}$, $\alpha < 1$, $t \geq 3$ and $h \geq 2$ we can show that every set candidate has a strictly higher score than candidates $p$ and $d$: 
        \begin{align*}
            & (T + 2hT + \alpha T \cdot(h-1) ) - (2  h  T + 4  t + \alpha  h T + 3 h t) \\
            = \quad & (1 - \alpha) \cdot T - 4 t - 3 h t \\
            = \quad & (1 - \alpha) \cdot h t^3 - 4 t - 3 h t \\
            = \quad & (1 - \alpha) \cdot \ceil{\frac{3}{1-\alpha}}h t^2 - 4 t - 3 h t \\
            \geq \quad & 3 h t^2 - 4 t - 3 h t \\
            \geq \quad & 9 h t - 7 h t \\
            = \quad & 2 h t \\
            \geq \quad & 12.
        \end{align*}
        Therefore, from the second to the $(h+1)$-th iteration, a selection of $h$ set candidates must be chosen. 
    \end{claimproof}
    \begin{claim}\label{claim:greedy_unit_Thiele_CI_NP__independent_set_candidate_first}
        Let $\mu \in [2, h+1]$ and let $m_i$ and $m_j$ be two set candidates such that in the first $\mu-1$ iterations, no set candidate was selected whose corresponding set has a non-empty intersection with the corresponding set of $m_i$. Then, in election $E$, iteration $\mu$ of \GRRR, candidate $m_i$ has a weakly higher score than candidate $m_j$.
    \end{claim}
    \begin{claimproof}
        One can verify that candidates $m_i$ and $m_j$ have the same number of points from voter groups two and three. Furthermore, candidate $m_i$ has the maximal number of $3 \cdot t$ points from voter group one, since no set candidate was chosen in a previous iteration whose corresponding set has a non-empty intersection with the corresponding set of $m_i$. Therefore, candidate $m_i$ has a weakly higher score than candidate $m_j$.
    \end{claimproof}

    \begin{claim}\label{claim:greedy_unit_Thiele_CI_NP__h+2}
        In election $E$, in iteration, $h+2$ of \GRRR, candidate $p$ or candidate $d$ is chosen, and candidate $p$ can only be chosen if the previously selected set candidates correspond to an exact cover of $\mathcal{U}$.
    \end{claim}
    \begin{claimproof}
        Because of \cref{claim:greedy_unit_Thiele_CI_NP__x_first} and \cref{claim:greedy_unit_Thiele_CI_NP__2_to_h+1}, candidate $x$ and a selection of $h$ set candidates were chosen in the first $h+1$ iterations. We denote by $I$ the selection of $h$ set candidates that were chosen in the previous iterations.

        Then, each set candidate will have at most $3 \cdot t$ points from the first group of voters, exactly $T$ points from the second group of voters and exactly $2 \cdot (h-1) \cdot T + \alpha \cdot h \cdot T$ points from the third group of voters. Thus, each set candidate has a score of at most $2 h T + \alpha h T + 3 \cdot t$. 
    
        On the other hand, candidate $p$ has $2 h T + 4 t$ points from the fourth group of voters, $\alpha h T$ points from the fifth group of voters, and $3 \alpha h t$ points from the sixth group of voters. Thus, candidate $p$ has a score of exactly $2 h T + 4 t + \alpha h T + 3 \alpha h t$.
    
        Therefore, with $t \geq 3$, $1 > \alpha \geq 0$ and $h \geq 2$, we can show that $p$ has a strictly higher score than any set candidate:
        \begin{align*} 
            & (2  h  T + 4  t + \alpha  h T + \alpha  3  h  t) - (2  h  T + \alpha  h  T + 3  t) \\
            = \quad & t + 3 \alpha  h  t \geq 3.
        \end{align*}

        Furthermore, candidate $d$ has $2 h T + 4 t$ points from the fourth group of voters, $\alpha h T$ points from the fifth group of voters and $3 h t - y$ points from the first group of voters, where the value of $y$ is as follows. If $I$ corresponds to an exact cover of $\mathcal{U}$ then $d$ has lost exactly $3 t \cdot (1-\alpha)$ points from the first group of voters in each iteration where a set candidate was selected. In this case, $y = 3 h t \cdot (1-\alpha)$ and $d$ has $2 h T + 4 t + \alpha h \cdot T + 3 h t - 3 h t \cdot (1-\alpha) = 2 h T + 4 t + \alpha h T + 3\alpha h t$ points, which is the exact score that $p$ has in this iteration. However, if $I$ does not correspond to an exact cover of $\mathcal{U}$, there must have been at least one iteration where some set candidate $m_i$ was chosen such that a different set candidate $m_j$ was chosen in a previous iteration such that the two corresponding sets $M_i$ and $M_j$ have a non-empty intersection, i.e.\ $M_i \cup M_j \neq \emptyset$. While some element voters then only contribute $\beta$ or $\gamma$ points to the score of candidate $d$, there must be some other element voters such that no set candidates approved by these voters were selected. Thus, there must be some iterations were $d$ only lost $t \cdot (a+b+c)$ points from the first group of voters, where $a, b \in \{(1-\alpha), (\alpha-\beta), (\beta-\gamma)\}$ and $c \in \{(\alpha-\beta), (\beta-\gamma)\}$. Since we assumed that $1 - \alpha > \alpha - \beta$ and $1 - \alpha > \beta - \gamma$, it holds that $t \cdot (a+b+c) < 3 t \cdot (1-\alpha)$ and $d$ will have strictly more points than $p$. Therefore, in the $(h+2)$-th iteration, \GRRR \emph{can} choose candidate $p$ if and only if the previous selection of set candidates $I$ corresponds to an exact cover of $\mathcal{U}$ and \GRRR has to choose $d$ otherwise. 
    \end{claimproof}

\begin{claim}\label{claim:greedy_unit_Thiele_CI_NP__h+2_to_3h+2}
        In election $E$, in iteration $h+3$ to $3 h+2$ of \GRRR, all remaining set candidates are chosen.
    \end{claim}
    \begin{claimproof}
        Because of \cref{claim:greedy_unit_Thiele_CI_NP__x_first}, \cref{claim:greedy_unit_Thiele_CI_NP__2_to_h+1} and \cref{claim:greedy_unit_Thiele_CI_NP__h+2}, candidate $x$, a selection of $h$ set candidates, as well as either candidate $p$ or candidate $d$ are chosen in the first $h+1$ iterations.
    
        Then, the remaining candidate in $\{p, d\}$ can have at most $3 h t$ points from the first or sixth group of voters, $\alpha \cdot (2 h T+4 t)$ points from the fourth group of voters and $\beta h T \leq \alpha h T$ points from the fifth group of voters, which gives a total score of at most $\alpha \cdot (2 h T + 4 t) + \alpha h T + 3 h t = 3 \alpha h T + 4 \alpha t + 3 h t$.
    
        Any remaining set candidate, on the other hand, will have $T$ points from the second group of voters and at least $3 \alpha h T$ points from the third group of voters in each of the remaining iterations, which gives a total score of at least $T + 3 \alpha h T$.
        
        Therefore, with $T = h t^3$, $h \geq 1 > \alpha$ and $t \geq 3$, we can show that any remaining set candidate has a strictly higher score than the remaining candidate in $\{p, d\}$:
        \begin{align*}
            & (T + 3\alpha h T) - (3\alpha h T + 4\alpha t + 3 h t) \\
            = \quad & T - (4\alpha t + 3 h t) \geq  h t^3 - 7 h t \geq 9 h t - 7 h t \geq 2. \qedhere
        \end{align*}
    \end{claimproof}

    To summarize, because of \cref{claim:greedy_unit_Thiele_CI_NP__x_first}, candidate $x$ is chosen in the first iteration, followed by a selection of $h$ set candidates, as shown in \cref{claim:greedy_unit_Thiele_CI_NP__2_to_h+1}. Furthermore, given an exact cover of $\mathcal{U}$ with sets from $\mathcal{M}$, from \cref{claim:greedy_unit_Thiele_CI_NP__independent_set_candidate_first} it follows that, in iterations $2$ to $h+2$, a set candidate who belongs to the exact cover can always be chosen over a set candidate who does not belong to the set cover. Then, as shown in \cref{claim:greedy_unit_Thiele_CI_NP__h+2}, if and only if the selection of $h$ set candidates corresponds to an exact cover of $\mathcal{U}$, \GRRR can choose candidate $p$, while otherwise candidate $d$ must be selected. Finally, due to \cref{claim:greedy_unit_Thiele_CI_NP__h+2_to_3h+2}, all remaining set candidates must be chosen in the remaining iterations. Therefore, the only possible winning committees are $S = \{m_i \mid M_i \in \mathcal{M}\} \cup \{x, y, p\}$ and $S' = \{m_i \mid M_i \in \mathcal{M}\} \cup \{x, y, d\}$ and $S$ is one of the winning committee if and only if there is a size-$h$ exact cover in the corresponding \problemname{RX3C} instance. Now, if we construct a set of candidates $P$, which consists of candidate $p$ and an arbitrary number of set candidates, it holds that there is a winning committee in $\GRRR(E, k)$ that includes all candidates in $P$ if and only if there is an exact cover of the universe set $\mathcal{U}$.

    We first prove the following result.\footnote{The proof is very technical in nature and only serves the purpose of generalizing the result of \cref{prop:cons_unit_Thiele_CI_NP} to all greedy unit-decreasing Thiele rules. However, the desired property is trivially satisfied by \GreedyCC and \GreedyPAV. Therefore, readers who are only interested in these rules can ignore the lemma.}

    \begin{lemma}\label{lem:uct_vector_difference}
    Let $\lambda : \NN^+ \rightarrow [0, 1]$ be an OWA function with $\lambda(1) = 1 > \lambda(2)$. Then there exists an $s \in \NN$ such that $\lambda(s) - \lambda(s+1) > \lambda(s+1) - \lambda(s+2)$ and $\lambda(s) - \lambda(s+1) > \lambda(s+2) - \lambda(s+3)$.
\end{lemma}

\begin{proof}
    Let $\lambda$ be an arbitrary OWA function with $\lambda(1) = 1 > \lambda(2)$ (i.e.\ $\lambda$ induces a (greedy) unit-decreasing Thiele rule), let $\alpha = \lambda(2)$ and let $m = 2\ceil{\frac{1}{1-\alpha}} + 2$. We claim that the desired property is satisfied for some $s \in [1, m-3]$.
    
    For every $i \in \NN^+$ we denote by $d_i = \lambda(i) - \lambda(i+1)$ the \emph{difference value} of two consecutive values of $\lambda$. As per the definition of OWA functions, it must hold that $\lambda(i) \geq \lambda(i+1)$ and therefore $d_i \geq 0$ for every $i \in \NN^+$.

    Assume for contradiction that for every $s \in [1, m-3]$ it holds that $d_s \leq d_{s+1}$ or $d_s \leq d_{s+2}$, i.e.\ $\lambda(s) - \lambda(s+1) \leq \lambda(s+1) - \lambda(s+2)$ or $\lambda(s) - \lambda(s+1) > \lambda(s+2) \leq \lambda(s+3)$. We initialize a variable $t_j$ for every $j \in [1, \frac{m}{2}]$. In each $t_j$ we will store a difference value $d_i$ for some $i \in [m-1]$. We set $t_1 = d_1$. For all subsequent values, we proceed in the following way. Assume that $t_j$ was set to the value of $d_i$ for some $i \in [1, m-3]$. We set $t_{j+1} = \max (d_{i+1}, d_{i+2})$. By doing so, we traverse through the difference values, making either one or two steps in each iteration and we store a subset of the difference values in the $t$ variables, i.e.\ $\{t_j \mid j \in [1, \frac{m}{2}]\} \subseteq \{d_i \mid i \in [1, m-1]\}$. By our assumption, for each difference value $d_i$, either $d_{i+1}$ or $d_{i+2}$ has to be greater or equal to $d_i$. Therefore, it must hold that $t_j \geq t_{j-1}$ and, due to transitivity, $t_j \geq t_1 = d_1$ for every $j \in [2, \frac{m}{2}]$.
    
    We can now give an upper bound to the value of $\lambda(m)$:
    \begin{align}
        \lambda(m) \quad & = \quad \lambda(1) - \sum_{i = 1}^{m-1} \bigbrace{\lambda(i) - \lambda(i+1)} \\
        & = \quad 1 - \sum_{i = 1}^{m-1} d_i \label{eq:uvd_1} \\
        & \leq \quad 1 - \sum_{i = 1}^{\frac{m}{2}} t_i \label{eq:uvd_2} \\
        & \leq \quad 1 - \sum_{i = 1}^{\frac{m}{2}} t_1 \label{eq:uvd_3} \\
        & = \quad 1 - \sum_{i = 1}^{\frac{m}{2}} d_1 \label{eq:uvd_4} \\
        & = \quad 1 - \frac{m}{2} d_1 \\
        & = \quad 1 - \frac{2\ceil{\frac{1}{1-\alpha}}+2}{2} (1-\alpha) \label{eq:uvd_5} \\
        & \leq \quad 1 - \bigbrace{\frac{1}{1-\alpha}+1} (1-\alpha) \\
        & = \quad  \alpha - 1 \\
        & < \quad 0, \label{eq:uvd_6}
    \end{align}
    where in \cref{eq:uvd_1}, we have used the definition of $\lambda(1)$ and of the $d$ values, in \cref{eq:uvd_2}, we have used the observations that the set of all $t$ values is a subset of all $d$ values, in \cref{eq:uvd_3}, we have used the observation that $t_i \geq t_1$ for all $i \in [1, \frac{m}{2}]$, in \cref{eq:uvd_4}, we have used the definition of $t_1$, in \cref{eq:uvd_5}, we have used the definition of $m$ and in \cref{eq:uvd_6}, we have used the fact that $\alpha < 1$.
    
    However, $\lambda(m) \geq 0$ needs to hold, as per the definition of OWA functions, a contradiction. Therefore, there must be an $s \in [1, m-3]$ such that $\lambda(s) - \lambda(s+1) > \lambda(s+1) - \lambda(s+2)$ and $\lambda(s) - \lambda(s+1) > \lambda(s+2) - \lambda(s+3)$.
\end{proof}
    Then, as mentioned previously, we will now show how the construction can be adapted if the assumption $1 - \alpha > \alpha - \beta$ and $1 - \alpha > \beta - \gamma$ is not satisfied. Let $\lambda$ be the underlying OWA function of \GRRR. Since \GRRR is unit-decreasing and because of \cref{lem:uct_vector_difference}, there must be an $s \in \NN$ such that $\lambda(s) - \lambda(s+1) > \lambda(s+1) - \lambda(s+2)$ and $\lambda(s) - \lambda(s+1) > \lambda(s+2) - \lambda(s+3)$. We fix the smallest $s$ for which the condition is satisfied. We then add $s-1$ additional dummy candidates, which are approved by all voters. It is clear that the dummy candidates are selected first. For the subsequent assessments, we will only need to consider the function values $\lambda(s), \lambda(s+1), \lambda(s+2)$ and $\lambda(s+3)$, for which the desired condition is satisfied. This concludes the proof.

\subsection{Proof of \cref{thm:ICE_gudT_all_ell}}
We first prove the following lemma.
\begin{lemma}\label{lem:ICE_gudT_reduce_from_R-CI}
     Let \GRRR be a greedy unit-decreasing Thiele rule. There exists a polynomial-time computable function $f$ that transforms a \GRRR-\CI instance $\Tilde{\mathcal{I}} = (\Tilde{E}, \Tilde{P}, \Tilde{k})$ into a \GRRR-\RCE instance $\mathcal{I} = f(\Tilde{\mathcal{I}}) = (E, E', S, k, \ell)$ with committee size $k = \Tilde{k}+2$, allowed difference between committees $\ell = \Tilde{k}-|\Tilde{P}|$ and arbitrary distance $r \in \NN^+$ between elections $E$ and $E'$, such that $\Tilde{\mathcal{I}} \in \text{\GRRR-\CI} \Leftrightarrow f(\Tilde{\mathcal{I}}) \in \text{\GRRR-\RCE}$.
\end{lemma}
\begin{proof}
    Let \GRRR be a greedy unit-decreasing Thiele rule and let $\mathcal{\Tilde{I}} = (\Tilde{E}, \Tilde{P}, \Tilde{k})$ be an instance of \CI for \GRRR where $\Tilde{E} = (\Tilde{C}, \Tilde{V})$ is an election with a set $\Tilde{C}$ of candidates and a collection $\Tilde{V}$ of voters, $\Tilde{k}$ is the committee size, and $\Tilde{P} \subseteq \Tilde{C}$ is a subset of candidates. Let $\Tilde{m} = |\Tilde{C}|$ be the number of candidates in election $\Tilde{E}$ and assume that $\Tilde{k} < \Tilde{m}$. 
    We construct an instance $\mathcal{I} = (E, E', S, k, \ell)$ of \RCE for \GRRR. We construct the election $E = (C, V)$ with $C = \Tilde{C} \cup B \cup \{x, y\}$ where $B = \{b_i \mid i \in [1, \ell]\}$ is a set of \emph{filler candidates}, we set the committee size $k = \Tilde{k} + 2$ and the allowed difference between committees $\ell$ to $\Tilde{k} - |\Tilde{P}|$. Let $\alpha = \lambda(2)$ where $\lambda$ is the underlying OWA function of \GRRR. Let $t = \ceil{\frac{1}{1-\alpha}} \cdot (\Tilde{n} + 6)$ and let $T = \ceil{\frac{1}{1-\alpha}} \cdot (\Tilde{m} \cdot t + 3)$ be two integers. We introduce the following thirteen voter groups.\footnote{We separate voter groups ten to thirteen from the rest since they come into play only at the final iteration of \GRRR to ensure that the remaining control candidate $x$ or $y$ is selected.}
    
    \begin{enumerate}
        \item For every voter $\Tilde{v} \in \Tilde{V}$, there is one voter who approves all candidates that are approved by $\Tilde{v}$.
        \item For every candidate $\Tilde{p} \in \Tilde{P}$, there are $t$ voters who approve candidate $\Tilde{p}$. 
        \item For every candidate $\Tilde{c} \in \Tilde{C} \setminus \Tilde{P}$, there are $t$ voters who approve the candidates $\Tilde{c}$ and $x$.
        \item For every candidate $\Tilde{c} \in \Tilde{C}$, there are $(\Tilde{m}-1) \cdot t - \Tilde{n}$ voters who approve candidate $\Tilde{c}$. 
        \item For every candidate $b \in B$, there are $t$ voters who approve the candidates $b$ and $y$.
        \item For every candidate $b \in B$, there are $(\Tilde{m}-1) \cdot t$ voters who approve candidate $b$.
        \item There are $t$ voters who approve candidates $x$ and $y$.
        \item There are $ (|\Tilde{P}| - 1) \cdot t + 3$ voters who approve candidate $x$.
        \item There are $(\Tilde{m}-\ell-1) \cdot t + 3$ voters who approve candidate $y$.\\
        \item For every pair of candidates $c, c' \in (\Tilde{C} \cup B)$ with $c \neq c'$, there are $T$ voters who approve candidates $c$ and $c'$.
        \item There are $(\Tilde{k} - 1) \cdot T$ additional voters who approve candidates $x$ and $y$.
        \item There are $(\Tilde{m} + \ell - \Tilde{k}) \cdot T$ additional voters who approve candidate $x$.
        \item There are $(\Tilde{m} + \ell - \Tilde{k}) \cdot T$ additional voters who approve candidate $y$.
    \end{enumerate}
    In the table below%\cref{tab:ICE_gudT_reduce_from_R-CI}
    , we give an overview of the newly added voters, i.e.\ voter groups two to thirteen. For a pair of candidates $c, c' \in C$, the corresponding table entry $\sigma_{c, c'}$ states the number of newly added voters that approve both candidates $c$ and $c'$.\footnote{Let the function $\app_v: C \rightarrow \{0, 1\}$ indicate whether voter $v$ approves some candidate $c \in C$. Then, additionally and not displayed in the table, for every set of candidates $\zeta \subseteq \Tilde{C}$, there are $\sum_{\Tilde{v} \in \Tilde{V}} \prod_{\Tilde{c} \in \zeta} \app_{\Tilde{v}}(\Tilde{c})$ voters in $\Tilde{V}$, i.e.\ the voters in group one, that approve all candidates in $\zeta$. However, since these voters add at most $\Tilde{n}$ points to the score of any given candidate in $\Tilde{C}$, as we will later see, they do not have a meaningful impact on the point difference of a pair of candidates $\Tilde{c} \in \Tilde{C}$ and $c \in C \setminus \Tilde{C}$.} Since every newly added voter approves at most two candidates in $C$ and at most one candidate in $\Tilde{C}$, if $c \neq c'$ and $c$ is selected, candidate $c'$ loses exactly $(1-\alpha) \cdot \sigma_{c, c'}$ points from the newly added voters.  
    \begin{table}  \label{tab:ICE_gudT_reduce_from_R-CI}
        \begin{tabular}{ c || p{1.2cm} | p{1.2cm} | p{0.95cm} | p{1.1cm} | p{1.1cm} }
          & $\Tilde{c}$ &  $\Tilde{p}$ & $b$ & $x$ & $y$  \\
    \hline \hline
    $\Tilde{c}$ & $\Tilde{m} \cdot t - \Tilde{n} + \allowbreak (\Tilde{m}+\ell-1) \cdot T$ & $T$ & $T$ & $t$ & $0$ \\
    \hline
    $\Tilde{p}$ & $T$ & $\Tilde{m} \cdot t - \Tilde{n} + \allowbreak (\Tilde{m}+\ell-1) \cdot T$ & $T$ & $0$ & $0$ \\
    \hline
    $b$ & $T$ & $T$ & $\Tilde{m} \cdot t  + \allowbreak (\Tilde{m}+\ell-1) \cdot T$ & $0$ & $t$ \\
    \hline
    $x$ & $t$ & $0$ & $0$ & $\Tilde{m} \cdot t + 3 + \allowbreak (\Tilde{m}+\ell-1) \cdot T$ & $t + \allowbreak (\Tilde{k}-1) \cdot T$ \\
    \hline
    $y$ & $0$ & $0$ & $t$ & $t + \allowbreak (\Tilde{k}-1) \cdot T$ & $\Tilde{m} \cdot t + 3 + \allowbreak (\Tilde{m}+\ell-1) \cdot T$\\
        \end{tabular}
        \caption{\normalfont
        Overview of approvals of newly added voters (i.e.\ not taking voter group one into account) in election $E$ in the proof of \cref{lem:ICE_gudT_reduce_from_R-CI} where $\Tilde{c} \in \Tilde{C} \setminus \Tilde{P}$, $\Tilde{p} \in \Tilde{P}$ and $b \in B$. For a pair of candidates, the corresponding table entry states the number of newly added voters that approve both candidates. Naturally, the table is symmetric.
    }
    \end{table}

    Let $S = B \cup \Tilde{P} \cup \{x, y\}$ be a size-$k$ committee. We will show that $S$ wins in election $E$, i.e.\ $S \in \RRR(E, k)$. We will then modify $E$ into an election $E'$ such that $\RRR(E', k) = \{ \Tilde{S} \cup \{x, y\} \mid \Tilde{S} \in \RRR(\Tilde{E}, \Tilde{k}) \}$. We will then show that there is a committee $\Tilde{S} \in \RRR(\Tilde{E}, \Tilde{k})$ such that $\Tilde{P} \subseteq \Tilde{S}$ if and only if there is a committee $S' \in \RRR(E', k)$ such that $\dist(S, S') \leq \ell$. \\
    
    Let us first consider how \GRRR operates on election $E$. 

    \begin{claim}\label{claim:ICE_gudT_reduce_from_R-CI__x_first}
        In election $E$, in iteration $1$ of \GRRR, candidates $x$ and $y$ have at least $3$ more points than any other candidate.
    \end{claim}
    \begin{claimproof}
        One can verify that, before any candidate was chosen, candidates $x$ and $y$ each have a total of $(\Tilde{m} + \ell - 1) \cdot T + \Tilde{m} \cdot t + 3$ points, any candidate from $B$ has a total of $(\Tilde{m} + \ell - 1) \cdot T + \Tilde{m} \cdot t$ points and any candidate from $\Tilde{C}$ has at most $(\Tilde{m} + \ell - 1) \cdot T + \Tilde{m} \cdot t$ points. Hence, the statement holds.
    \end{claimproof}

    \begin{claim}\label{claim:ICE_gudT_reduce_from_R-CI__tilP_u_B}
        In election $E$, if candidate $x$ is chosen in the first iteration, in the iterations $2$ to $k-1$ of \GRRR, every candidate in $\Tilde{P} \cup B$ has at least $3$ more points than any candidate in $\{y\} \cup \Tilde{C} \setminus \Tilde{P}$.
    \end{claim}
    \begin{claimproof}
        We show the statement by induction. Because of \cref{claim:ICE_gudT_reduce_from_R-CI__x_first}, candidate $x$ can be chose in the first iteration. Assume that candidate $y$ has not been chosen in the first $i$ iterations for $i \in [1, k-2]$.

        Let us first consider the role of voter groups ten to thirteen. Candidate $y$ has $(\Tilde{m} + \ell - \Tilde{k}) \cdot T$ points from voter group thirteen, and, since candidate $x$ was chosen in the first iteration, $\alpha \cdot (\Tilde{k} - 1) \cdot T$ points from voter group eleven. Thus, candidate $y$ has exactly $(\Tilde{m} + \ell - \Tilde{k}) \cdot T + \alpha \cdot (\Tilde{k} - 1) \cdot T$ points from voter groups eleven and thirteen. Furthermore, by our assumption, candidates from $\Tilde{C} \cup B$ must have been chosen in iterations $2$ to $i$. Thus, every candidate in $\Tilde{C} \cup B$ has exactly $(\Tilde{m} + \ell - 1 - i) \cdot T + \cdot \alpha \cdot T \cdot i$ points from voter group ten. As $i \leq \Tilde{k}-1$, it holds that:
        \begin{align*}
            (\Tilde{m} + \ell - \Tilde{k}) \cdot T + \alpha \cdot (\Tilde{k} - 1) \cdot T 
            \leq 
            (\Tilde{m} + \ell - 1 - i) \cdot T + \alpha T i.
        \end{align*}
        Therefore, if we disregard voter groups ten, eleven and thirteen, this only decreases the point difference of any candidate in $\Tilde{P} \cup B$ and candidate $y$, while it leaves the point difference of any candidate in $\Tilde{P} \cup B$ and any candidate in $\Tilde{C} \setminus \Tilde{P}$ unchanged. Since we are only interested in showing that every candidate in $\Tilde{P} \cup B$ has more points than every candidate in $\Tilde{P} \cup B$, we can safely disregard these voter groups in our following assessment.
        
        Now, after candidate $x$ was chosen, candidate $y$ and every candidate from $\Tilde{C} \setminus \Tilde{P}$ lose $(1-\alpha) \cdot t$ points from the seventh and third group of voters, respectively. Therefore, without taking voter groups ten, eleven and thirteen into account, candidate $y$ has at most $\alpha \cdot t + (\Tilde{m}-1) \cdot t + 3$ points and any candidate from $\Tilde{C} \setminus \Tilde{P}$ has at most $\alpha \cdot t + (\Tilde{m}-1) \cdot t$ points. 

        On the other hand, for every candidate $\Tilde{p} \in \Tilde{P}$, there are $t$ voters in group two and $(\Tilde{m}-1) \cdot t - \Tilde{n}$ voters in group four who only approve $\Tilde{p}$. Thus, every candidate in $\Tilde{P}$ has a score of at least $\Tilde{m} \cdot t - \Tilde{n}$. Furthermore, since we assumed that candidate $y$ has not yet been chosen, every candidate in $B$ has $t$ points from the fifth group of voters and $(\Tilde{m}-1) \cdot t$ points from the sixth group of voters. Thus, every candidate in $B$ has at least $\Tilde{m} \cdot t$ points. 

        In summary, without voter groups ten, eleven and thirteen, in the $(i+1)$-th iteration, any candidate in $\{y\} \cup \Tilde{C} \setminus \Tilde{P}$ has at most $\alpha \cdot t + (\Tilde{m}-1) \cdot t + 3$ points and any candidate in $\Tilde{P} \cup B$ has at least $\Tilde{m} \cdot t - \Tilde{n}$ points. Therefore, with $t = \ceil{\frac{1}{1-\alpha}} \cdot (\Tilde{n} + 6)$, we can show that all candidates in $\Tilde{P} \cup B$ have at least $3$ points more than all candidates in $\{y\} \cup \Tilde{C} \setminus \Tilde{P}$: 
        \begin{align*}
            & \bigbrace{\Tilde{m} \cdot t - \Tilde{n}} - \bigbrace{\alpha \cdot t + (\Tilde{m}-1) \cdot t + 3} \\
            = \quad & (1-\alpha) \cdot t - \Tilde{n} - 3 \\
            = \quad & (1-\alpha) \cdot \ceil{\frac{1}{1-\alpha}} \cdot (\Tilde{n} + 6) - \Tilde{n} - 3 \\
            \geq \quad & \Tilde{n} + 6 - \Tilde{n} - 3 \\
            = \quad & 3.
        \end{align*}
        Therefore, a candidate from $\Tilde{P} \cup B$ must be chosen in the $(i+1)$-th iteration, and the induction hypotheses -- candidate $y$ is not chosen in the first $i+1$ iterations -- also holds in the next iteration. 
    \end{claimproof}

    \begin{claim}\label{claim:ICE_gudT_reduce_from_R-CI__y_last}
        In election $E$, if candidate $x$ is chosen in the first iteration, in iteration $k$ of \GRRR, candidate $y$ has at least $3$ more points than any candidate in $\Tilde{C} \setminus \Tilde{P}$.
    \end{claim}
    \begin{claimproof}
        From \cref{claim:ICE_gudT_reduce_from_R-CI__x_first},  \cref{claim:ICE_gudT_reduce_from_R-CI__tilP_u_B} and since $|B \cup \Tilde{P}| = k-1$ it follows that all candidates in $\{x\} \cup \Tilde{P} \cup B$ can be chosen in the first $k-1$ iterations. 

        Then, candidate $y$ has exactly $(\Tilde{m} + \ell - \Tilde{k}) \cdot T + \alpha \cdot (\Tilde{k} - 1) \cdot T$ points from voter groups eleven and thirteen. On the other hand, any candidate in $\Tilde{C} \setminus \Tilde{P}$ has at most $\Tilde{n}$ points from voter group one, $\alpha \cdot t$ points from voter group three, $(\Tilde{m}-1) \cdot t - \Tilde{n}$ points from voter group four, and, since $\Tilde{k}$ candidates from $\Tilde{C} \cup B$ were selected in the previous iterations, exactly $(\Tilde{m} + \ell - \Tilde{k} - 1) \cdot T + \alpha \cdot \Tilde{k} \cdot T$ points from voter group ten. Thus, any candidate in $\Tilde{C} \setminus \Tilde{P}$ has at most $(\Tilde{m} + \ell - \Tilde{k} - 1) \cdot T + \alpha \cdot \Tilde{k} \cdot T + (\Tilde{m}-1) \cdot t + \alpha \cdot t$ points.

        Therefore, with $T = \ceil{\frac{1}{1-\alpha}} \cdot (\Tilde{m} \cdot t + 3)$, we can show that candidate $y$ has at least $3$ points more than all candidates in $\Tilde{C} \setminus \Tilde{P}$: 

        \begin{align*}
            & ((\Tilde{m} + \ell - \Tilde{k}) \cdot T + \alpha \cdot (\Tilde{k} - 1) \cdot T) \\
            & - ((\Tilde{m} + \ell - \Tilde{k} - 1) \cdot T + \alpha \cdot \Tilde{k} \cdot T + (\Tilde{m}-1) \cdot t + \alpha \cdot t) \\ 
            = \quad & ((1-\alpha) \cdot T) - ((\Tilde{m}-1) \cdot t + \alpha \cdot t) \\
            \geq \quad & (1-\alpha) \cdot T - \Tilde{m} \cdot t \\
            = \quad & (1-\alpha) \cdot \ceil{\frac{1}{1-\alpha}} \cdot (\Tilde{m} \cdot t + 3) - \Tilde{m} \cdot t \\
            \geq \quad & \Tilde{m} \cdot t + 3 - \Tilde{m} \cdot t \\
            = \quad & 3. \qedhere
        \end{align*}
    \end{claimproof}
    To summarize the behaviour of \GRRR on election $E$, because of \cref{claim:ICE_gudT_reduce_from_R-CI__x_first}, candidate $x$ can be chosen in the first iteration. Then, because of \cref{claim:ICE_gudT_reduce_from_R-CI__tilP_u_B}, all candidates in $\Tilde{P} \cup B$ must be chosen in the following $\Tilde{k}$ iterations, and, because of \cref{claim:ICE_gudT_reduce_from_R-CI__y_last}, candidate $y$ must be chosen in the final iteration. Hence, $S = B \cup \Tilde{P} \cup \{x, y\}$ wins in election $E$, i.e.\ $S \in \GRRR(E, k)$.

    We will now modify the election $E = (C, V)$ into an election $E' = (C, V')$ by selecting an arbitrary voter from group eight or twelve, who only approves candidate $x$, and removing their approval of $x$.\footnote{Like in the proof of \cref{thm:ICE_ucT_coNP_ell=0_r=1foreverything}, note that we can increase the distance between elections $E$ and $E'$ arbitrarily by adding additional dummy candidates and ensuring that they never have sufficiently many approvals to be part of any winning committee.} Let us consider how \GRRR operates on $E'$. 
    \begin{claim}\label{claim:ICE_gudT_reduce_from_R-CI__y_first}
        In election $E'$, in iteration $1$ of \GRRR, candidate $y$ has at least $1$ more point than any other candidate.
    \end{claim}
    \begin{claimproof}
        In $E'$ candidate $x$ has lost $1$ approval compared to election $E$. Therefore, one can verify that before any candidate was selected, candidate $x$ has a total of $(\Tilde{m} + \ell - 1) \cdot T + \Tilde{m} \cdot t + 2$ points, while candidate $y$ has a total of $(\Tilde{m} + \ell - 1) \cdot T + \Tilde{m} \cdot t + 3$ points, any candidate from $B$ has a total of $(\Tilde{m} + \ell - 1) \cdot T + \Tilde{m} \cdot t$, and any candidate from $\Tilde{C}$ has at most $(\Tilde{m} + \ell - 1) \cdot T + \Tilde{m} \cdot t$ points. Thus, $y$ has at least one more point than any other candidate. 
    \end{claimproof}

    \begin{claim}\label{claim:ICE_gudT_reduce_from_R-CI__tilC}
        In election $E'$, in iterations $2$ to $k-1$ of \GRRR, every candidate in $\Tilde{C}$ has at least $3$ more points than any candidate in $\{x\} \cup B$.
    \end{claim}
    \begin{claimproof}
        We show the statement by induction. Assume that candidate $x$ has not been chosen in the first $i$ iterations for $i \in [1, k-2]$. 

        With a similar argument to the one given in the proof of \cref{claim:ICE_gudT_reduce_from_R-CI__tilP_u_B}, we can safely disregard voter groups ten, eleven and twelve, as this only decreases the point difference of any candidate in $\Tilde{C}$ and candidate $x$, while it leaves the point difference of any candidate in $\Tilde{C}$ and any candidate in $B$ unchanged.
        
        Then, any candidate in $\Tilde{C}$ has $t$ points from the second or third group of voters and $(\Tilde{m}-1) \cdot t - \Tilde{n}$ points from the fourth group of voters. Thus, without taking voter groups ten, eleven and twelve into account, any candidate in $\Tilde{C}$ has at least $\Tilde{m} \cdot t - \Tilde{n}$ points. On the other hand, after candidate $y$ was chosen, candidate $x$ and every candidate from $B$ lose $(1-\alpha) \cdot t$ points from the seventh and fifth group of voters, respectively. Thus, without taking voter groups ten, eleven and twelve into account, any candidate from $\{x\} \cup B$ has at most $\alpha \cdot t + (\Tilde{m}-1) \cdot t + 2$ points in the $i$-th iteration. As shown in \cref{claim:ICE_gudT_reduce_from_R-CI__tilP_u_B}, it holds that $\bigbrace{\Tilde{m} \cdot t - \Tilde{n}} - \bigbrace{\alpha \cdot t + (\Tilde{m}-1) \cdot t + 3} \geq 3$. Consequently, any candidate from $\Tilde{C}$ has at least $3$ more points than any candidate from $\{x\} \cup B$. Thus, a candidate from $\Tilde{C}$ must be chosen in the $(i+1)$-th iteration, and the induction hypotheses -- candidate $x$ is not chosen in the first $i+1$ iterations -- also holds for the next iteration. 
    \end{claimproof}

    \begin{claim}\label{claim:ICE_gudT_reduce_from_R-CI__x_last}
        In election $E'$, in iteration $k$ of \GRRR, candidate $x$ has at least $3$ more points than any remaining candidate in $\Tilde{C} \cup B$.
    \end{claim}
    \begin{claimproof}
        From \cref{claim:ICE_gudT_reduce_from_R-CI__y_first} and \cref{claim:ICE_gudT_reduce_from_R-CI__tilC} it follows that candidate $x$ and a selection of $\Tilde{k}$ candidates from $\Tilde{C}$ must be chosen in the first $k-1$ iterations.

        Then, candidate $x$ has at least $(\Tilde{m} + \ell - \Tilde{k}) \cdot T + \alpha \cdot (\Tilde{k} - 1) \cdot T$ points from voter groups eleven and twelve. 

        On the other hand, any remaining candidate in $C$ has at most $\Tilde{n}$ points from voter group one, $t$ points from voter group two or three and $(\Tilde{m}-1) \cdot t - \Tilde{n}$ points from voter group four. Similarly, any candidate from $B$ has at most $t$ points from voter group five and $(\Tilde{m}-1) \cdot t$ points from voter group six. Furthermore, since $\Tilde{k}$ candidates from $\Tilde{C}$ were selected in the previous iterations, any candidate from $\Tilde{C} \cup B$ has exactly $(\Tilde{m} + \ell - \Tilde{k} - 1) \cdot T + \alpha \cdot \Tilde{k} \cdot T$ additional points from voter group ten. Thus, any candidate in $\Tilde{C} \cup B$ has at most $(\Tilde{m} + \ell - \Tilde{k} - 1) \cdot T + \alpha \cdot \Tilde{k} \cdot T + \Tilde{m} \cdot t$ points.

        As show in \cref{claim:ICE_gudT_reduce_from_R-CI__y_last}, it holds that 
        \begin{align*}
            & ((\Tilde{m} + \ell - \Tilde{k}) \cdot T + \alpha \cdot (\Tilde{k} - 1) \cdot T) \\
            & - ((\Tilde{m} + \ell - \Tilde{k} - 1) \cdot T + \alpha \cdot \Tilde{k} \cdot T + (\Tilde{m}-1) \cdot t + \alpha \cdot t) \\
            \geq \quad & 3.
        \end{align*}
        Consequently, it follows that candidate $x$ has at least $3$ more points than any remaining candidate in $\Tilde{C} \cup B$. 
    \end{claimproof}

    To summarize the behaviour of \GRRR on election $E'$, because of \cref{claim:ICE_gudT_reduce_from_R-CI__y_first}, candidate $y$ must be chosen in the first iteration, because of \cref{claim:ICE_gudT_reduce_from_R-CI__tilP_u_B}, a selection of $\Tilde{k}$ candidates from $\Tilde{C}$ must be chosen in following $\Tilde{k}$ iterations and, because of \cref{claim:ICE_gudT_reduce_from_R-CI__x_last}, candidate $x$ must be chosen in the final iteration. 
    
    Furthermore, as $x$ will not be selected before the final iteration, all candidates in $\Tilde{C}$ have the same combined number of points from voter groups two, three, four and ten in iterations $2$ to $k-1$. Thus, the scores of the candidates in $\Tilde{C}$ only differ by the approvals of voter group one. Since voter group one directly corresponds to the voters in $\Tilde{V}$, \GRRR chooses exactly those size-$\Tilde{k}$ subsets of $\Tilde{C}$ that are winning committees in election $\Tilde{E}$. Therefore, it holds that $\RRR(E', k) = \{ \Tilde{S} \cup \{x, y\} \mid \Tilde{S} \in \RRR(\Tilde{E}, \Tilde{k}) \}$.
    
    Then, the only possible intersection between the committee $S$ and a winning committee in $E'$ are the candidate in $\Tilde{P}$, as well as candidates $x$ and $y$. Therefore, there is a committee $S' \in \RRR(E', k)$ such that $\dist(S, S') \leq \ell = \Tilde{k} - |\Tilde{P}|$ if and only if all candidates from $\Tilde{P}$ are included in $S'$, which holds if and only if there is a committee $\Tilde{S} \in \RRR(\Tilde{E}, k)$ such that $\Tilde{P} \subseteq \Tilde{S}$.

    We conclude by noting that the committee size $k$ for our constructed election $E$ only changed by a constant factor in regards to the committee size $\kappa$ of the given election $\Tilde{E}$. Therefore, we follow the notion of a parameterized reduction.
\end{proof}

Given a greedy unit-decreasing rule \GRRR, from \cref{lem:ICE_gudT_reduce_from_R-CI} and \cref{prop:cons_unit_Thiele_CI_NP}, it can be followed that \RCE for \GRRR is \NP-hard for all values of $\ell \in [k-3]$. 
Furthermore, in the proof of \cref{lem:ICE_gudT_reduce_from_R-CI}, given a \CI instance $(\Tilde{E}, \Tilde{P}, \Tilde{k})$, where $\Tilde{E} = (\Tilde{C}, \Tilde{V})$ is an election, $\Tilde{P} \subseteq \Tilde{C}$ is a subset of candidates and $\Tilde{k}$ is the committee size, we constructed an \RCE instance $(E, E', S, k, \ell)$, where $E = (C, V)$ and $E' = (C, V')$ are two elections, $k = \Tilde{k}+2$ is the committee size, $S \in \GRRR(E, k)$ is a winning committee and $\ell = \Tilde{k} - |\Tilde{P}|$ is the allowed difference between committees. We then constructed voters in election $E$, such that, for the construction of committee $S$, some control candidate $x$ was chosen in the first iteration of \GRRR, while some other control candidate $y$ was chosen in the final iteration. On the other hand, we constructed voters in election $E'$ such that, for every winning committee $S' \in \GRRR(E', k)$, candidate $y$ must be chosen in the first iteration and candidate $x$ must be chosen in the last iteration. Therefore, it holds that candidates $x$ and $y$ are part of committee $S$ and every winning committee $S' \in \GRRR(E', k)$, and it is guaranteed that $\dist(S, S') \leq k-2$. However, if we reduce the committee size by one, i.e., we set $k$ to $\Tilde{k}+1$, the committee $S \setminus \{y\}$ wins in election $E$, while in election $E'$, control candidate $x$ can not be part of any winning committee. Thus, we can alter the construction in the proof \cref{prop:cons_unit_Thiele_CI_NP} so that $k = \Tilde{k} + 1$ and $\ell = \Tilde{k} - |\Tilde{P}| + 2$. Therefore, we can follow that \RCE for \GRRR is \NP-hard for all values of $\ell \in [k-1]$.

\subsection{Proof of \cref{thm:ICE_cons_unit_Thiele_W[1]_k-ell+r}}
We provide a reduction from \textsc{Independent Set With Forced Vertices (\ISWFV)}, a variant of the  \textsc{Independent Set (\IS)} problem which is well-known to be \NP-hard \cite{garey_computers_1979} and \Wone-hard, parameterized by the solution size $\kappa$ \cite{downey_fundamentals_2013}.
An instance of \ISWFV comprises of a  graph $G = (W, F)$ with a set $W$ of vertices and a set $F$ of edges, a non-negative integer $\kappa \in \NN$ and a set $Q \subseteq W$ of forced vertices; it is a yes-instance if there is an independent set $I \subseteq W$ with $|I| = \kappa$ such that $Q \subseteq I$, and a no-instance otherwise.

We will first show that \ISWFV is \Wone-hard with respect to $\kappa$, for any fixed $|Q|$, with a reduction from \IS.

\begin{lemma}\label{lem:ISWFV_W[1]-hard}
    For every fixed size of $|Q|$, \ISWFV is \Wone-hard, parameterized by the solution size $\kappa$.
\end{lemma}
\begin{proof}
    Note that in the case of $|Q| = 0$, \problemname{ISWFV} is identical to the \problemname{IS} problem. Furthermore, for every fixed size of $|Q|$, we can easily reduce from \problemname{IS} to \problemname{ISWFV}. Given a graph $G$ and a solution size $\kappa$, we add $|Q|$ vertices to $G$, which do not have any edges, and we increase the solution size $\kappa$ by $|Q|$. Then, there is an independent set $I$ of size $\kappa$ in the original graph if and only if there is an independent set of size $\kappa + |Q|$ in the new graph, i.e.\ $I$ and the newly added vertices. Notably, since we only increased the solution size by a constant (we fixed the size of $|Q|$), we follow the notion of a parameterized reduction. 
\end{proof}
Next, given a greedy unit-decreasing Thiele rule \GRRR, we will present a parameterized reduction from \ISWFV to \RCE for \GRRR.

The general idea is as follows. We are given an instance $(G = (W, F), Q, \kappa)$ of \ISWFV, and we ask if there is a size-$\kappa$ independent set in $G$ that includes all candidates in $Q$. We construct an election $E = (C, V)$ that consists of a \emph{vertex candidate} for every vertex in $W$, a \emph{dummy candidate} $x$, as well as two candidates $\varphi$ and $\psi$. We set the committee size $k$ to $\kappa + 2$. We denote by $C_Q$ the set of vertex candidates that correspond to the vertices in $Q$. Inspired by \citet{aziz_computational_2015}, we add a large group of voters for every edge in $F$, approving the two vertex candidates that correspond to the edge, who encode the given \ISWFV instance into the election. We say that a vertex candidate is \emph{independent} in a given iteration if the corresponding vertex does not have an edge to a vertex that corresponds to a previously selected vertex candidate. We construct voters such that candidate $x$ is chosen in the first iteration as a means to balance the points of the remaining candidates. In each of the following $\kappa$ iterations, \GRRR chooses a remaining independent vertex candidate if one exists. Then, if and only if $\kappa$ independent vertex candidates have been chosen (i.e.\ vertex candidates that correspond to a size-$\kappa$ independent set in $G$), \GRRR can choose candidate $\varphi$ in the final iteration. However, if in any but the final iteration, there is no independent vertex candidate remaining, \GRRR must choose candidate $\psi$ instead and needs to select vertex candidates in the remaining iterations, i.e.\ candidate $\varphi$ will not be chosen. Therefore, all winning committees consist of candidate $x$, a selection of $\kappa$ vertex candidates and either candidate $\varphi$ or candidate $\psi$. Furthermore, for a selection $I \subseteq W$ of $\kappa$ vertices, the committee that consists of all corresponding vertex candidates, as well as candidates $x$ and $\varphi$, is winning in $E$ if and only if $I$ is an independent set in $G$. Therefore, there is a size-$\kappa$ independent set $I$ in $G$ such that $Q \subseteq I$ if and only if there is a winning committee $S \in \GRRR(E, k)$ such that $C_Q \cup \{\varphi\} \subseteq S$. 

\begin{lemma}\label{lem:cons_unit_Thiele_CI_W[1]}
    Let \RRR be a consistent and unit-decreasing Thiele rule. For every fixed size of $|P| \geq 1$, \CI for \GRRR is \Wone-hard parameterized by the committee size $k$, even if every voter approves at most $2$ candidates. 
\end{lemma}
\begin{proof}
    We reduce from the \ISWFV problem, for which we have shown \Wone-hardness when parameterized by the solution size $\kappa$ in \cref{lem:ISWFV_W[1]-hard}. Let $G = (W, F)$ be an arbitrary graph with a set $W$ of vertices and a set $F$ of edges, let $\eta = |W|$ be the number of vertices, let $\kappa \in \NN$ be the solution size and let $Q \subseteq W$ with $|Q| \leq \kappa$ be a set of forced vertices. Assume that $\eta-\kappa \geq 2$. We ask if there is an independent set $I \subseteq W$ of size $\kappa$ in $G$ such that $Q \subseteq I$.

    We construct a \CI instance $(E, P, k)$ for \GRRR as follows. Let $C_W = \{c_w \mid w \in W\}$ be a set of \emph{vertex candidates} that correspond to the vertices in $W$. We construct an election $E = (C, V)$ where $C = C_W \cup \{\varphi, \psi, x\}$ and we set the committee size $k = \kappa + 2$. Let $C_Q = \{c_q \mid q \in Q\}$ be the set of candidates that corresponds to the vertices in $Q$ and let $P = C_Q \cup \{\varphi\}$. Let $\alpha = \lambda(2)$, where $\lambda$ is the underlying OWA function of \GRRR. Let $t = \ceil{\frac{1}{1-\alpha}} \cdot \eta$ and let $T = \ceil{\frac{1}{1-\alpha}} \cdot (\kappa \cdot t + 2)$ be two integers. We introduce the following nine voter groups. 
    \begin{enumerate}
        \item For every edge $\{u, w\} \in F$, there are $T$ \emph{edge voters} who approve the candidates $c_u$ and $c_w$.
        \item For every vertex $w \in W$, there is one voter who approves candidates $c_w$ and $\psi$.
        \item For every pair of vertices $u, w \in W$ with $u \neq w$, there are $t$ voters who approve candidates $c_u$ and $c_w$.
        \item For every vertex $w \in W$, there are $(\Delta(G) - \degg(w)) \cdot T + t$ voters who approve candidate $c_w$.
        \item There are $\Delta(G) \cdot T + (\eta - \kappa) \cdot t$ voters who approve candidates $\psi$ and $\varphi$.
        \item There are $\kappa \cdot t$ voters who approve candidates $\psi$ and $x$.
        \item There are $\kappa \cdot t + \kappa$ voters who approve candidates $\varphi$ and $x$.
        \item There are $\eta - \kappa$ voters who approve candidate $\varphi$.
        \item There are $T^5$ voters who approve candidate $x$.
    \end{enumerate} 

    In the table below, %\cref{tab:cons_unit_Thiele_CI_W[1]}
    we give an overview of all voter approvals. For a pair of candidates $c, c' \in C$, the corresponding table entry $\sigma_{c, c'}$ states the number of voters that approve both candidates $c$ and $c'$. Since every voter approves at most two candidates in $C$, if $c \neq c'$ and $c$ is selected, candidate $c'$ looses exactly $(1-\alpha) \cdot \sigma_{c, c'}$ points.\\

\begin{table}  \label{tab:cons_unit_Thiele_CI_W[1]}
        \begin{tabular}{ c || p{1.2cm} | p{1.2cm} | p{0.95cm} | p{1.1cm} | p{1.1cm} }
          \qquad \qquad & $c_w$  &  $c_u$  & $\psi$ & $\varphi$ & $x$  \\
    \hline \hline
    $c_w$ & $\Delta(G) \cdot T + \eta \cdot t + 1$ & $\mathrm{adj}_{w, u} \cdot T + t$ & $1$ & $0$ & $0$ \\
    \hline
    $c_u$ & $\mathrm{adj}_{w, u} \cdot T + t$ & $\Delta(G) \cdot T + \eta \cdot t + 1$ & $1$ & $0$ & $0$ \\
    \hline
    $\psi$ & $1$ & $1$ & $\Delta(G) \cdot T + \eta \cdot t + \eta$ & $\Delta(G) \cdot T + (\eta-\kappa) \cdot t$ & $\kappa \cdot t$ \\
    \hline
    $\varphi$ & $0$ & $0$ & $\Delta(G) \cdot T + (\eta-\kappa) \cdot t$ & $\Delta(G) \cdot T + \eta \cdot t + \eta$ & $\kappa \cdot t + \kappa$  \\
    \hline
    $x$ & $0$ & $0$ & $\kappa \cdot t$ & $\kappa \cdot t + \kappa$ & $\Delta(G) \cdot T + (\eta + \kappa) \cdot t + \kappa + T^5$ \\
        \end{tabular}
        \caption{\normalfont
        Overview of voter approvals in election $E$ in the proof of \cref{lem:cons_unit_Thiele_CI_W[1]} where $c_w, c_u \in C_W$ and $\mathrm{adj}_{w, u}$ is a variable that indicates whether vertices $w$ and $u$ are adjacent, formally $\mathrm{adj}_{w, u} = [\{w, u\} \in F]$. For a pair of candidates, the corresponding table entry states the number of voters that approve both candidates. Naturally, the table is symmetric. Since every voter in $E$ approves at most two candidates, the table captures all voter approvals.
    }
    \end{table}

We will show that every winning committee $S \in \RRR(E, k)$ with $\varphi \in S$ directly corresponds to an independent set in $G$. From this, we will conclude that there is a winning committee $S \in \RRR(E, k)$ with $P \subseteq S$ if and only if there is a size-$\kappa$ independent set $I \subseteq W$ in $G$ with $Q \subseteq I$. \\

    We will first show that any committee that consists of candidates $x$ and $\varphi$ and a selection of $\kappa$ vertex candidates that correspond to an independent set in $G$ wins election $E$. Specifically, we will show that it is possible to first select candidate $x$, followed by the vertex candidates corresponding to the size-$\kappa$ independent set and finally candidate $\varphi$. To structure the proof, we will make a number of claims that we will prove individually. 

    \begin{claim}\label{claim:x_first}
        In election $E$, in the first iteration of \GRRR, candidate $x$ is chosen.
    \end{claim}
    \begin{claimproof}
        One can easily verify that $x$ has a higher score than any other candidate in the first iteration. 
    \end{claimproof}

    Now, since we have seen that candidate $x$ must be chosen in the first iteration, we will show that it is possible to select $\kappa$ vertex candidates that correspond to an independent set in the following $\kappa$ iterations. For this, we first show \cref{claim:psi_>_phi}, which we will use in the proof of \cref{claim:ind_vertex_geq_psi}. \cref{claim:psi_>_phi} will also become more relevant later in the reduction proof.

    \begin{claim}\label{claim:psi_>_phi}
        Let $i \in [2, \kappa+1]$ and assume that candidates $\varphi$ and $\psi$ have not been selected in the first $i-1$ iterations of \GRRR in election $E$. Then candidate $\psi$ has a strictly higher score than candidate $\varphi$ in the $i$-th iteration.
    \end{claim}
    \begin{claimproof}
        As shown in \cref{claim:x_first}, candidate $x$ is chosen in the first iteration. Since the voters in group five approve both candidates $\varphi$ and $\psi$, they do not need to be considered when assessing the score difference of the two candidates. Then, not taking into account voters from group five, candidate $\varphi$ has $\eta-\kappa + \alpha \cdot \kappa \cdot t + \alpha \cdot \kappa$ points from voter groups seven and eight. On the other hand, candidate $\psi$ has at least $\eta-\kappa+1 + \alpha \cdot (\kappa-1)$ points from the second group of voters, since at most $\kappa-1$ vertex candidates have been selected, and exactly $\alpha \cdot \kappa \cdot t$ points from the sixth group of voters. Therefore, with $\alpha < 1$, we can show that $\psi$ has strictly more points than $\varphi$.
        \begin{align*}
            & \bigbrace{\eta-\kappa + \alpha \cdot (\kappa-1) + \alpha \cdot \kappa \cdot t + 1} \\
            & - \bigbrace{\eta-\kappa + \alpha \cdot \kappa \cdot t + \alpha \cdot \kappa} \\
            = \quad & \bigbrace{\alpha \cdot (\kappa-1) + 1} - \bigbrace{\alpha \cdot \kappa} = 1 - \alpha > 0. \qedhere
        \end{align*}
    \end{claimproof}
    \begin{claim}\label{claim:ind_vertex_geq_psi}
        Let $i \in [2, \kappa+1]$ and assume that candidates $\varphi$ and $\psi$ have not been selected in the first $i-1$ iterations of \GRRR in election $E$. Let $c_w, c_u \in C_W$ be two vertex candidates that were not selected in the first $i-1$ iterations and such that $c_w$ is independent. Then, in iteration $i$, candidate $c_w$ has a weakly higher score than candidate $c_u$ and a strictly higher score than candidates $\varphi$ and $\psi$.
    \end{claim}
    \begin{claimproof}
        Since $w$ does not share an edge with any $v \in W$ for a previously selected vertex candidate, candidate $c_w$ has exactly $\Delta(w) \cdot T + t$ points from voter groups one and four. Additionally, $c_w$ has one point from voter group two, and at least $(\eta-\kappa) \cdot t + \alpha \cdot (\kappa - 1) \cdot t$ points from voter group three, since $x$ is chosen in the first iteration as shown in \cref{claim:x_first} and therefore at most $\kappa-1$ vertex candidates can have been selected in the first $i-1$ iterations. Thus, $c_w$ has a total score of at least $\Delta(G) \cdot T + (\eta-\kappa+1) \cdot t + \alpha \cdot (\kappa - 1) \cdot t + 1$ points. On the other hand, candidate $\psi$ has at most $\eta$ points from voter group two, $\Delta(G) \cdot T + (\eta - \kappa) \cdot t$ points from voter group five and $\alpha \cdot \kappa \cdot t$ points from voter group six, as $x$ was chosen in the first iteration. Thus, candidate $\psi$ has at most $\Delta(G) \cdot T + (\eta - \kappa) \cdot t + \alpha \cdot \kappa \cdot t + \eta$ points. Furthermore, as shown in \cref{claim:psi_>_phi}, candidate $\psi$ has a strictly higher score than candidate $\varphi$ in the $i$-th iteration. Therefore, with $t = \ceil{\frac{1}{1-\alpha}} \cdot \eta$, we can show that $c_w$ has a strictly higher score than $\psi$ and $\varphi$.
        \begin{align*}
            & \bigbrace{\Delta(G) \cdot T + (\eta-\kappa+1) \cdot t + \alpha \cdot (\kappa - 1) \cdot t + 1} \\
            & - \bigbrace{\Delta(G) \cdot T + (\eta - \kappa) \cdot t + \alpha \cdot \kappa \cdot t + \eta} \\
            = \quad & \bigbrace{t + 1} 
            - \bigbrace{\alpha \cdot t + \eta} \\
            = \quad & (1-\alpha) \ceil{\frac{1}{1-\alpha}} \cdot \eta - \eta + 1 \\
            \geq \quad & \frac{1-\alpha}{1-\alpha} \cdot \eta - \eta + 1 \\
            = \quad & 1.
        \end{align*}

        Let us now consider the score of candidate $c_u$. In the $i$-th iteration, all vertex candidates have the same amount of points from voter groups two and three. Furthermore, since $w$ does not share an edge with the corresponding vertex of a previously selected vertex candidate, $c_w$ has the maximal combined amount of $\Delta(G) \cdot T + t$ points from voter groups one and four. Therefore, candidate $c_w$ has a weakly higher score than candidate $c_u$.
    \end{claimproof}

    Now, in the following two claims, we will show that after candidate $x$ and a selection of $\kappa$ vertex candidates have been chosen, it is possible to select candidate $\varphi$.

    \begin{claim}\label{claim:psi_equal_phi_last_it}
        Assume that candidates $\varphi$ and $\psi$ have not been selected in the first $\kappa+1$ iterations of \GRRR in election $E$. Then candidates $\varphi$ and $\psi$ have the same score in iteration $\kappa+2$.
    \end{claim}
    \begin{claimproof}
        As argued in the proof of \cref{claim:psi_>_phi}, the approvals from voter group five do not need to be considered. Then, not taking into account voter group five, candidate $\varphi$ has $\eta-\kappa + \alpha \cdot \kappa \cdot t + \alpha \cdot \kappa$ points from voter groups seven and eight. Furthermore, because of \cref{claim:x_first}, $x$ has been selected in the first iteration and, since $\varphi$ and $\psi$ have not yet been selected, an additional $\kappa$ vertex candidates must have been selected. Then, candidate $\psi$ has exactly $\eta-\kappa + \alpha \cdot \kappa$ points from the second group of voters and exactly $\alpha \cdot \kappa \cdot t$ points from the sixth group of voters. Therefore, candidates $\varphi$ and $\psi$ have the same amount of points. 
    \end{claimproof}

    \begin{claim}\label{claim:phi_psi_>_vertex_last_it}
        Assume that neither candidates $\varphi$ and $\psi$ have not been selected in the first $\kappa+1$ iterations of \GRRR in election $E$. Then, in iteration $\kappa+2$, both candidates $\varphi$ and $\psi$ have a strictly higher score than any remaining vertex candidate.
    \end{claim}
    \begin{claimproof}
        Because of \cref{claim:x_first}, $x$ is chosen in the first iteration. As shown in \cref{claim:psi_equal_phi_last_it}, candidates $\varphi$ and $\psi$ have the same score in iteration $\kappa+2$. Therefore, it is sufficient to show that $\varphi$ has a strictly higher score than any remaining vertex candidate. After candidate $x$ has been selected, candidate $\varphi$ has exactly $\Delta(G) \cdot T + (\eta - \kappa) \cdot t$ points from voter group five, $\alpha \cdot (\kappa \cdot t + \kappa)$ points from voter group seven and $\eta-\kappa$ points from voter group eight. Thus, candidate $\varphi$ has a total of $\Delta(G) \cdot T + (\eta - \kappa) \cdot t + \alpha \cdot (\kappa \cdot t + \kappa) + \eta-\kappa$ points. On the other hand, any remaining vertex candidate $c_w$ has at most $\Delta(G) \cdot T + t$ points from voter groups one and four, $1$ point from voter group two and, since $\kappa$ vertex candidates must have been chosen in the first $\kappa+1$ iterations, exactly $(\eta-\kappa-1) \cdot t + \alpha \cdot \kappa \cdot t$ points from voter group three. Thus, candidate $c_w$ has a total of at most $\Delta(G) \cdot T + (\eta-\kappa) \cdot t + \alpha \cdot \kappa \cdot t + 1$ points. Then, with $\eta - \kappa \geq 2$, we can show that $\varphi$ have a strictly higher score than $c_w$.
        \begin{align*}
            & \bigbrace{\Delta(G) \cdot T + (\eta - \kappa) \cdot t + \alpha \cdot (\kappa \cdot t + \kappa) + \eta-\kappa} \\
            & - \bigbrace{\Delta(G) \cdot T + (\eta-\kappa) \cdot t + \alpha \cdot \kappa \cdot t + 1} \\
            = \quad & \alpha \cdot \kappa + \eta-\kappa - 1 \\
            \geq \quad & \eta - \kappa - 1 \\
            \geq \quad & 1. \qedhere
        \end{align*}
    \end{claimproof}

    Finally, we utilize the previous claims to show that any committee that consists of candidates $x$ and $\varphi$, as well as a selection of $\kappa$ vertex candidates that correspond to an independent set in $G$, is winning in election $E$.

    \begin{claim}\label{claim:IS=>WC}
        Let $I \subseteq W$ be a size-$\kappa$ independent set in $G$, let $C_I = \{c_w \mid w \in I\}$ and let $S = C_I \cup \{x, \varphi\}$ be a size-$k$ committee. Then $S$ wins in election $E$, i.e.\ $S \in \RRR(E, k)$. 
    \end{claim}
    \begin{claimproof}
        As shown in \cref{claim:x_first}, $x$ must be chosen in the first iteration. By induction, we can show that it is possible to select all candidates in $C_I$ in the following $\kappa$ iterations. Assume that neither $\varphi$ nor $\psi$ was chosen in the first $i-1$ iterations, for $i \in [\kappa+1]$. Furthermore, let $c_w \in C_I$ be a vertex candidate that has not been chosen in the first $i-1$ iterations (such a candidate must exist, since $|C_I| = \kappa$). Since $I$ is an independent set in $G$, there can be no other vertex candidate $c_v \in C_I$ such that $c_v$ was previously chosen and there is an edge between $w$ and $v$. Then, by \cref{claim:ind_vertex_geq_psi}, it holds that $c_w$ has a strictly higher score than both $\psi$ and $\varphi$ and furthermore, no other set candidate has a higher score than $c_w$. Therefore, it is possible to select $c_w$ in the $i$-th iteration. Finally, due to \cref{claim:psi_equal_phi_last_it} and \cref{claim:phi_psi_>_vertex_last_it}, it is possible to choose $\varphi$ in the final iteration. Therefore, $S \in \RRR(E, k)$. 
    \end{claimproof}

    Next, we will show that any committee that consists of a candidates $x$ and $\varphi$, as well as a selection of vertex candidates that does not correspond to an independent set in $G$, does not win in $E$. For this, we will first address how \GRRR operates on $E$ if in any of the iterations $2$ to $\kappa+1$, there is no remaining independent vertex candidate. 
    \begin{claim}\label{claim:psi_geq_non_ind_vertex}
        Let $i \in [2, \kappa+1]$ and assume that neither $\varphi$ nor $\psi$ was chosen in the first $i-1$ iterations of \GRRR in election $E$. Furthermore, assume that there is no independent vertex candidate remaining. Then, in iteration $i$, candidate $\psi$ has a strictly higher score than all vertex candidates.
    \end{claim}
    \begin{claimproof}
        Because of \cref{claim:x_first}, $x$ was chosen in the first iteration. Then, candidate $\psi$ has at least $\Delta(G) \cdot T + (\eta-\kappa) \cdot t$ points from the fifth group of voters. On the other hand, any remaining vertex candidate $c_w$ has at most $(\Delta(G)-1) \cdot T + \alpha \cdot T + t$ points from voter groups one and four, $1$ point from voter group two and at most $(\eta-1) \cdot t$ points from voter group three. Thus, $c_w$ has a total of at most $(\Delta(G)-1) \cdot T + \alpha \cdot T + \eta \cdot t + 1$ points. Therefore, with $T = \ceil{\frac{1}{1-\alpha}} \cdot (\kappa \cdot t + 2)$, we can show that $\psi$ has a strictly higher score than $c_w$.
        \begin{align*}
            & \bigbrace{\Delta(G) \cdot T + (\eta-\kappa) \cdot t} \\ 
            & - \bigbrace{(\Delta(G)-1) \cdot T + \alpha \cdot T + \eta \cdot t + 1} \\
            = \quad & T - \alpha \cdot T - \kappa \cdot t - 1 \\
            = \quad & (1-\alpha) \ceil{\frac{1}{1-\alpha}} \cdot (\kappa \cdot t + 2) - \kappa \cdot t - 1 \\
            \geq \quad & \frac{1-\alpha}{1-\alpha} \cdot (\kappa \cdot t + 2) - \kappa \cdot t - 1 \\
            = \quad & 1. \qedhere
        \end{align*}
    \end{claimproof}
    Now, we can show the following claim.
    \begin{claim}\label{claim:noIS_noWC} 
        Let $I \subseteq W$ with $|I| = \kappa$ be a set of vertices that is not an independent set in $G$. Let $C_I = \{c_w \mid w \in I\}$ and let $S = C_I \cup \{x, \varphi\}$ be a size-$k$ committee. Then, $S$ does not win in election $E$, i.e.\ $S \notin \RRR(E, k)$.
    \end{claim}
    \begin{claimproof}
     As shown in \cref{claim:x_first}, $x$ must be chosen in the first iteration. Then, since $\psi \notin S$ and because of \cref{claim:psi_>_phi}, $\varphi$ must be chosen in the last iteration. Therefore, in iterations $2$ to $\kappa+1$, all vertex candidates in $C_I$ must be chosen. However, since $I$ is not an independent set in $G$, there must be one iteration in which there are no independent vertex candidates in $C_I$ remaining. Then however, because of \cref{claim:psi_geq_non_ind_vertex}, candidate $\psi$ must have a strictly higher score than all remaining vertex candidates from $C_I$ and therefore, no vertex candidate from $C_I$ can be chosen. A contradiction.
    \end{claimproof}

    From \cref{claim:IS=>WC} and \cref{claim:noIS_noWC} it follows that there is a direct correspondence between size-$\kappa$ independent sets in $G$ and winning committees in $E$ that include candidate $\varphi$. With one additional claim, we can show the correctness of our reduction.

    \begin{claim}\label{claim:not_both_phi_psi}
        There is no winning committee $S \in \RRR(E, k)$ such that $\{\varphi, \psi\} \subseteq S$.
    \end{claim}
    \begin{claimproof}
        Assume for contradiction that there exists a winning committee $S \in \RRR(E, k)$ such that $\{\varphi, \psi\} \subseteq S$. From \cref{claim:psi_>_phi} it follows that $\psi$ needs to be chosen before $\varphi$. Let us consider the iteration in which $\varphi$ is chosen. Then, since $\psi$ has previously been chosen, $\varphi$ has exactly $\alpha \cdot (\Delta(G) \cdot T + (\eta - \kappa) \cdot t)$ points from voter group five, $\eta-\kappa$ points from voter group eight, and, since $x$ was previously chosen as shown in \cref{claim:x_first}, $\alpha \cdot (\kappa \cdot t + \kappa)$ points from voter group seven. Thus, candidate $\varphi$ has a total of $\alpha \cdot (\Delta(G) \cdot T + (\eta - \kappa) \cdot t + \kappa \cdot t + \kappa) + \eta-\kappa$ points.
        
        On the other hand, any vertex candidate $c_w$ has at least $\alpha \cdot \Delta(G) \cdot T + t$ points from voter groups one and four, $1$ point from voter group two and $(\eta-\kappa-1) \cdot t + \alpha \cdot \kappa \cdot t$ points from voter group three. Thus, candidate $c_w$ has a total of at least $\alpha \cdot \Delta(G) \cdot T + (\eta-\kappa) \cdot t + \alpha \cdot \kappa \cdot t + 1$ points. Therefore, with $\eta - \kappa \geq 2$ and $t = \ceil{\frac{1}{1-\alpha}} \cdot \eta$, we can show that $c_w$ has a strictly higher score than $\varphi$.
        \begin{align*}
            & \bigbrace{\alpha \cdot \Delta(G) \cdot T + (\eta-\kappa) \cdot t + \alpha \cdot \kappa \cdot t + 1} \\
            & - \bigbrace{\alpha \cdot (\Delta(G) \cdot T + (\eta - \kappa) \cdot t + \kappa \cdot t + \kappa) + \eta-\kappa} \\
            \geq \quad & \bigbrace{(\eta-\kappa) \cdot t + 1} - \bigbrace{\alpha \cdot (\eta - \kappa) \cdot t + \eta} \\
            = \quad & \bigbrace{(1-\alpha) \cdot (\eta-\kappa) \cdot t + 1} - \bigbrace{\eta} \\
            \geq \quad & (1-\alpha) \cdot \ceil{\frac{1}{1-\alpha}} \cdot \eta + 1 - \eta \\
            \geq \quad & \frac{1-\alpha}{1-\alpha} \cdot \eta + 1 - \eta \\
            \geq \quad & 1.
        \end{align*}
        But then, candidate $\varphi$ can not be chosen in this iteration. A contradiction.
    \end{claimproof}

    Because of \cref{claim:x_first}, \cref{claim:phi_psi_>_vertex_last_it} and \cref{claim:not_both_phi_psi}, every winning committee in $E$ must consist of candidate $x$, exactly $\kappa$ vertex candidates and either candidate $\varphi$ or candidate $\psi$. Now, because of \cref{claim:IS=>WC} and \cref{claim:noIS_noWC}, there is a winning committee $S \in \RRR(E, k)$ with $P = C_Q \cup \{\varphi\} \subseteq S$ if and only if there is a size-$\kappa$ independent set $I \subseteq W$ with $Q \subseteq I$.

    We conclude by noting that the committee size $k$ in our constructed election only changed by a constant factor in regards to the solution size $\kappa$ of the given \ISWFV instance. Therefore, we follow the notion of a parameterized reduction.
\end{proof}

From \cref{lem:ICE_gudT_reduce_from_R-CI} and \cref{lem:cons_unit_Thiele_CI_W[1]}, together with the observation that in the construction of \cref{lem:ICE_gudT_reduce_from_R-CI}, every newly added voter approves at most two candidates, we conclude the proof.

\section{Proofs Omitted from 
\cref{sec:parameterized}}

\subsection{Proof of \cref{pro:ICE_CCAV_FPT_n}}
If $k \leq 2^n$, \cref{pro:ICE_udT_FPT_n+k} offers an \FPT algorithm. Otherwise, we modify the algorithm from \cref{pro:ICE_udT_FPT_n+k} as follows. 
    Note first that, since each voter is required to approve at least one candidate, if $k>n$, a committee in \CCAV-$\RCE(E', k)$ has CCAV score of $n$. 
    %As before, we partition the candidates into $2^n$ classes, so that all candidates in the same class are approved by the same set of voters. 
    We therefore guess a subset of non-empty candidate classes $\{K_1, \dots, K_t\}$, $1\le t\le 2^n$,\footnote{It can be shown that $n$ suffices as an upper bound.} and discard the guess if, by picking one candidate in each class, we fail to obtain a committee with CCAV score of $n$. Otherwise, we create a committee by including one candidate from each $K_j$, picking a candidate in $K_j\cap S$ whenever $K_j\cap S\neq\emptyset$, and then adding $k-t$ additional candidates from $S$. 
    Among all committees obtained in this way, let $S^*$ be one that has the largest intersection with $S$; 
    we return `yes' if $\dist(S^*, S)\le \ell$. There are $2^{2^n}$ guesses to consider and a polynomial amount of work for each guess, so the bound on the running time follows.
    %We guess a selection of at most $n$ non-empty candidate classes and we check if every voter is associated with at least one of the selected classes. If so, we can construct a committee with a maximal score of $n$ by choosing a representative from each of the selected candidate classes and an arbitrary selection of remaining candidates and we prioritize candidates in $S$ whenever possible. It is easy to see that this procedure finds a winning committee in $E'$ that includes a maximum number of candidates from $S$. There are at most $2^{2^n}$ possible selections of candidate classes and all other operations run in polynomial time.
    
\subsection{Proof of \cref{pro:ICE_G-CCAV_FPT_n}}
If $k \leq 2^n$, \cref{pro:ICE_GcudT_FPT_n+k} offers an \FPT algorithm. Otherwise, observe that after at most $n < k$ iterations of Greedy-CC on $E'$ we obtain a committee whose CC score is $n$, i.e., the maximum possible; moreover, as long as the CC score of the current committee is less than $n$, Greedy-CC will pick at most one candidate from each candidate class (as the second candidate's marginal contribution to the CC score will be $0$).
    
    Let $C'\subseteq C$ be a subset of candidates that contains one candidate from each non-empty candidate class; for each class that contains a candidate in $S$, we require that the representative of that class in $C'$ is from $S$ as well. Note that $|C'|\le 2^n < k$.

    We can now iterate through subsets of $C'$. For each subset $S^*\subset C'$, we check whether its CC score is $n$. If yes, we guess a permutation of $S^*$ and check if Greedy-CC can select candidates in $S^*$ in the first $|S^*|$ iterations, in order specified by this permutation; if the answer is yes, we can 
    add $k-|S^*|$ candidates from $S$ to obtain a committee $S'$ in the output of Greedy-CC with $\dist(S', S)= |S^*\setminus S|$. We return `yes' 
    if $|S^*\setminus S|\le \ell$ for some $S^*$ that passes our checks. There are at most $2^{2^n}$ possibilities for $S^*$ and at most $n!$ permutations to consider, so the bound on the running time follows.
    %
   % Otherwise, there can be at most $2^n$ tied candidates in each iteration, since all candidates that belong to the same candidate class always have the same marginal contribution. Furthermore, one can verify that after at most $n$ iterations, every candidate has a marginal contribution of $0$. Therefore, we can stop the iterative process of \GreedyCC after the $n$-th iteration and check if the current selection of candidates can be extended by an arbitrary selection of $k-n$ of the remaining candidates to a committee $S'$ such that $\dist(S, S') \leq \ell$.

\section{Omitted Graphs from \cref{sec:experiments}}

In addition to the two sampling methods introduced in \Cref{sec:experiments}, we consider the following three additional sampling methods:
\begin{itemize}
    \item In the \emph{Resampling Model} (introduced by \citet{szufa2022sampleapproval}), we have two parameters, $p \in [0,1]$ and $\phi \in [0,1]$. %To generate an election with candidate set $C = \{c_1, \dots, c_m\}$ and with $n$ voters, 
    We first choose uniformly at random a central vote $u$ approving exactly $\lfloor p \cdot m \rfloor$ candidates. Then, we generate the votes, considering the candidates one by one independently for each vote. For a vote $v$ and candidate $c$, with probability $1-\phi$ we copy $c$’s approval status from $u$ to $v$ (i.e., if $u$ approves $c$, then so does $v$; if $u$ does not approve $c$ then neither does $v$), and with probability $\phi$ we \enquote{resample} the approval status of $c$, i.e., we let $v$ approve $c$ with probability $p$ (and disapprove it with probability $1-p$). On average, each voter approves $p \cdot m$ candidates.

    \item In the \emph{1D-Euclidean Model with Resampling (1D+Res)}, given a radius $\tau$ and a resampling probability $\phi \in [0, 1]$, we first sample a regular 1D election, and then apply resampling as follows. We operate similarly to the above resampling model, treating each voter as her own central vote. Given a single vote $v$ and a candidate $c$, in the resampling phase, the approval status is copied with probability $1 - \phi$ and resampled with probability $\phi$. In the latter case, $v$ approves of $c$ with probability $\nicefrac{|\mathrm{app}_v|}{m}$, where $\mathrm{app}_v$ is the set of candidates originally approved by $v$, and $m$ is the total number of candidates. In doing so, on average, the total number of approvals stays the same as in the original 1D election.

    \item The \emph{2D-Euclidean Model with Resampling (2D+Res)} is defined anagously to 1D+Res.
\end{itemize}

For both 1D+Res, and 2D+Res, we fix the value if $\phi$ to $0.1$, and for Resampling, we fix the value of $\phi$ to $0.75$. The latter is shown by \citet{szufa2022sampleapproval} to produce elections that resemble real-world ones.

Further, we consider a greater range of average numbers of approvals per vote than before. Whereas earlier, we only sampled elections where, on average, each voter approves
$10$ candidates, we now consider elections with $5$, $10$ and $15$ average approvals per vote. Respectively, for 1D and 1D+Res, this translates to radii of $0.025, 0.051,$ and $0.078$, for 2D and 2D+Res, to radii of $0.134, 0.195,$ and $0.244$, and for Resampling, to values for $\rho$ of $0.05, 0.1$, and $0.15$.

\subsection{Further Trends related to Experiment 1}

We present our results in \Cref{fig:Exp1_1Dallparams,,fig:Exp1_2Dallparams,,fig:Exp1_1D+res,,fig:Exp1_2D+res,,fig:Exp1_Res}.

We observe that, in general, 1D and 2D produce less resilient elections than their counterparts with Resampling. 

In terms of the Resampling model, while the produced elections tend to be more resilient overall, the parameter $\rho$ plays a vital role. In case $\rho = 0.1$, then the central vote approves exactly $10$ candidates, i.e., the committee size. In turn, since all other votes \enquote{inherit} from the central vote, in virtually all cases, this will result in the winning committee of the original election consisting of exactly these $10$ candidates. Further, these candidates have a very high level of support, and across all considered percentage changes, no candidate needs to be replaced under \GreedyPAV, and only in a few cases do candidates need to be replaced under \GreedyCC. On the other hand, for $\rho = 0.05$ and $\rho = 0.15$, there are either too few strong candidates or too many. Hence, the resulting committees are a lot more non-resilient.

\subsection{Further Trends related to Experiment 2}
We present our results in \Cref{fig:Exp2_1Dallparams,,fig:Exp2_2Dallparams,,fig:Exp2_1D+res,,fig:Exp2_2D+res,,fig:Exp2_Res}.

We observe that the behaviour on 1D and 2D elections is very similar to that on 1D+Res and 2D+Res elections, respectively, only that in the latter there are slightly more extreme outliers, especially for 1D+Res under \GreedyCC.

Given the results from experiment $1$, it is unsurprising that Resampling elections with $\rho = 0.1$ have rather small distances, as, on average, the committees under lexicographical in the original adapted elections already have a very low distance to the original winning committee. However, we consider it very surprising that under \GreedyCC, we find outliers of up to value $6$ for $\rho = 0.15$, and up to value $5$ for $\rho = 0.5$, as also for these values, the average distance between the two lexicographical committees is much lower than on 1D and 2D, as well as their counterparts with Resampling.

\subsection{Further Trends related to Experiment 3}

We present our results in \Cref{fig:Exp3_1Dallparams,,fig:Exp3_2Dallparams,,fig:Exp3_1D+res,,fig:Exp3_2D+res,,fig:Exp3_Res}.

There is no significant difference between 1D and 2D, and their Resampling counterparts. 

On the other hand, for Resampling, in the case of $\rho = 0.05$, unsurprisingly the first $5$ candidates do not ever get replaced, as they receive a large amount of support from a majority of the voters. In the case of $\rho = 0.1$, no candidate is ever replaced under \GreedyPAV, but \GreedyCC shows a strong correlation between when a candidate was selected and how often she needs to be replaced. Finally, in the case of $\rho = 0.15$, there is a strong correlation for both voting rules, with the first chosen candidate under \GreedyPAV being the only one that never gets replaced.

\begin{figure}
\centering
\subfloat[\centering \GreedyCC, 1D, $\tau = 0.025$]{{\includegraphics[height=3.3cm]{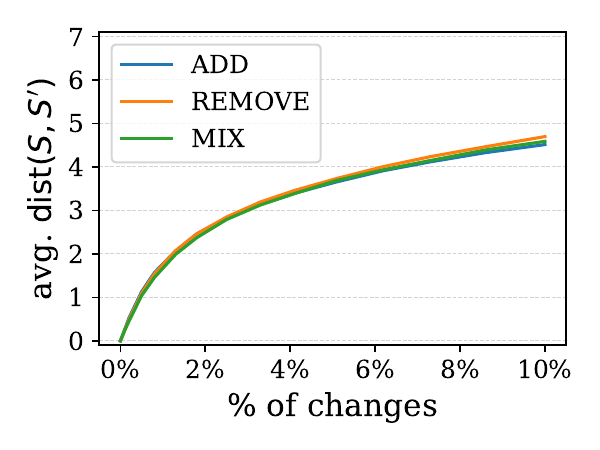}}}
\hspace{.1cm}
\subfloat[\centering \GreedyPAV, 1D, $\tau = 0.025$]{{\includegraphics[height=3.3cm]{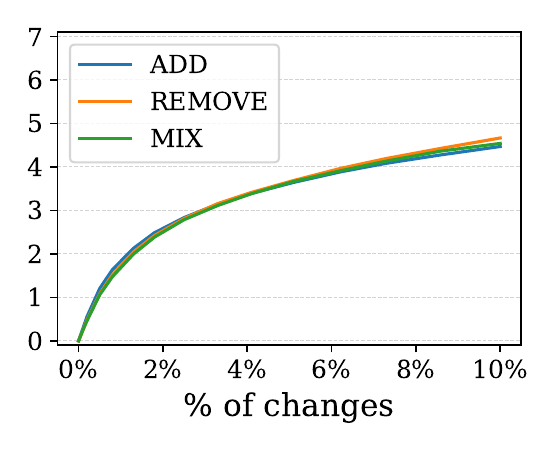}}}

\subfloat[\centering \GreedyCC, 1D, $\tau = 0.051$]{{\includegraphics[height=3.3cm]{figures/EXP1_seqcc_1D_0.051.pdf}}}
\hspace{.1cm}
\subfloat[\centering \GreedyPAV, 1D, $\tau = 0.051$]{{\includegraphics[height=3.3cm]{figures/EXP1_seqpav_1D_0.051.pdf}}}

\subfloat[\centering \GreedyCC, 1D, $\tau = 0.078$]{{\includegraphics[height=3.3cm]{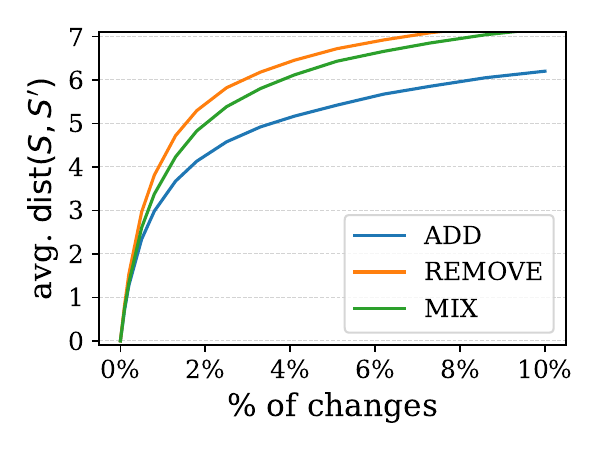}}}
\hspace{.1cm}
\subfloat[\centering \GreedyPAV, 1D, $\tau = 0.078$]{{\includegraphics[height=3.3cm]{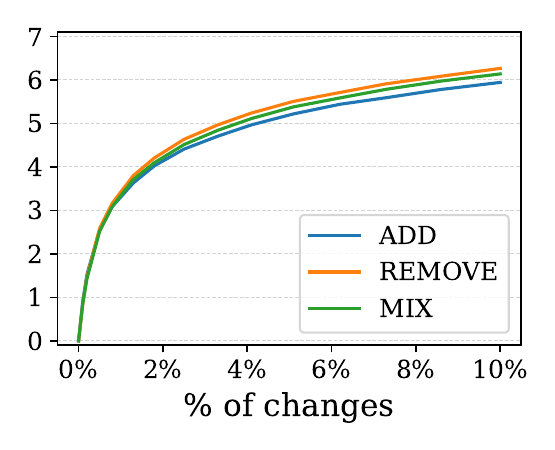}}}
\caption{Results of Experiment $1$ under the 1D model.}
\label{fig:Exp1_1Dallparams}
\end{figure}

\begin{figure}
\subfloat[\centering \GreedyCC, 2D, $\tau = 0.134$]{{\includegraphics[height=3.3cm]{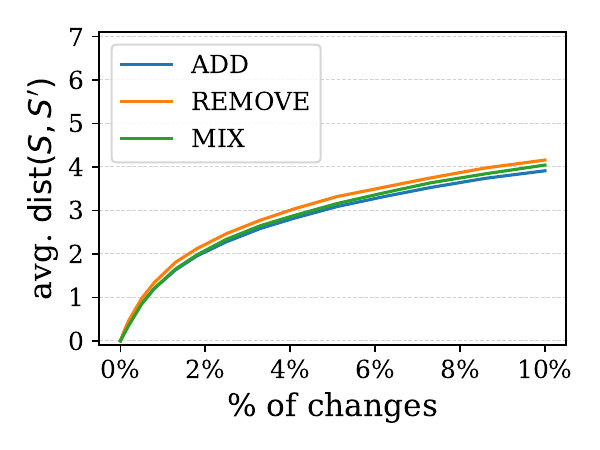}}}
\hspace{.1cm}
\subfloat[\centering \GreedyPAV, 2D, $\tau = 0.134$]{{\includegraphics[height=3.3cm]{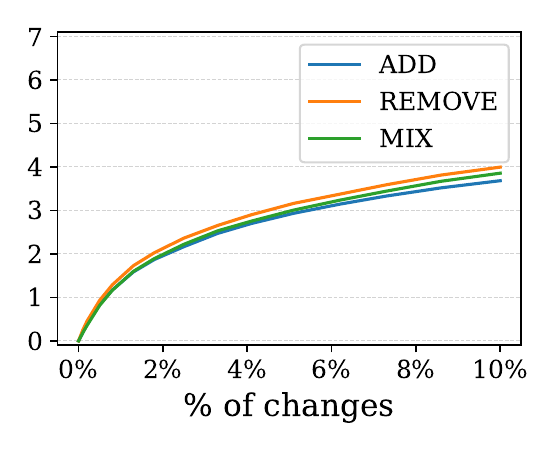}}}

\subfloat[\centering \GreedyCC, 2D, $\tau = 0.195$]{{\includegraphics[height=3.3cm]{figures/EXP1_seqcc_2D_0.195.pdf}}}
\hspace{.1cm}
\subfloat[\centering \GreedyPAV, 2D, $\tau = 0.195$]{{\includegraphics[height=3.3cm]{figures/EXP1_seqpav_2D_0.195.pdf}}}

\subfloat[\centering \GreedyCC, 2D, $\tau = 0.244$]{{\includegraphics[height=3.3cm]{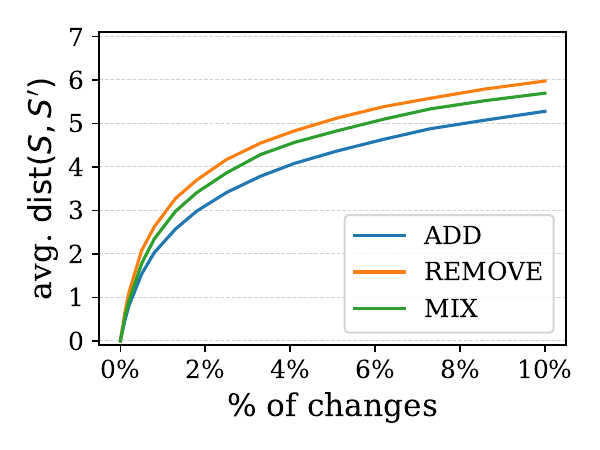}}}
\hspace{.1cm}
\subfloat[\centering \GreedyPAV, 2D, $\tau = 0.244$]{{\includegraphics[height=3.3cm]{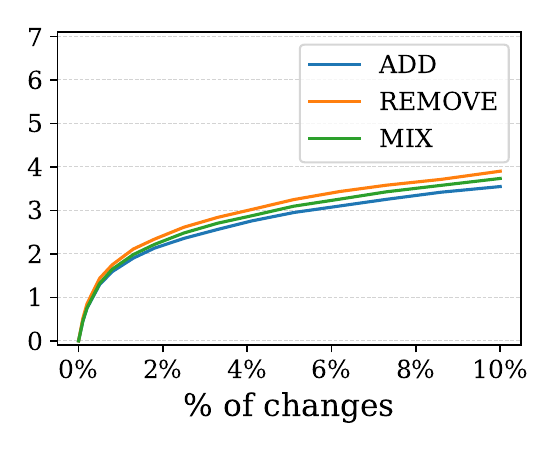}}}
\caption{Results of Experiment $1$ under the 2D model.}
\label{fig:Exp1_2Dallparams}
\end{figure}

\begin{figure}
\centering
\subfloat[\centering \GreedyCC, 1D+Res, $\tau = 0.025$]{{\includegraphics[height=3.3cm]{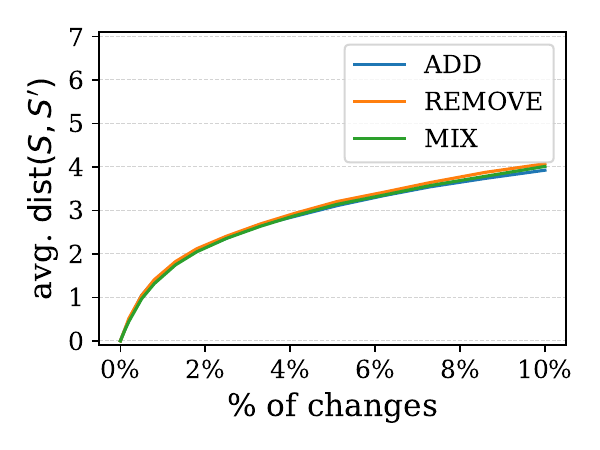}}}
\hspace{.1cm}
\subfloat[\centering \GreedyPAV, 1D+Res, $\tau = 0.025$]{{\includegraphics[height=3.3cm]{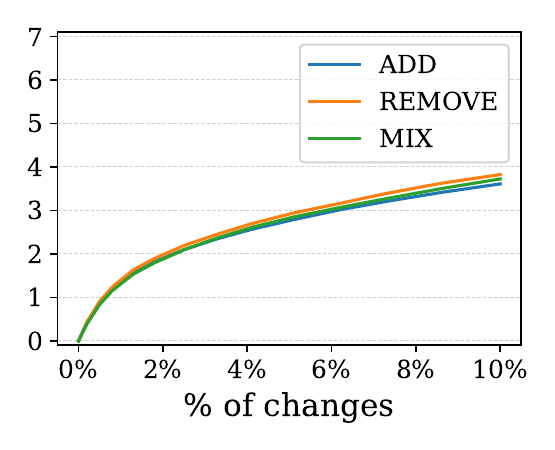}}}

\subfloat[\centering \GreedyCC, 1D+Res, $\tau = 0.051$]{{\includegraphics[height=3.3cm]{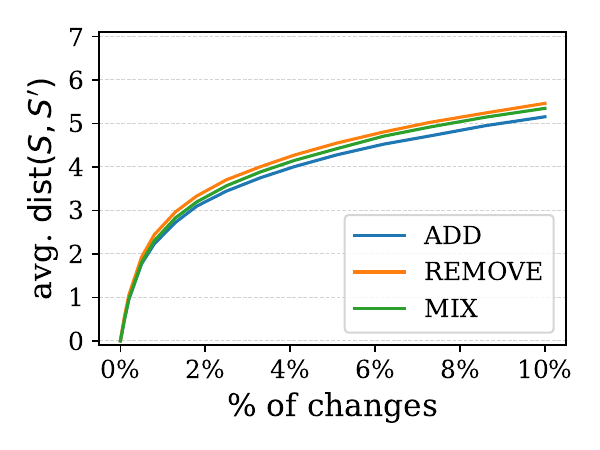}}}
\hspace{.1cm}
\subfloat[\centering \GreedyPAV, 1D+Res, $\tau = 0.051$]{{\includegraphics[height=3.3cm]{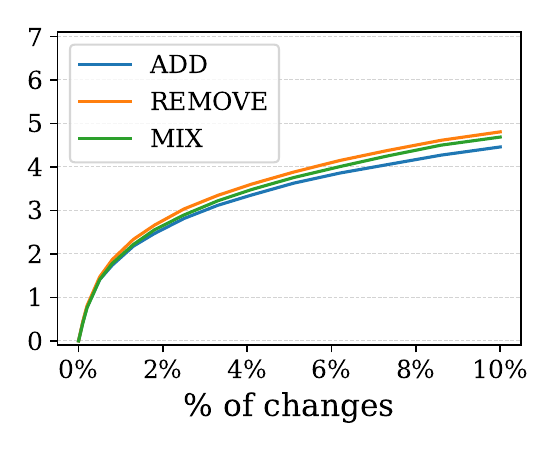}}}

\subfloat[\centering \GreedyCC, 1D+Res, $\tau = 0.078$]{{\includegraphics[height=3.3cm]{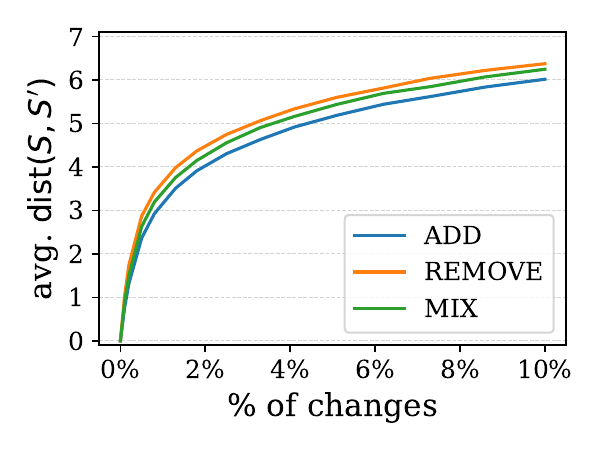}}}
\hspace{.1cm}
\subfloat[\centering \GreedyPAV, 1D+Res, $\tau = 0.078$]{{\includegraphics[height=3.3cm]{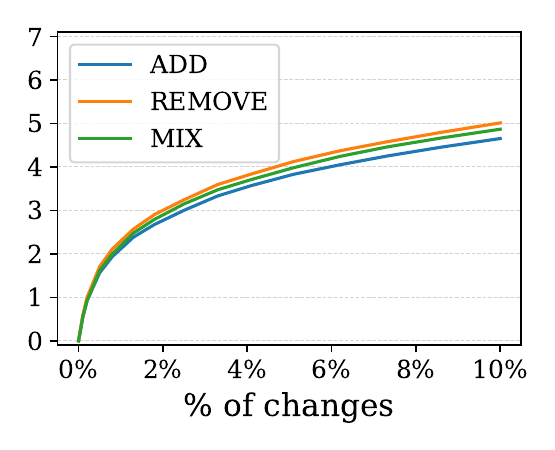}}}
\caption{Results of Experiment $1$ under the 1D+Res model.}
\label{fig:Exp1_1D+res}
\end{figure}

\begin{figure}
\centering
\subfloat[\centering \GreedyCC, 2D+Res, $\tau = 0.134$]{{\includegraphics[height=3.3cm]{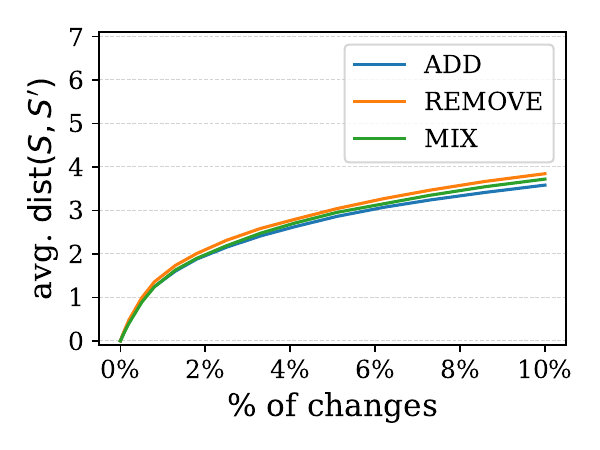}}}
\hspace{.1cm}
\subfloat[\centering \GreedyPAV, 2D+Res, $\tau = 0.134$]{{\includegraphics[height=3.3cm]{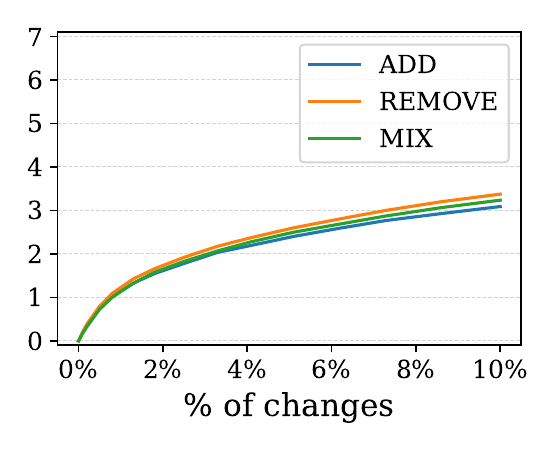}}}

\subfloat[\centering \GreedyCC, 2D+Res, $\tau = 0.195$]{{\includegraphics[height=3.3cm]{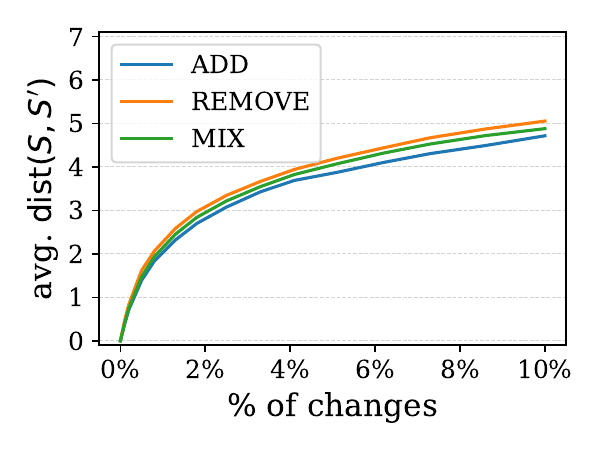}}}
\hspace{.1cm}
\subfloat[\centering \GreedyPAV, 2D+Res, $\tau = 0.195$]{{\includegraphics[height=3.3cm]{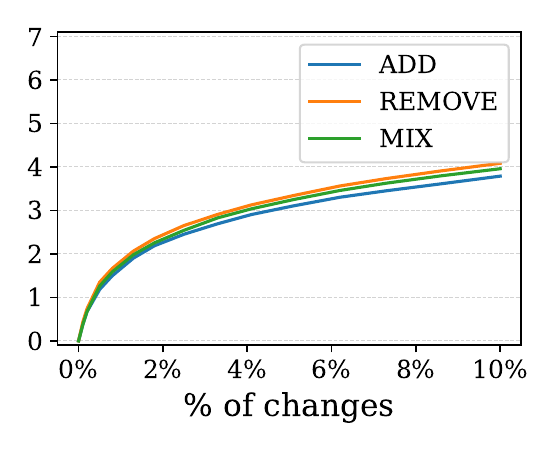}}}

\subfloat[\centering \GreedyCC, 2D+Res, $\tau = 0.244$]{{\includegraphics[height=3.3cm]{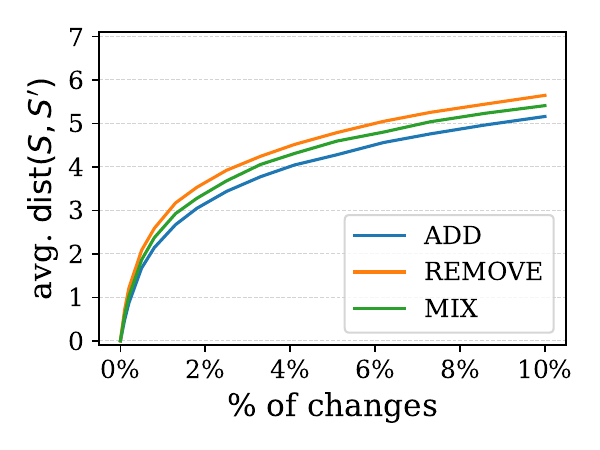}}}
\hspace{.1cm}
\subfloat[\centering \GreedyPAV, 2D+Res, $\tau = 0.244$]{{\includegraphics[height=3.3cm]{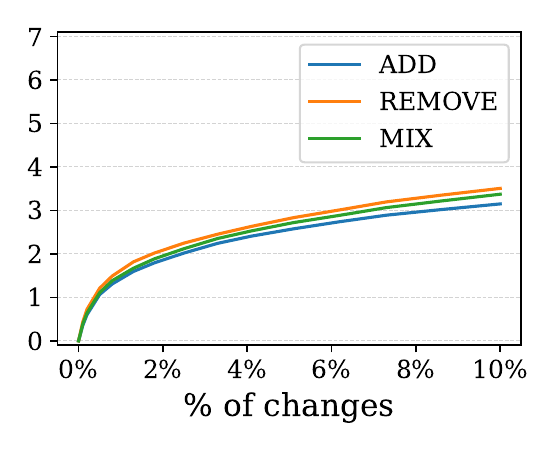}}}
\caption{Results of Experiment $1$ under the 2D+Res model.}
\label{fig:Exp1_2D+res}
\end{figure}

\begin{figure}
\centering
\subfloat[\centering \GreedyCC, Res, $\rho = 0.05, \phi = 0.75$]{{\includegraphics[height=3.3cm]{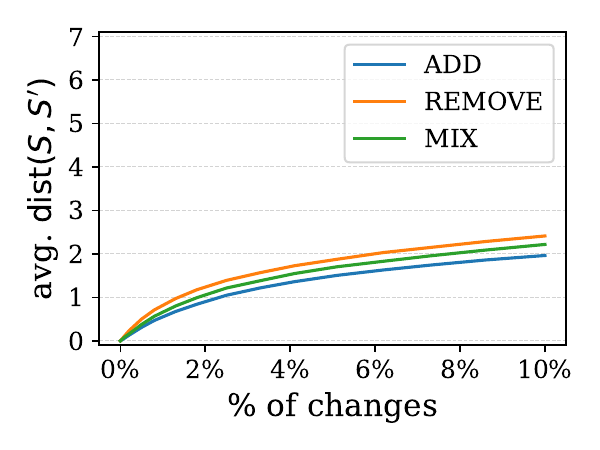}}}
\hspace{.1cm}
\subfloat[\centering \GreedyPAV, Res, $\rho = 0.05, \phi = 0.75$]{{\includegraphics[height=3.3cm]{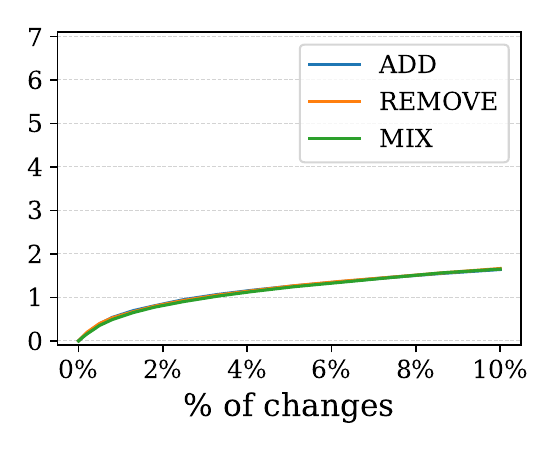}}}

\subfloat[\centering \GreedyCC, Res, $\rho = 0.1, \phi = 0.75$]{{\includegraphics[height=3.3cm]{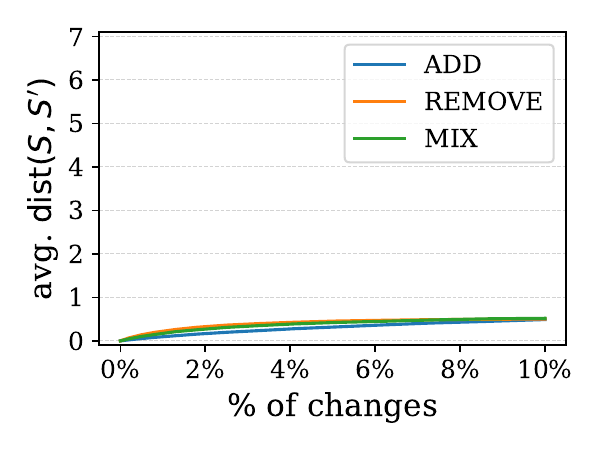}}}
\hspace{.1cm}
\subfloat[\centering \GreedyPAV, Res, $\rho = 0.1, \phi = 0.75$]{{\includegraphics[height=3.3cm]{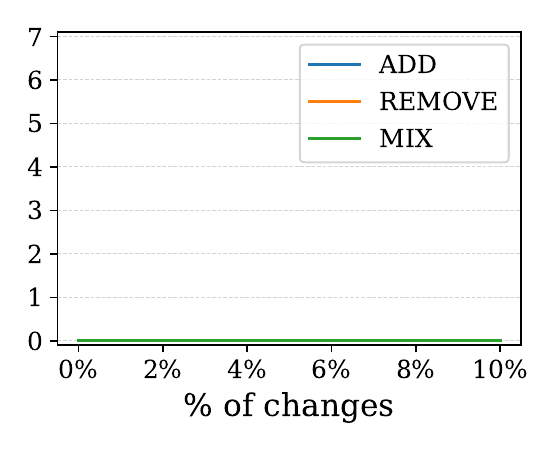}}}

\subfloat[\centering \GreedyCC, Res, $\rho = 0.15, \phi = 0.75$]{{\includegraphics[height=3.3cm]{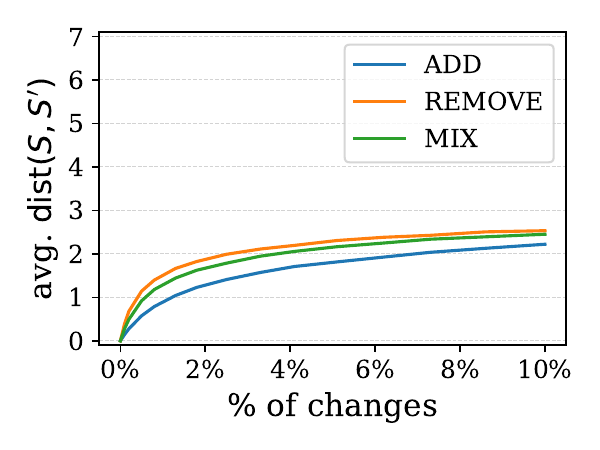}}}
\hspace{.1cm}
\subfloat[\centering \GreedyPAV, Res, $\rho = 0.15, \phi = 0.75$]{{\includegraphics[height=3.3cm]{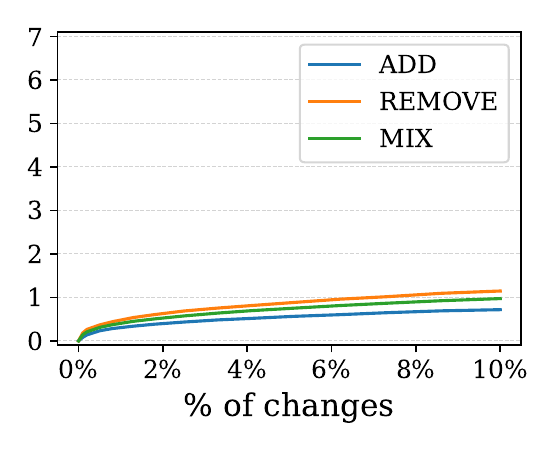}}}
\caption{Results of Experiment $1$ under the Res model.}
\label{fig:Exp1_Res}
\end{figure}

%%%%%%%%%%%%%%%%%%%%%%%%%%%%%%%%%%%%%%%%%%%%%%%%%%%%%

\begin{figure}
\centering
\subfloat[\centering \GreedyCC, 1D, $\tau = 0.025$]{{\includegraphics[height=3.3cm]{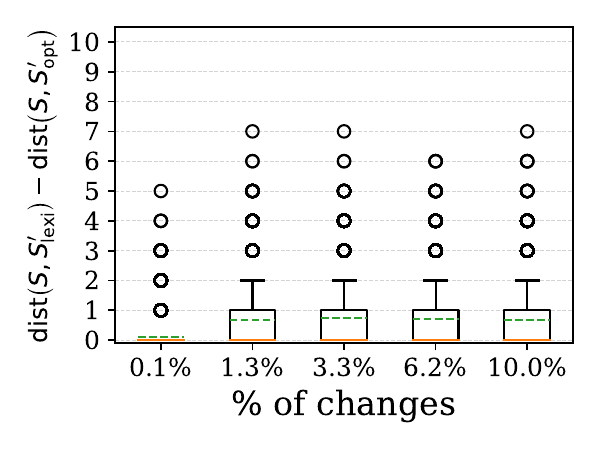}}}
\hspace{.1cm}
\subfloat[\centering \GreedyPAV, 1D, $\tau = 0.025$]{{\includegraphics[height=3.3cm]{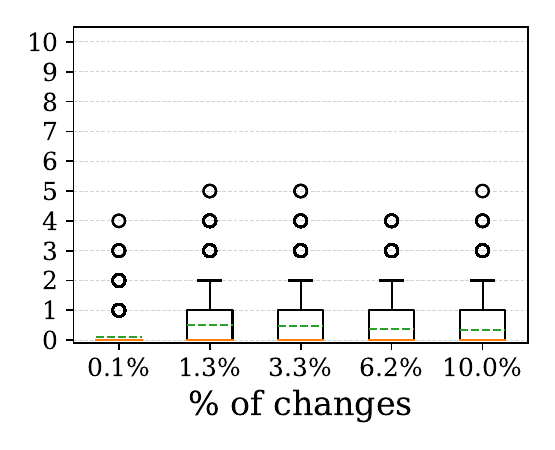}}}

\subfloat[\centering \GreedyCC, 1D, $\tau = 0.051$]{{\includegraphics[height=3.3cm]{figures/EXP2_seqcc_1D_0.051.pdf}}}
\hspace{.1cm}
\subfloat[\centering \GreedyPAV, 1D, $\tau = 0.051$]{{\includegraphics[height=3.3cm]{figures/EXP2_seqpav_1D_0.051.pdf}}}

\subfloat[\centering \GreedyCC, 1D, $\tau = 0.078$]{{\includegraphics[height=3.3cm]{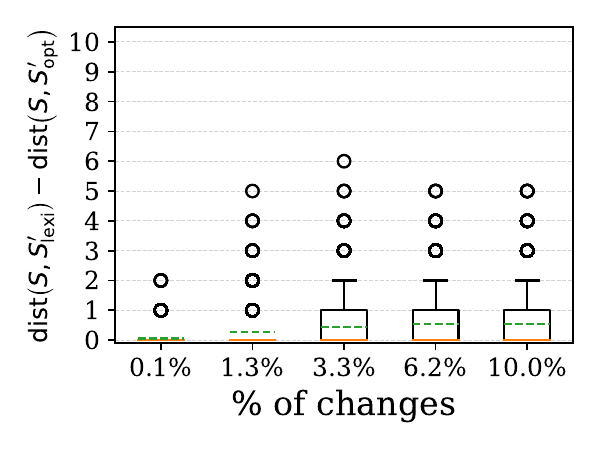}}}
\hspace{.1cm}
\subfloat[\centering \GreedyPAV, 1D, $\tau = 0.078$]{{\includegraphics[height=3.3cm]{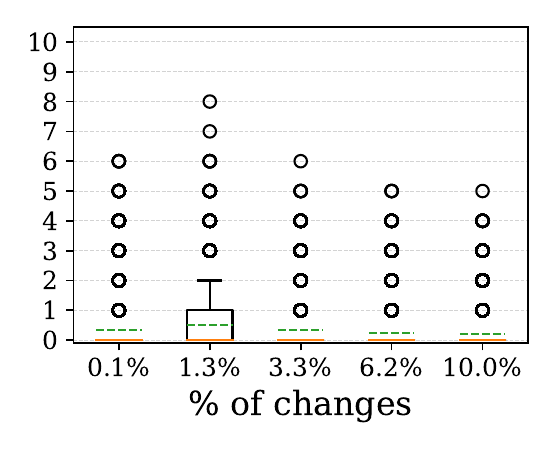}}}
\caption{Results of Experiment $2$ under the 1D model.}
\label{fig:Exp2_1Dallparams}
\end{figure}

\begin{figure}
\centering
\subfloat[\centering \GreedyCC, 2D, $\tau = 0.134$]{{\includegraphics[height=3.3cm]{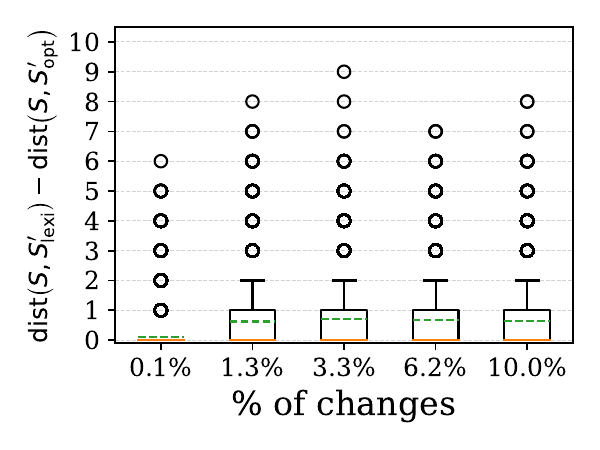}}}
\hspace{.1cm}
\subfloat[\centering \GreedyPAV, 2D, $\tau = 0.134$]{{\includegraphics[height=3.3cm]{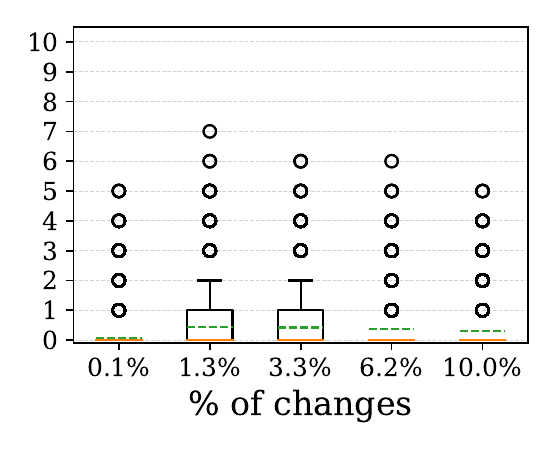}}}

\subfloat[\centering \GreedyCC, 2D, $\tau = 0.195$]{{\includegraphics[height=3.3cm]{figures/EXP2_seqcc_2D_0.195.pdf}}}
\hspace{.1cm}
\subfloat[\centering \GreedyPAV, 2D, $\tau = 0.195$]{{\includegraphics[height=3.3cm]{figures/EXP2_seqpav_2D_0.195.pdf}}}

\subfloat[\centering \GreedyCC, 2D, $\tau = 0.244$]{{\includegraphics[height=3.3cm]{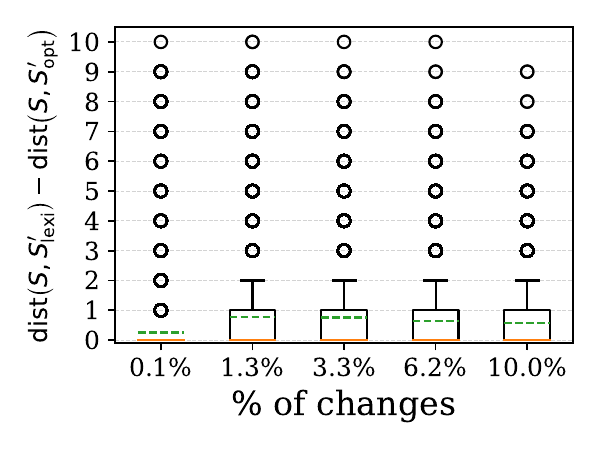}}}
\hspace{.1cm}
\subfloat[\centering \GreedyPAV, 2D, $\tau = 0.244$]{{\includegraphics[height=3.3cm]{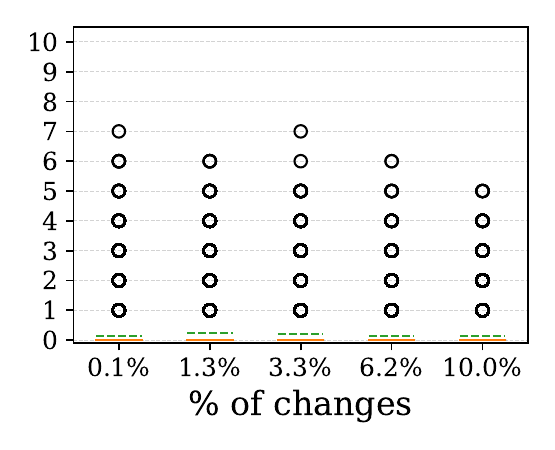}}}
\caption{Results of Experiment $2$ under the 2D model.}
\label{fig:Exp2_2Dallparams}
\end{figure}

\begin{figure}
\centering
\subfloat[\centering \GreedyCC, 1D+Res, $\tau = 0.025$]{{\includegraphics[height=3.3cm]{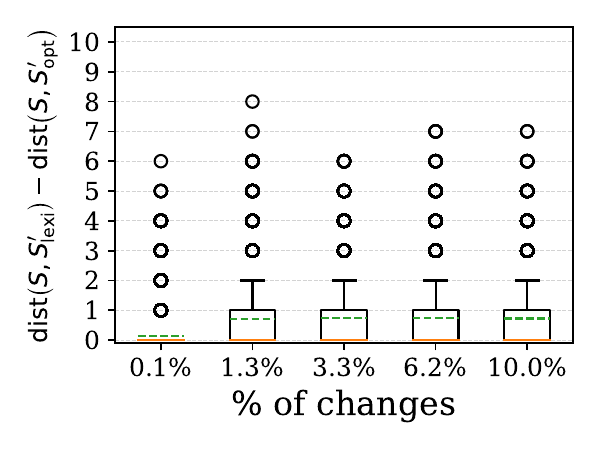}}}
\hspace{.1cm}
\subfloat[\centering \GreedyPAV, 1D+Res, $\tau = 0.025$]{{\includegraphics[height=3.3cm]{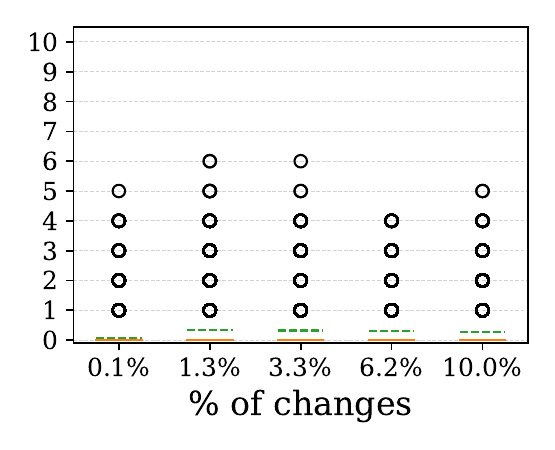}}}

\subfloat[\centering \GreedyCC, 1D+Res, $\tau = 0.051$]{{\includegraphics[height=3.3cm]{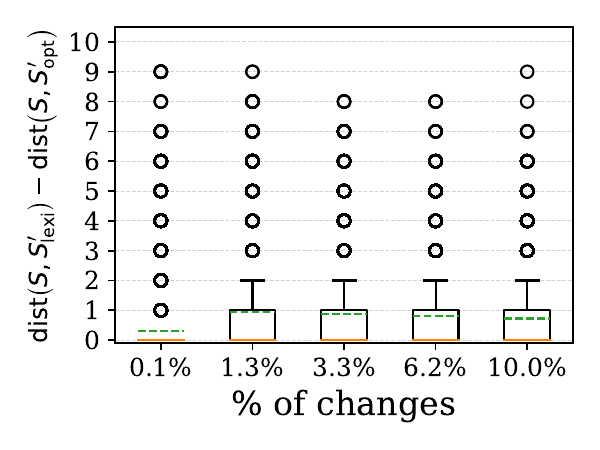}}}
\hspace{.1cm}
\subfloat[\centering \GreedyPAV, 1D+Res, $\tau = 0.051$]{{\includegraphics[height=3.3cm]{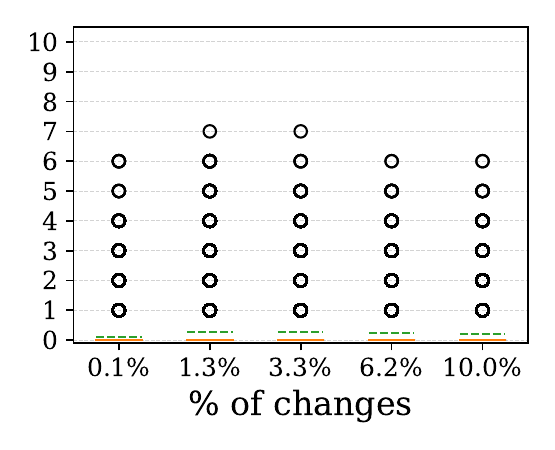}}}

\subfloat[\centering \GreedyCC, 1D+Res, $\tau = 0.078$]{{\includegraphics[height=3.3cm]{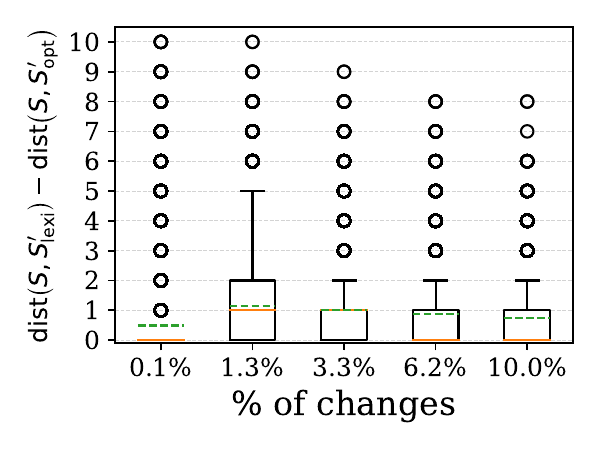}}}
\hspace{.1cm}
\subfloat[\centering \GreedyPAV, 1D+Res, $\tau = 0.078$]{{\includegraphics[height=3.3cm]{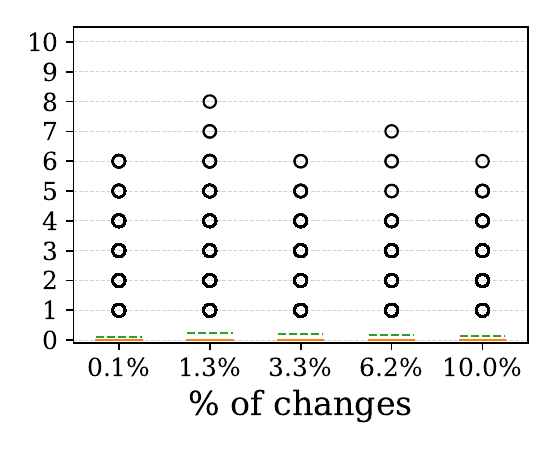}}}
\caption{Results of Experiment $2$ under the 1D+Res model.}
\label{fig:Exp2_1D+res}
\end{figure}

\begin{figure}
\centering
\subfloat[\centering \GreedyCC, 2D+Res, $\tau = 0.134$]{{\includegraphics[height=3.3cm]{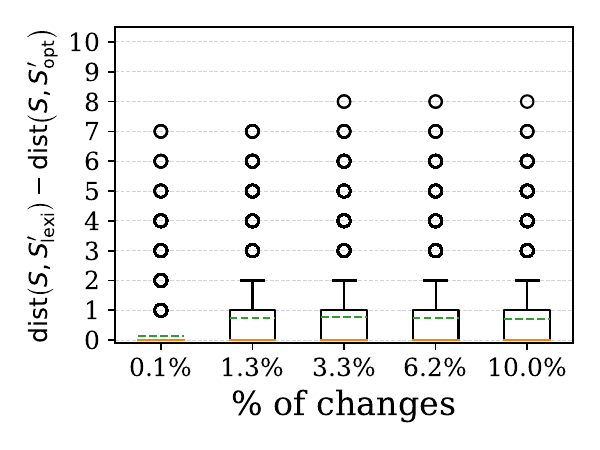}}}
\hspace{.1cm}
\subfloat[\centering \GreedyPAV, 2D+Res, $\tau = 0.134$]{{\includegraphics[height=3.3cm]{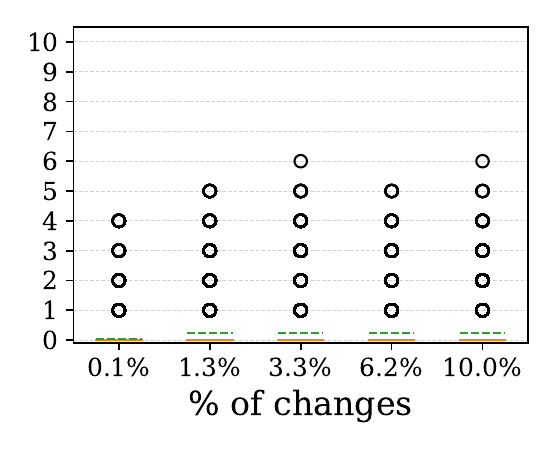}}}

\subfloat[\centering \GreedyCC, 2D+Res, $\tau = 0.195$]{{\includegraphics[height=3.3cm]{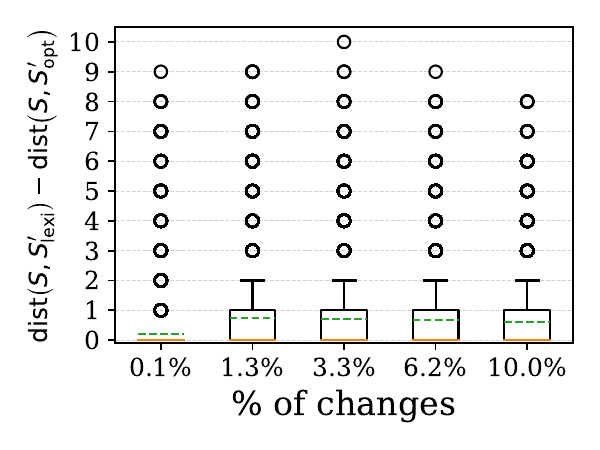}}}
\hspace{.1cm}
\subfloat[\centering \GreedyPAV, 2D+Res, $\tau = 0.195$]{{\includegraphics[height=3.3cm]{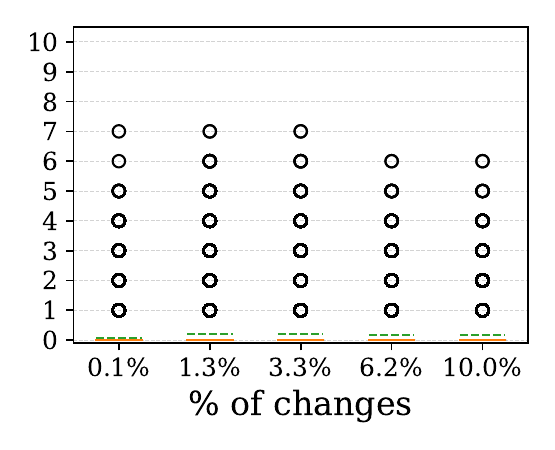}}}

\subfloat[\centering \GreedyCC, 2D+Res, $\tau = 0.244$]{{\includegraphics[height=3.3cm]{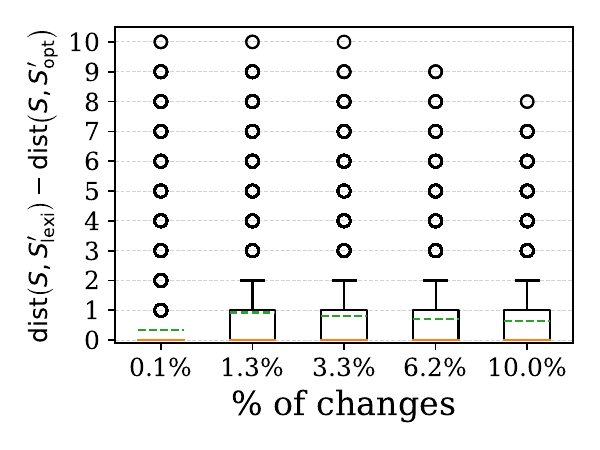}}}
\hspace{.1cm}
\subfloat[\centering \GreedyPAV, 2D+Res, $\tau = 0.244$]{{\includegraphics[height=3.3cm]{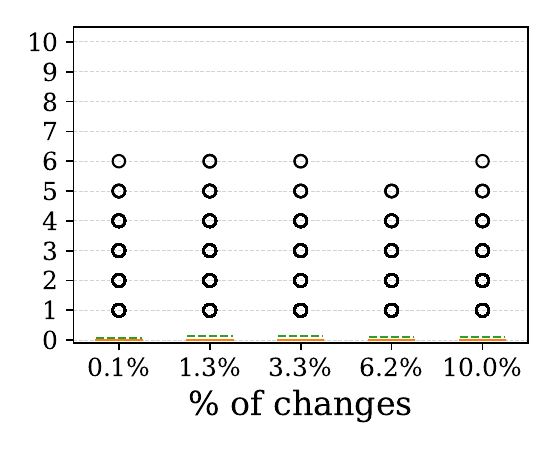}}}
\caption{Results of Experiment $2$ under the 2D+Res model.}
\label{fig:Exp2_2D+res}
\end{figure}

\begin{figure}
\centering
\subfloat[\centering \GreedyCC, Res, $\rho = 0.05, \phi = 0.75$]{{\includegraphics[height=3.3cm]{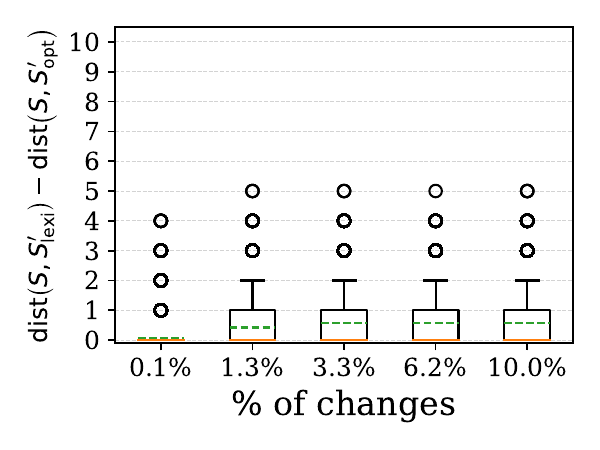}}}
\hspace{.1cm}
\subfloat[\centering \GreedyPAV, Res, $\rho = 0.05, \phi = 0.75$]{{\includegraphics[height=3.3cm]{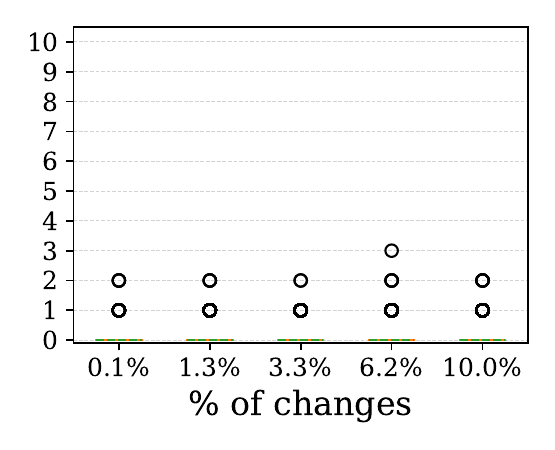}}}

\subfloat[\centering \GreedyCC, Res, $\rho = 0.1, \phi = 0.75$]{{\includegraphics[height=3.3cm]{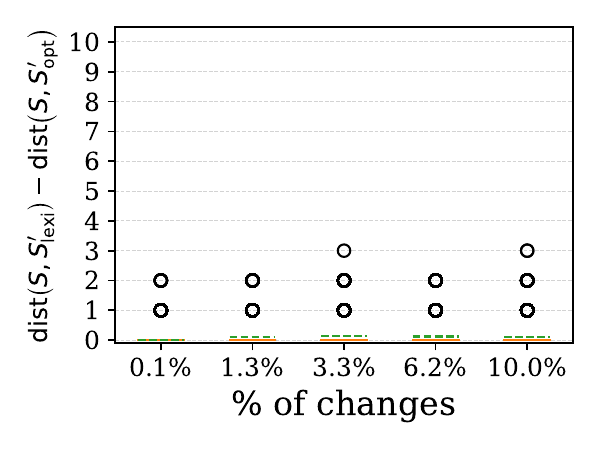}}}
\hspace{.1cm}
\subfloat[\centering \GreedyPAV, Res, $\rho = 0.1, \phi = 0.75$]{{\includegraphics[height=3.3cm]{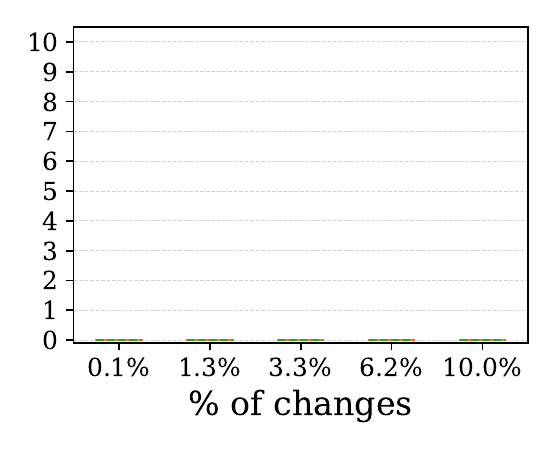}}}

\subfloat[\centering \GreedyCC, Res, $\rho = 0.15, \phi = 0.75$]{{\includegraphics[height=3.3cm]{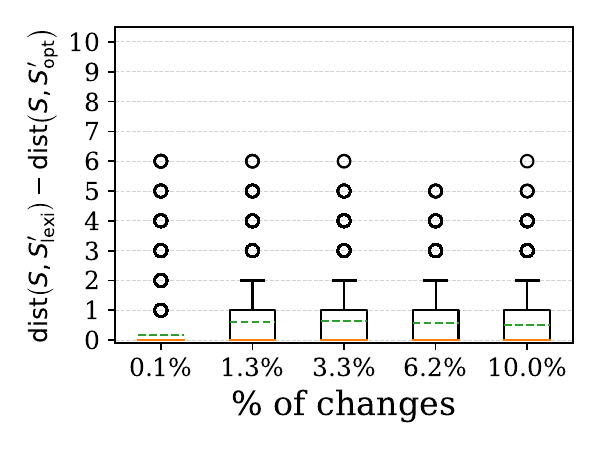}}}
\hspace{.1cm}
\subfloat[\centering \GreedyPAV, Res, $\rho = 0.15, \phi = 0.75$]{{\includegraphics[height=3.3cm]{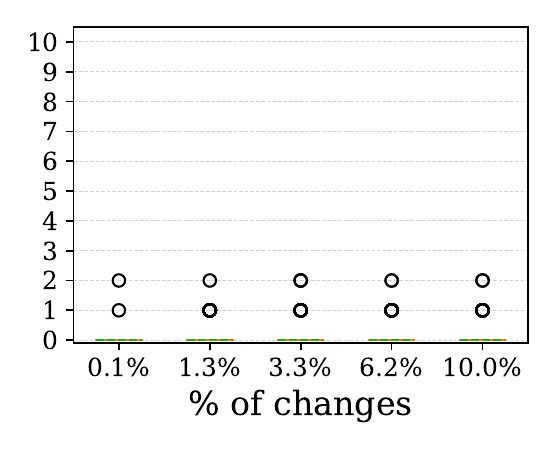}}}
\caption{Results of Experiment $2$ under the Res model.}
\label{fig:Exp2_Res}
\end{figure}

%%%%%%%%%%%%%%%%%%%%%%%%%%%%%%%%%%%%%%%%%%%%%%%%%%%%%

\begin{figure}
\centering
\subfloat[\centering \GreedyCC, 1D, $\tau = 0.025$]{{\includegraphics[height=3.3cm]{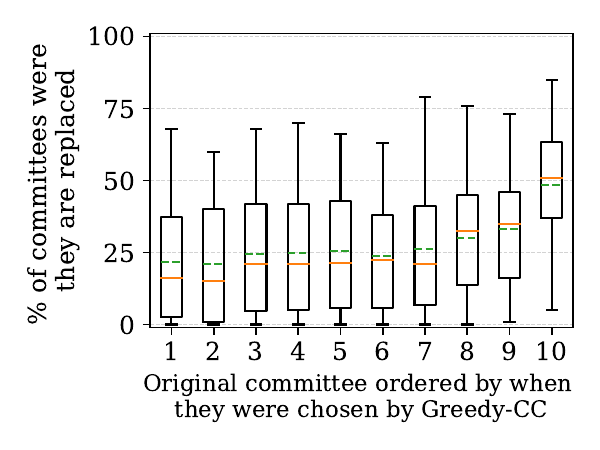}}}
\hspace{.1cm}
\subfloat[\centering \GreedyPAV, 1D, $\tau = 0.025$]{{\includegraphics[height=3.3cm]{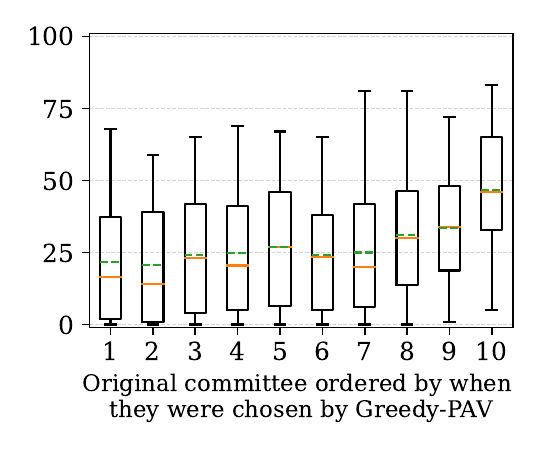}}}

\subfloat[\centering \GreedyCC, 1D, $\tau = 0.051$]{{\includegraphics[height=3.3cm]{figures/EXP3_seqcc_1D_0.051.pdf}}}
\hspace{.1cm}
\subfloat[\centering \GreedyPAV, 1D, $\tau = 0.051$]{{\includegraphics[height=3.3cm]{figures/EXP3_seqpav_1D_0.051.pdf}}}

\subfloat[\centering \GreedyCC, 1D, $\tau = 0.078$]{{\includegraphics[height=3.3cm]{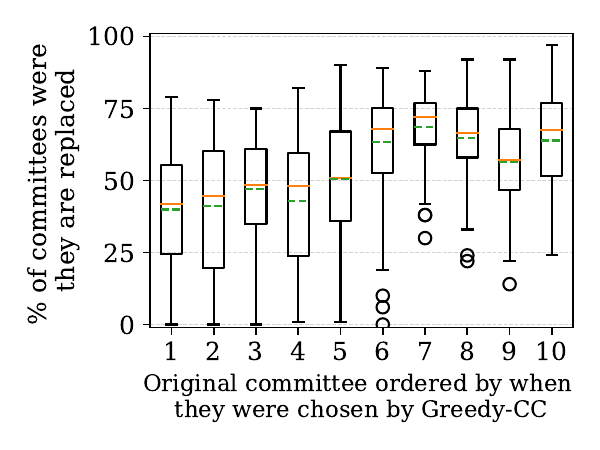}}}
\hspace{.1cm}
\subfloat[\centering \GreedyPAV, 1D, $\tau = 0.078$]{{\includegraphics[height=3.3cm]{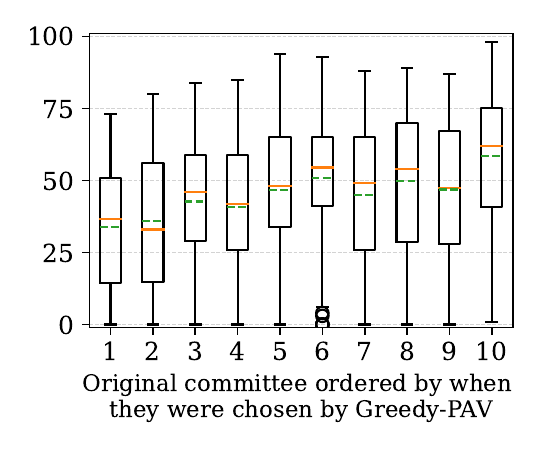}}}
\caption{Results of Experiment $3$ under the 1D model.}
\label{fig:Exp3_1Dallparams}
\end{figure}

\begin{figure}
\centering
\subfloat[\centering \GreedyCC, 2D, $\tau = 0.134$]{{\includegraphics[height=3.3cm]{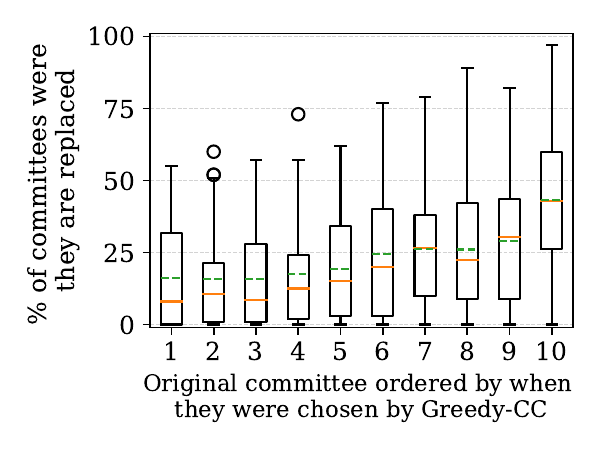}}}
\hspace{.1cm}
\subfloat[\centering \GreedyPAV, 2D, $\tau = 0.134$]{{\includegraphics[height=3.3cm]{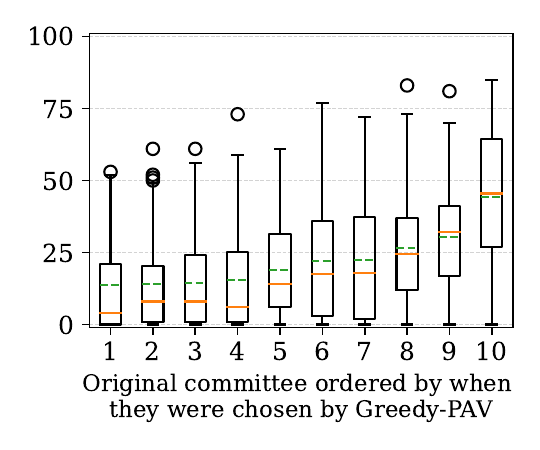}}}

\subfloat[\centering \GreedyCC, 2D, $\tau = 0.195$]{{\includegraphics[height=3.3cm]{figures/EXP3_seqcc_2D_0.195.pdf}}}
\hspace{.1cm}
\subfloat[\centering \GreedyPAV, 2D, $\tau = 0.195$]{{\includegraphics[height=3.3cm]{figures/EXP3_seqpav_2D_0.195.pdf}}}

\subfloat[\centering \GreedyCC, 2D, $\tau = 0.244$]{{\includegraphics[height=3.3cm]{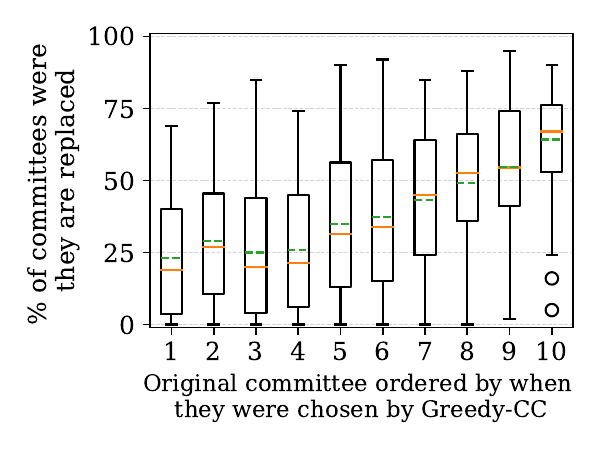}}}
\hspace{.1cm}
\subfloat[\centering \GreedyPAV, 2D, $\tau = 0.244$]{{\includegraphics[height=3.3cm]{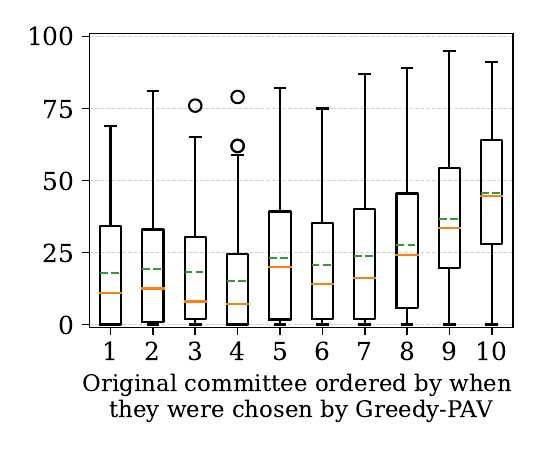}}}
\caption{Results of Experiment $3$ under the 2D model.}
\label{fig:Exp3_2Dallparams}
\end{figure}

\begin{figure}
\centering
\subfloat[\centering \GreedyCC, 1D+Res, $\tau = 0.025$]{{\includegraphics[height=3.3cm]{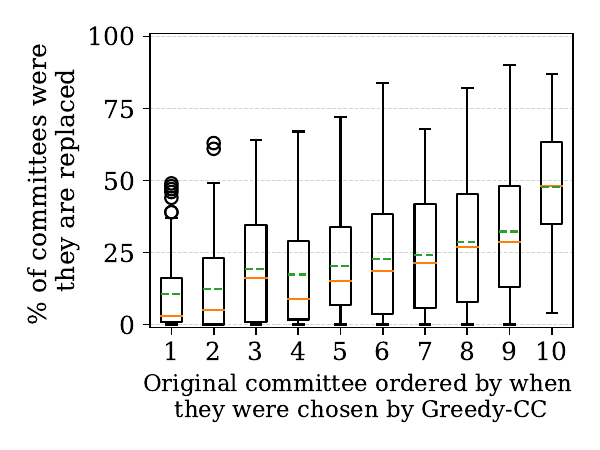}}}
\hspace{.1cm}
\subfloat[\centering \GreedyPAV, 1D+Res, $\tau = 0.025$]{{\includegraphics[height=3.3cm]{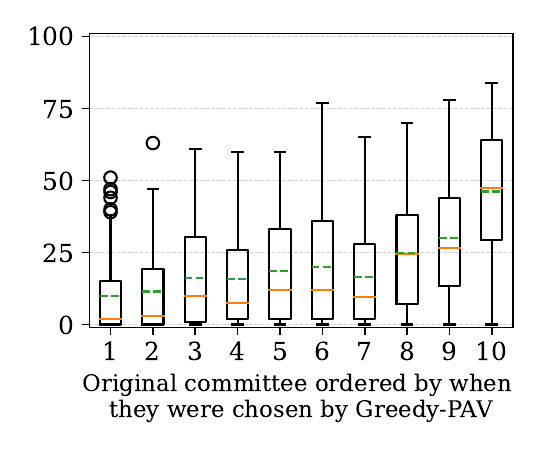}}}

\subfloat[\centering \GreedyCC, 1D+Res, $\tau = 0.051$]{{\includegraphics[height=3.3cm]{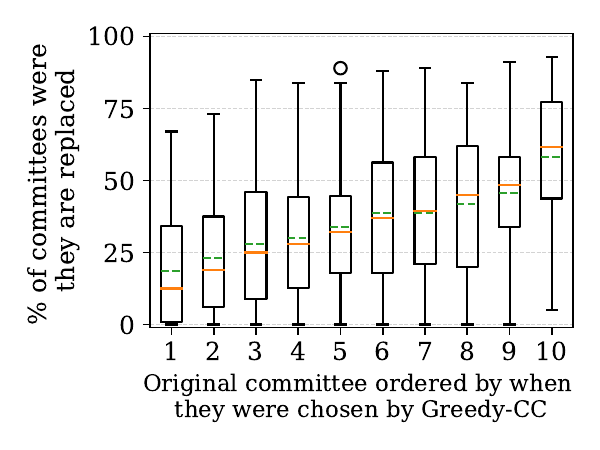}}}
\hspace{.1cm}
\subfloat[\centering \GreedyPAV, 1D+Res, $\tau = 0.051$]{{\includegraphics[height=3.3cm]{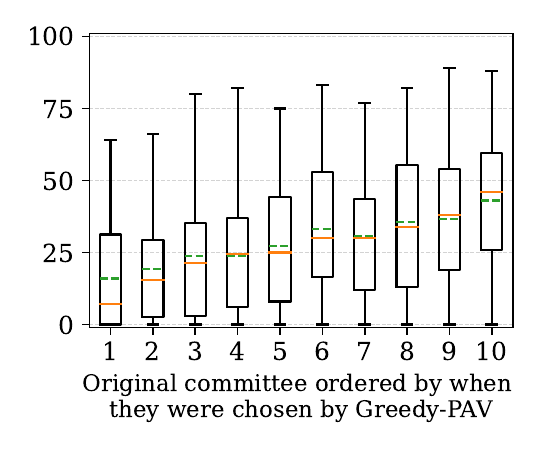}}}

\subfloat[\centering \GreedyCC, 1D+Res, $\tau = 0.078$]{{\includegraphics[height=3.3cm]{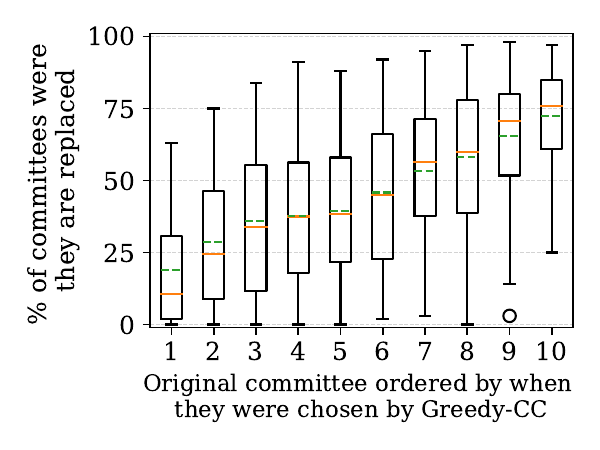}}}
\hspace{.1cm}
\subfloat[\centering \GreedyPAV, 1D+Res, $\tau = 0.078$]{{\includegraphics[height=3.3cm]{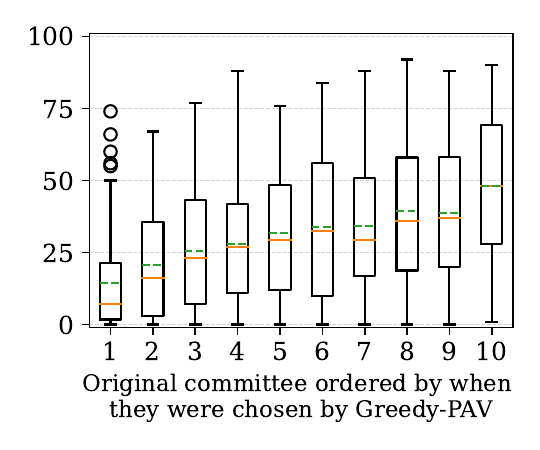}}}
\caption{Results of Experiment $3$ under the 1D+Res model.}
\label{fig:Exp3_1D+res}
\end{figure}

\begin{figure}
\centering
\subfloat[\centering \GreedyCC, 2D+Res, $\tau = 0.134$]{{\includegraphics[height=3.3cm]{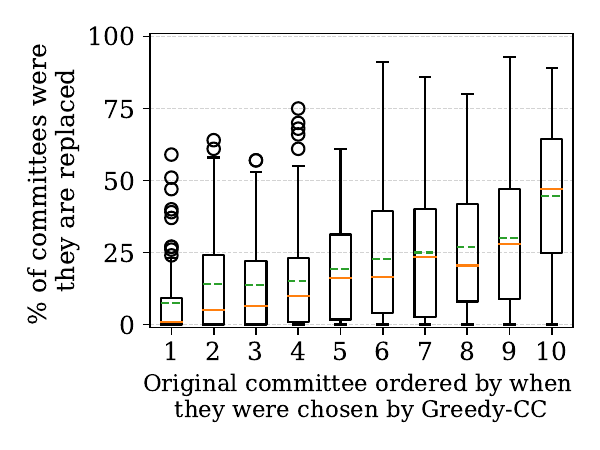}}}
\hspace{.1cm}
\subfloat[\centering \GreedyPAV, 2D+Res, $\tau = 0.134$]{{\includegraphics[height=3.3cm]{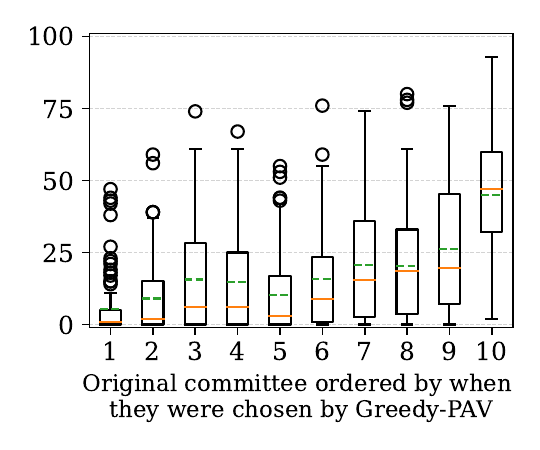}}}

\subfloat[\centering \GreedyCC, 2D+Res, $\tau = 0.195$]{{\includegraphics[height=3.3cm]{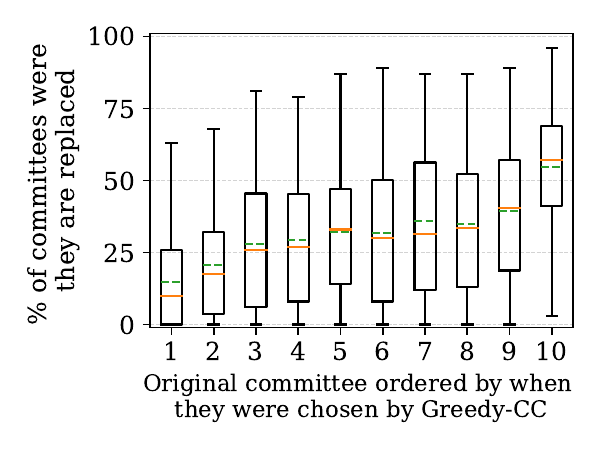}}}
\hspace{.1cm}
\subfloat[\centering \GreedyPAV, 2D+Res, $\tau = 0.195$]{{\includegraphics[height=3.3cm]{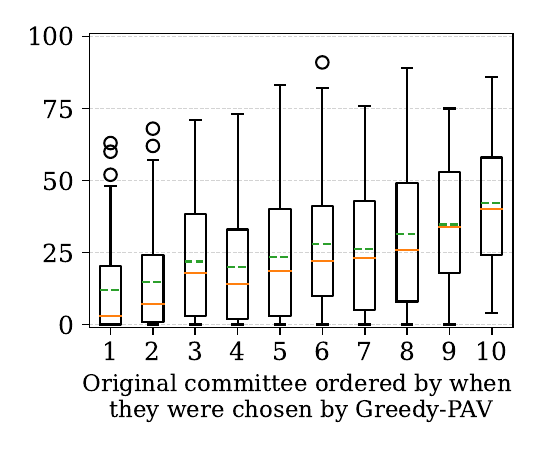}}}

\subfloat[\centering \GreedyCC, 2D+Res, $\tau = 0.244$]{{\includegraphics[height=3.3cm]{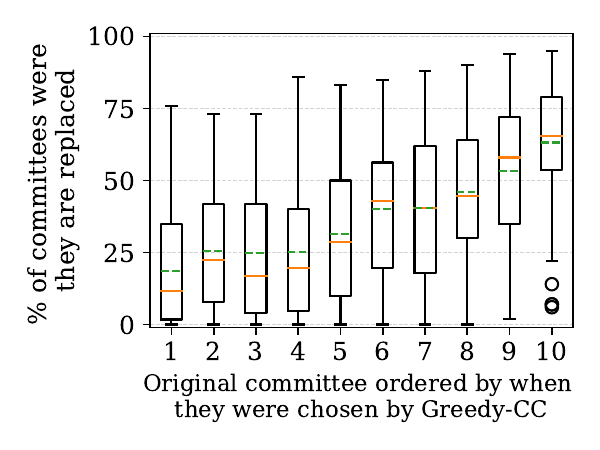}}}
\hspace{.1cm}
\subfloat[\centering \GreedyPAV, 2D+Res, $\tau = 0.244$]{{\includegraphics[height=3.3cm]{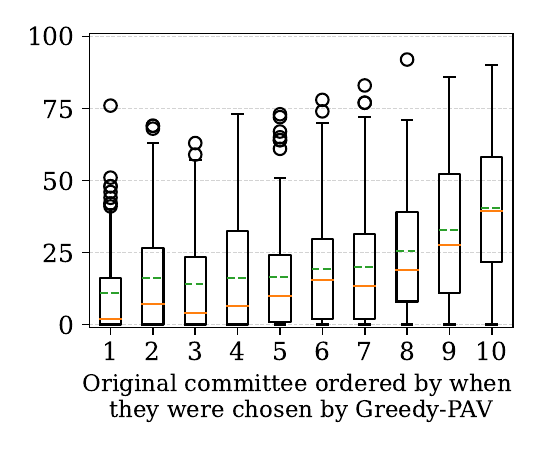}}}
\caption{Results of Experiment $3$ under the 2D+Res model.}
\label{fig:Exp3_2D+res}
\end{figure}

\begin{figure}
\centering
\subfloat[\centering \GreedyCC, Res, $\rho = 0.05, \phi = 0.75$]{{\includegraphics[height=3.3cm]{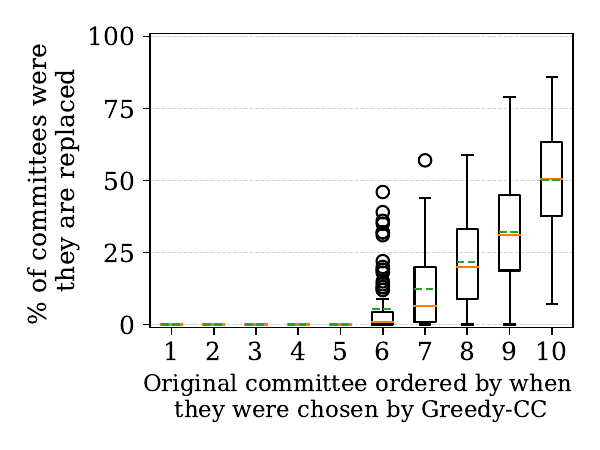}}}
\hspace{.1cm}
\subfloat[\centering \GreedyPAV, Res, $\rho = 0.05, \phi = 0.75$]{{\includegraphics[height=3.3cm]{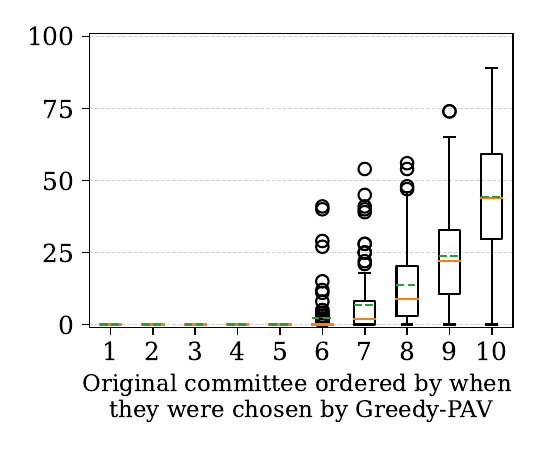}}}

\subfloat[\centering \GreedyCC, Res, $\rho = 0.1, \phi = 0.75$]{{\includegraphics[height=3.3cm]{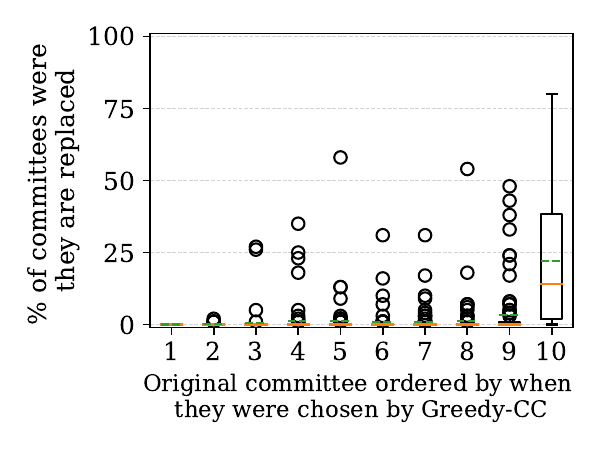}}}
\hspace{.1cm}
\subfloat[\centering \GreedyPAV, Res, $\rho = 0.1, \phi = 0.75$]{{\includegraphics[height=3.3cm]{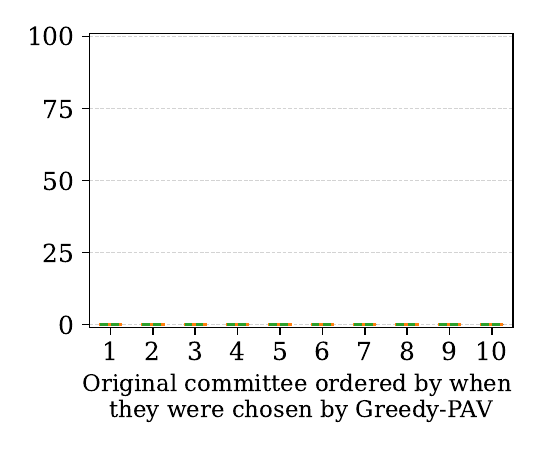}}}

\subfloat[\centering \GreedyCC, Res, $\rho = 0.15, \phi = 0.75$]{{\includegraphics[height=3.3cm]{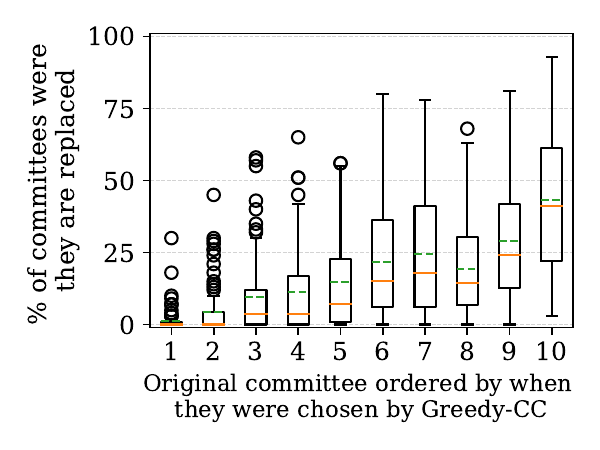}}}
\hspace{.1cm}
\subfloat[\centering \GreedyPAV, Res, $\rho = 0.15, \phi = 0.75$]{{\includegraphics[height=3.3cm]{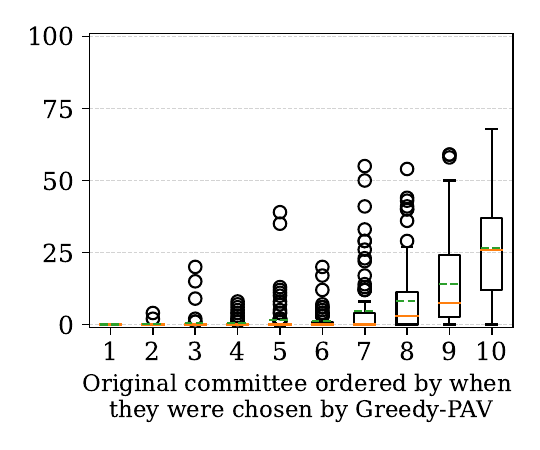}}}
\caption{Results of Experiment $3$ under the Res model.}
\label{fig:Exp3_Res}
\end{figure}

%\fi

\end{document}